\documentclass[final,1p,times]{elsarticle}
\usepackage{amssymb}
\usepackage{amsmath}

\usepackage{graphicx}
\usepackage{pgfplots}
\usepackage{hyperref}
\usepackage{caption}
\usepackage{longtable}
\usepackage{longtable}
\usepackage{mathtools}

\usepackage{float}
\usepackage[caption = true]{subfig}

\usepackage[ruled,vlined,algosection]{algorithm2e}
\usepackage{multirow}

\usepackage{algorithmic}

\newtheorem{theorem}{Theorem}

\newtheorem {lemma}{Lemma}

\newtheorem{defnn}{Definition}
\newtheorem {corollaryy}{Corollary}
\newtheorem{example}{Example}
\newenvironment{proof}{\noindent {\bf Proof :\ } }{$\Box$ }

\newenvironment{corollary}{\begin{corollaryy}{\bf :}\sl}{\end{corollaryy}}
\newtheorem {conjecture}{Conjecture}
\newtheorem {oproblem}{Open Problem}

\def\boxit#1{%
  \smash{\fboxsep=0pt\llap{\rlap{\fbox{\strut\makebox[#1]{}}}~}}\ignorespaces
}

\begin{document}

\begin{frontmatter}



\title{Maximal Length Cellular Automata : A Survey} %% Article title


\author[1]{Sumit Adak}\corref{cor1}%
\fnref{fn1}
\ead{suad@dtu.dk}
\affiliation[1]{organization={Technical University of Denmark},
            addressline={Kgs. Lyngby},
            postcode={2800},
            country={Denmark}}           
\author[2]{Sukanta Das}
\ead{sukanta@it.iiests.ac.in}
\affiliation[2]{organization={Indian Institute of Engineering Science and Technology, Shibpur},
            addressline={Howrah},
            postcode={711103},
            country={India}}

\cortext[cor1]{Corresponding author}
\fntext[fn1]{The majority of the work was completed when the author was in Indian Institute of Engineering Science and Technology, Shibpur} 

\begin{abstract}
This article surveys some theoretical aspects of Cellular Automata (CAs) research. In particular, we discuss on {\em maximal length} CA. An $n$-cell CA is a maximal length CA, if all the configurations except one form a single cycle. There is a bonding between maximal length CA and primitive polynomial. So, primitive polynomials occupy a good amount of space in this survey. The main goal of this survey is to provide a tutorial on maximal length CA theory to researchers with classical and new results on {\em maximality}. We also give a compact collection of known results with references to their proofs, and to suggest some open problems. Additionally, some new theorems and corollaries are added to bridge the gaps among several known results.
\end{abstract}


% \begin{highlights}
% \end{highlights}

%% Keywords

\end{frontmatter}

\section{Introduction}
\label{section:intro}

This article presents a tutorial on a special class of non-uniform cellular automata (CAs); for a general survey on CAs, please see \cite{kamalikasurvey}. Non-uniform (hybrid) finite CAs under null boundary condition share a very special property which the classical finite (non-uniform) CAs under periodic boundary condition do not. The property is the {\em maximality} in cycle length, which means, the presence of a cycle of length $2^n-1$ in an $n$-cell binary cellular automaton (CA). This property was first observed by Pries et al.~\cite{Pries86}. These CAs are traditionally named as {\em maximal length} CAs. There is a connection between primitive polynomials over {\em Galois Field} (GF) and maximal length CAs. The main focus of this work on linear maximal length CA and primitive polynomial. In our further reference, by ``CA'' and ``maximal length CA'', we shall mean ``linear CA'' and ``linear maximal length CA'', if not mentioned otherwise.

Although from 1980s, maximal length CA is being observed, they were initially searched by experiments. In early phase, Hortensius et al.~\cite{Horte89a,Horte89b,Horte89c} used maximal length CAs to generate pseudo-random number. Using computer simulation, they proposed a table which gives the non-uniform constructions necessary to get a CA with maximal length cycle. Serra et al.~\cite{Serra90c} showed that maximal length CAs with all non-zero configurations lying in a single cycle produce high quality pseudo-random patterns. In \cite{CattellM96, Cattell2}, Cattell and Muzio developed an adequate algebraic framework for the study of maximal length CAs. It was also shown that deciding primitive polynomials and deciding maximal length CAs are two equivalent problems. 

In this tutorial article, we explore the maximal length CAs under three different categories of cellular automata - linear, complemented and non-linear. At first, we define cellular automata and maximal length CAs. Next, we establish the relation between linear maximal length CA and primitive polynomial (Section~\ref{subsection:relation}). Next, we show the different types of approaches which are applied to analyze maximal length CAs (Section~\ref{section:analysis}). We explore various types research on maximal length CA which are done from beginning and the present situation for some basic conditions to get maximality, such as dependence on the boundary conditions. Synthesis of maximal length CA from a given primitive polynomial over GF(2) is discussed in next section (Section~\ref{section:synthesis}). This work was most attractive work for many researchers till date. Next, we synthesize the primitive polynomials using cellular automata as a tool (Section \ref{section:synthesis-PP}). Different types of greedy strategies are discussed here. All the prior topics are based on linear maximal length CA, but in next two, we show the complemented and non-linear maximal length CAs. That means, generation of complemented and non-linear maximal length CAs (Section~\ref{section:complemented} and~\ref{section:nonlinear}). All the previous discussions are based on GF(2), next we extend it over GF($q$) to finding maximal length CAs (Section~\ref{section:MLCAGF(q)}). Besides, maximal length CA have been used for diverse applications, but most attractive applications are pseudo-random number generator (PRNG) and cryptography (Section~\ref{section:applications}). Several open problems are discussed at the end of most of the sections. In discussion, we conclude with an important open problem based on maximal length CA.

\section{One-dimensional Hybrid Cellular Automata}
\label{section:preliminaries}

\subsection{Basics}
\label{subsection:CA}

A cellular automaton (CA) that we are considering here consists of an array of $n$ cells numbered from 0 to $n-1$. Cells can assume either 0 or 1 as their state. Let $x_i$ denote the state of cell $i$. Then, a configuration of the CA is $x=(x_0x_1\cdots x_{n-1})$ where $x_i \in \{0,1\}$. Let us consider that $\mathcal{C}$ is the collection of all possible configurations of an $n$-cell CA (that is $|\mathcal{C}|$=$2^n$). Then, a CA acn be seen as a function $G$: $\mathcal{C}$ $\rightarrow$ $\mathcal{C}$, which satisfies the following conditions: $y=G(x)$, $x,y \in {\mathcal{C}}$, where $x=(x_i)_{0\leq i\leq {n-1}}$, $y=(y_i)_{0\leq i\leq {n-1}}$ and $y_i={\mathcal{R}_i}(x_{i-1},x_i,x_{i+1})$. The ${\mathcal{R}_i}:\{0,1\}^3 \mapsto \{0, 1\}$ is a next state function for the cell $i$, commonly known as {\em rule}. Here, $y$ is called the {\em successor} of $x$, and $x$ the {\em predecessor} of $y$. In this work, we consider null boundary condition where left and right neighbors of cell 0 and cell $n-1$ are always in state 0. That is, $y_0=\mathcal{R}_0(0,x_0,x_1)$ and $y_{n-1}=\mathcal{R}_{n-1}(x_{n-2},x_{n-1},0)$. 

\begin{table}\small
	\begin{center}	
		\caption{The rules 90 and 150}	
		\label{table:rules90150}
		\begin{tabular}{cccccccccc}\hline
		Present~state &  111 & 110 & 101 & 100 & 011 &  010 &  001 &  000 & Rule \\
		(RMT)& (7) & (6) & (5) & (4) & (3) & (2) & (1) & (0) &  \\\hline
{\rm (i)~Next~state}    &   0 &  1  &  0  &  1  &   1  &   0  &   1  & 0 & 90 \\
{\rm (ii)~Next~state}    &   1  &  0 &  0  &  1  &   0  &   1  &   1  &   0  &  150 \\\hline
\end{tabular}
	\end{center}
\end{table}

The rule $\mathcal{R}_i$ can be expressed in tabular form (Table~\ref{table:rules90150}) and decimal equivalent of the eight-bit binary sequence generally identifies the rule. Obviously, there are $2^8$ distinct rules. If all the cells of a CA use a single rule, the CA is called {\em uniform} CA; otherwise it is a {\em non-uniform} or {\em hybrid} CA \cite{SukantaTh,Adak-IS-2021}. To define an arbitrary CA, therefore, we need a {\em rule vector} ${\mathcal {R}} = ({\mathcal R_0}, {\mathcal R_1}, \cdots, {\mathcal R_i}, \cdots, {\mathcal R_{n-1}})$, where ${\mathcal R_i}$ is the rule of cell $i$. If not specified otherwise, by ``CA'' we shall mean hereafter one-dimensional hybrid CA under null boundary condition.

For a particular rule ${\mathcal R}_i$,  let ${\mathcal R_i[x_{i-1}x_{i}x_{i+1}]}$ denote the next state of cell $i$ for the present states combination $x_{i-1}x_{i}x_{i+1}$ of cell $i$ and its neighbors. Each $x_{i-1}x_{i}x_{i+1}$ is called as Rule Min Term (RMT). RMTs are identified by their decimal equivalents. For example, 010 of the first row of Table~\ref{table:rules90150} is the RMT 2, next state against which is 0 for rule 90, 1 for rule 150. If $r$ is an RMT of ${\mathcal R_i}$, we write ${\mathcal R_i}[r]$ to denote its next state. Hence, 90[2]=0, 150[2]=1 (see Table~\ref{table:rules90150}).

Sometimes, a configuration can alternatively be represented by its {\em RMT Sequence (RS)}.
\begin{defnn}
Let $x=(x_i)_{0\le i\le n-1}$ be a configuration of an $n$-cell CA. The  RMT sequence of $x$, denoted as $\tilde{x}$, is $(r_i)_{0\le i\le n-1}$ where $r_i$ is the RMT $x_{i-1}x_ix_{i+1}$.
\end{defnn}
For example, (0, 1, 3, 7, 7, 6) or simply, 013776 is the RMT sequence of the configuration 001111. An interesting relation is followed in the sequence: $r_{i+1}$ is either $2r_i$ or $2r_i+1$ under modulo 8 operation. Obviously, an arbitrary sequence of RMTs cannot form an RMT sequence.

\begin{defnn}
Two RMTs $r$ and $s$ ($r\ne s$) are said to be equivalent to each other if $2r\equiv 2s\pmod8$ and said to be sibling to each other if $\lfloor \frac{r}{2} \rfloor$=$\lfloor \frac{s}{2} \rfloor$ (\cite{SukantaTh}). 
\end{defnn}
Therefore, RMTs 0 and 4 are equivalent to each other, and RMTs 0 and 1 are sibling to each other. That is, the RMTs $0xy$ and $1xy$, and the RMTs $xy0$ and $xy1$ form equivalent and sibling RMT sets respectively to any $x,y\in\{0,1\}$.

The 256 rules can be divided into three categories - {\em linear}, {\em complemented} and {\em non-linear} rules \cite{ppc1}. Following is the basic definition of these categories.
\begin{defnn}
\label{definition:linearrules}
A rule ${R}:\{0,1\}^3\rightarrow \{0,1\}$ is {\bf linear} if $R(x,y,z)=ax+dy+bz\pmod2$ for some constants $a,d,b\in \{0,1\}$. A rule $\overline{R}:\{0,1\}^3\rightarrow \{0,1\}$ is {\bf complemented} if $\overline{R}(x,y,z)=1-{R}(x,y,z)$ for some linear rule $R$. We call a rule \textbf{non-linear} if it is neither linear nor complemented.  
\end{defnn}
There are eight linear rules - 0, 60, 90, 102, 150, 170, 204 and 240, and eight complemented rules --15, 51, 85, 105, 153, 165, 195 and 255 (see Table~\ref{table:linearrules}). From Table~\ref{table:linearrules}, we get that $a=1$, $d=1$ and $b=0$ for rule 60, whereas for rule 90, $a=1$, $d=0$ and $b=1$. Sometime, a linear or complemented rule can be written as - ($a,d,b$). As an example, rule 90 can be written as (1,0,1). Apart from the 16 rules noted in Table~\ref{table:linearrules}, the rest 240 rules are non-linear.

\begin{table}[h]\tiny
	\setlength{\tabcolsep}{5.5pt}
	\begin{center}
		\caption{Linear and complemented CA rules}	
		\label{table:linearrules}
		\resizebox{0.7\textwidth}{!}{
		\begin{tabular}{c||c}\hline
Linear Rules & Complemented Rules \\\hline
${0}(x,y,z)=0$ & ${255}(x,y,z)=1$\\
${60}(x,y,z)=x+y\pmod{2}$ & ${195}(x,y,z)=1-{60}(x,y,z)$\\
${90}(x,y,z)=x+z\pmod{2}$ & ${165}(x,y,z)=1-{90}(x,y,z)$\\
${102}(x,y,z)=y+z\pmod{2}$ & ${153}(x,y,z)=1-{102}(x,y,z)$\\
${150}(x,y,z)=x+y+z\pmod{2}$ & ${105}(x,y,z)=1-{150}(x,y,z)$\\
${170}(x,y,z)=z$ & ${85}(x,y,z)=1-{170}(x,y,z)$\\
${204}(x,y,z)=y$ & ${51}(x,y,z)=1-{204}(x,y,z)$\\
${240}(x,y,z)=x$ & ${15}(x,y,z)=1-{240}(x,y,z)$\\ \hline
\end{tabular}
		}
	\end{center}
	%\vspace{-1.0em}
\end{table}

\begin{defnn}
\label{definition:diffCA}
If all the rules of a rule vector ${\mathcal R}$ are linear, then the CA is {\bf linear}. If at least one rule is complemented and the rest are linear, then the CA is a {\bf complemented CA}. Otherwise, (that is, at least one rule of ${\mathcal R}$ is non-linear) the CA is {\bf non-linear}.
\end{defnn}
For example, the CA with rule vector (90, 170, 102, 60) is linear whereas the CA (90, 170, 153, 60) is a complemented CA because rule 153 is a complemented rule. The CA (90, 170, 86, 60) is a non-linear CA because 86 is a non-linear rule.

The sequence of configurations generated during their evolution (with time) directs the CA behaviour. If CAs are finite, the sequence of configurations can be represented by {\em transition diagram}. Figure~\ref{figure:st-150-90-90-90} shows such a diagram of 4-cell CA $(150, 90, 90, 90)$. The transition diagram of a CA may contain {\em cyclic} and {\em acyclic} configurations (see Definition~\ref{definition:cyclelength}). Based on this, a CA can be categorized as either {\em reversible} and {\em irreversible}.

\begin{figure}
	\centering
	\includegraphics[width= 4.2in, height = 1.2in]{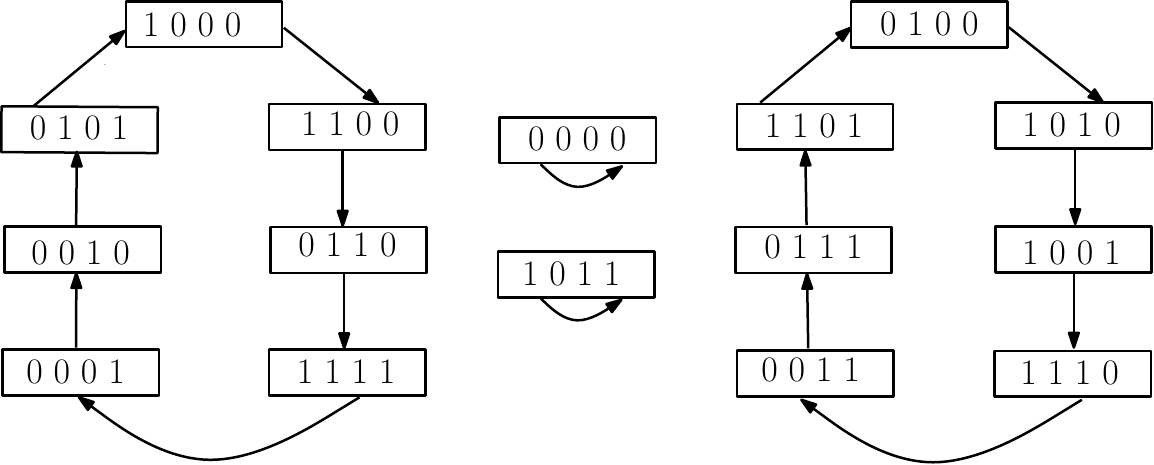}
	\caption{Transition diagram of CA $(150, 90, 90, 90)$}
	\vspace{-1.0em}  
	\label{figure:st-150-90-90-90}
\end{figure}

\begin{defnn}
\label{definition:cyclelength}
A configuration $x\in \{0,1\}^n$ is said to be cyclic if $x=G^t(x)$ for some $t\in\{1, 2, \cdots, 2^n-1\}$. Otherwise the configuration is acyclic. The CA has a cycle of length $t$ if $x=G^t(x)$ but $x\ne G^{t_1}(x)$ for any $t_1< t$.
\end{defnn}

\begin{defnn}
\label{sefinition:reversible}
A CA is reversible if all the configurations are cyclic; otherwise the CA is irreversible.
\end{defnn}

In other word, a CA is reversible if its (global transition) function $G$ is bijective. Deciding an arbitrary CA as reversible is an issue. This issue has been addressed in \cite{SukantaTh} for an $n$-cell non-uniform CA under null boundary condition. In a reversible CA, the initial configuration repeats after certain number of time steps and each configuration has exactly one predecessor (Figure~\ref{figure:st-150-90-90-90}). On the other hand, in an irreversible CA, there are some configurations which are having more than one predecessor.

\subsection{Maximal Length Cellular Automata}
\label{section:MLCA}

Reversible CA can be classified as maximal length and non-maximal length CA: 

%The maximal length CA is the special class of reversible CA having a cycle of length $2^n-1$ where $n$ is the number of cells in the CA. Following is the basic definition of maximal length CA.

\begin{defnn}
\label{definition:MLCA}
A finite CA is a maximal length CA if all but one configuration of it are in a single cycle.
%An $n$-cell CA is maximal length if the CA consists of two cycles, one is length 1, and the other is of length $2^n-1$.
\end{defnn}

%Furthermore, we can divide maximal length CAs in three categories - {\em linear}, {\em complemented} and {\em non-linear} maximal length CA. Linear and complemented CA can be characterized by matrix algebra where non-linear CA cannot. Further, in Section~\ref{section:complemented}, we discuss about complemented maximal length CA and in Section~\ref{section:nonlinear} we discuss about non-linearity in  maximal length CA. Rest of the sections are dependent on linear maximal length CA. Therefore, further by ``maximal length CA'' we shall mean ``linear maximal length CA''.

%For linear maximal length CA the length 1 cycle is always formed by the configuration $0^n$. Maximal length CAs are always reversible and length of cycle formed by a binary maximal length CA is $2^n-1$. The CA $(150, 150, 90, 150)$ is a maximal length CA, which is having a cycle of length $2^4-1 = 15$ (see Figure~\ref{figure:st-150-150-90-150}). Whether the CA $(90, 150, 90, 90)$ is not maximal length (see Figure~\ref{figure:st-90-150-90-90}). 
%Sometimes, the local rule of a CA is a linear function, which can be represented as below.

The researches on maximal length CA are primarily based on binary CAs that use Wolfram's rules. So, major share of this survey is occupied by this class of CAs. For binary maximal length CA, length of maximal cycle is $2^{n}-1$ where $n$ is number of cells of the CA. In our further discussion, if not stated otherwise, by maximal length CA we shall mean this class of CAs. In a linear maximal length CA of size $n$, the marginal configuration is $0^n$ which forms a single length cycle (fixed point), and the rest $2^n-1$ (non-zero) configurations form another cycle. For example, the CA $(150, 150, 90, 150)$ is a maximal length CA, which is having a cycle of length $2^4-1 = 15$ (see Figure~\ref{figure:st-150-150-90-150}). In case of complemented CAs and some of the non-linear maximal length CAs, the marginal configuration is some non-zero configuration. In literature, however, most of works with maximality are centred around linear CAs. In our further reference, if not stated otherwise, by ``maximal length CA'', we shall mean only linear maximal length CA. 

\begin{figure}
	\centering
	\includegraphics[width= 4.2in, height = 1.2in]{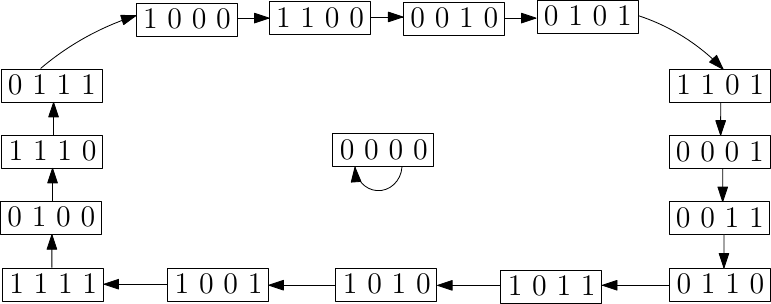}
	\caption{Transition diagram of CA $(150, 150, 90, 150)$}
	\vspace{-1.0em}  
	\label{figure:st-150-150-90-150}
\end{figure}

\subsubsection{Matrix Algebra : Characterization Tool}
\label{subsection:matrix}

Out of total 256 CA rules, 8 are linear rules (see Table~\ref{table:linearrules}). Linear CAs are formed with these linear rules only. The linear CAs can be efficiently characterized by matrix algebra. An $n$-cell 1-dimensional linear CA is represented by a characterization matrix $[T]_{n\times n}$. The $i^{th}$ row of $T$ corresponds to the neighborhood relation of the $i^{th}$ cell, where
%These rules has a special interest to the CA researchers as it can be characterized by matrix algebra tools. The matrix algebraic tools are also used to represent a linear CA that consider different rules for different CA cells (non-uniform). In this work, we concentrate on characterizing such a non-uniform CA. A brief overview of this model is next outlined. 
%
%An $n$-cell 1-dimensional linear CA is characterized by a linear operator $[T]_{n\times n}$ matrix. $T$ is the characterized matrix of the CAs. The $i^{th}$ row of $T$ corresponds to the neighborhood relation of the $i^{th}$ cell, where
\[
T[i,j]=\left\{ 
\begin{array}{cl}
1, & \mbox{if the next state of cell $i$ depends on the present state of cell $j$}\\
0, & \mbox{otherwise.}\\
\end{array} 
\right.
\]
Since our CA uses 3-neighborhood dependency, $T[i,j]$ can have non-zero values for $j$ = $(i-1)$, $i$, $(i+1)$. That is, the $T$-matrix is a tridiagonal matrix where the main, upper and lower diagonal can only be non zero. If $c_t$ is the configuration of the CA at $t^{th}$ instant of time, then the next configuration - that is, the configuration at $(t+1)^{th}$ time instant, is $c_{t+1} = T.c_t$ where $c_t$ and $c_{t+1}$ are presented as column vectors. In our further discussion, when we multiply $T$ with a configuration, we assume that the configuration is a column vector. However, if $c_{t+p}$ is reachable from $c_t$ after $p$ time steps, that is, if $c_{t+p}=G^p(c_t)$, then we get that  
\begin{equation}
\label{equation:nextstatelinear}
c_{t+p} = T^p.c_t
\end{equation}
Reversibility of a CA can easily be decided from its matrix representation. If $T$ is non-singular, the CA is reversible; otherwise it is irreversible.
\[
\det(T) =\left\{ 
\begin{array}{cl}
1, & \mbox{if the CA is reversible}\\
0, & \mbox{if the CA is irreversible}\\
\end{array} 
\right.
\]
where $\det(T)$ is determinant of matrix $T$. Under null boundary condition, $T$ can be represented as the following: 
%As per above discussion, an $n$-cell 3-neighborhood linear CA can also be expressed by a square matrix ($T$) of dimension $n$. The matrix shows the dependency of a cell on other cells. The matrix is tridiagonal, where the lower, upper and main diagonals depend on the rules $\mathcal{R}_i$ in the form of ($a_i, d_i, b_i$) where $0\leq i\leq n-1$. Under null boundary condition, it can be represented as the following: 
\[
T=
    \begin{bmatrix}
        d_0 & b_0 & 0 & \cdots & 0 & 0 \\
        a_1 & d_1 & b_1 & \ddots &  & 0 \\
        0 & a_2 & d_2 & \ddots &  & \vdots \\
        \vdots & \ddots & \ddots & \ddots & \ddots & \vdots \\
        0 &  &  & \ddots & d_{n-2} & b_{n-2} \\
        0 & 0 & \cdots & \cdots & a_{n-1} & d_{n-1} \\

    \end{bmatrix}_{n \times n}
\]

%\begin{figure}
%     \begin{subfigure}[b]{1.0\textwidth}
%      \centering
%     \includegraphics[width= 4.5in, height = 0.5in]{Survey/Null}
%     \caption{Null Boundary CA}
%     \label{Figure:NUll-Boundary}
%     \end{subfigure}
%     \begin{subfigure}[b]{1.00\textwidth}
%      \centering
%     \includegraphics[width= 4.5in, height = 0.62in]{Survey/Periodic}
%     \caption{Periodic Boundary CA}
%     \label{Figure:Perodic-Boundary}
%     \end{subfigure}
%     \begin{subfigure}[b]{1.00\textwidth}
%      \centering
%     \includegraphics[width= 5.0in, height = 0.62in]{Survey/Intermediate}
%     \caption{Intermediate Boundary CA}
%     \label{Figure:Intermediate-Boundary}
%     \end{subfigure}
%     \caption{Block diagram of null, periodic and intermediate boundary CAs}        		
%\label{Figure:Null-Perodic-CA}       		
%\end{figure}

\begin{figure}
				\subfloat[Null Boundary CA]{\includegraphics[width= 4.5in, height = 0.4in]{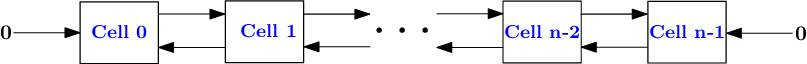}}\\
				\subfloat[Periodic Boundary CA]{\includegraphics[width= 4.5in, height = 0.5in]{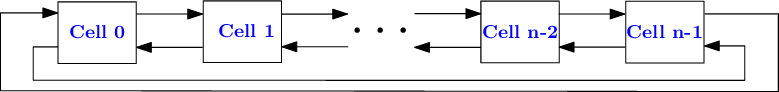}}\\
				\subfloat[Intermediate Boundary CA]{\includegraphics[width= 5.0in, height = 0.5in]{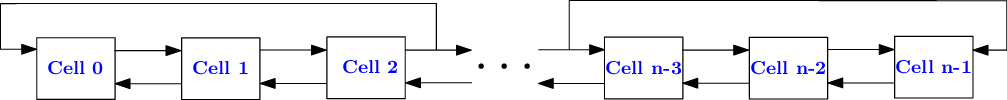}}
\caption{Block diagram of null, periodic and intermediate boundary CAs}        		
\label{figure:boundary-3}  
%\vspace{-1.0em}      		
\end{figure}

Note that, under null boundary condition, the left (resp. right) neighbor of the first (resp. last) cell is considered as in state 0, so the cell is independent of the neighbor. In case of periodic boundary condition, the rightmost and leftmost cells are neighbors of each other. On the other hand, in case for intermediate boundary, the leftmost (rightmost) cell consider the two successive cells on its right (left) as its neighbors (see Figure \ref{figure:boundary-3}). As an example, a $5$-cell CA with rule vector (90, 150, 90, 150, 150) represented by matrix $T_{NB}$ under null boundary condition. This CA is represented by the matrices $T_{PB}$ and $T_{IB}$ when the used boundary conditions are periodic and intermediate respectively. 
\[
T_{NB} =
    \begin{bmatrix}
        0   & 1   & 0   & 0   & 0   \\
        1   & 1   & 1   & 0   & 0   \\
        0   & 1   & 0   & 1   & 0   \\
        0   & 0   & 1   & 1   & 1   \\
        0   & 0   & 0   & 1   & 1   \\
    \end{bmatrix}
    \hspace{2em}
    T_{PB} =
     \begin{bmatrix}
        0   & 1   & 0   & 0   & {\bf 1}   \\
        1   & 1   & 1   & 0   & 0   \\
        0   & 1   & 0   & 1   & 0   \\
        0   & 0   & 1   & 1   & 1   \\
        {\bf 1}   & 0   & 0   & 1   & 1   \\
    \end{bmatrix}
     \hspace{2em}
    T_{IB} =
     \begin{bmatrix}
        0   & 1   & {\bf 1}   & 0   & 0   \\
        1   & 1   & 1   & 0   & 0   \\
        0   & 1   & 0   & 1   & 0   \\
        0   & 0   & 1   & 1   & 1   \\
        0   & 0   & {\bf 1}   & 1   & 1   \\
    \end{bmatrix}
\]
Further, we get the characteristic polynomial of a $T$-matrix as $P(x)=\det(xI+T)$ $\pmod 2$. The characteristic polynomial can be obtained using an efficient recurrence relation which is noted down in Lemma~\ref{lemma:recurrence} \cite{Cattell1}. Before going to prove that, let us define some terms. We adopt the notation $\mathcal{R}_{i,j}$ ($i\leq j$) to refer to a sub-rule vector of $\mathcal{R}$: $\mathcal{R}_{i,j}$ = ($\mathcal{R}_{i}, \mathcal{R}_{i+1}, \cdots, \mathcal{R}_{j}$). Such a sub-rule vector consisting of rules from cell $i$ to cell $j$. We use ${\Delta}_{i,j}$ to denote the characteristic polynomial of $\mathcal{R}_{i,j}$, and abbreviate ${\Delta}_{0,j}$ as ${\Delta}_{j}$. We call ${\Delta}_{i,j}$ a subpolynomial. The complete rule vector $\mathcal{R}$ is written as $\mathcal{R}_{0,n-1}$ with corresponding characteristic polynomial ${\Delta}$. In our further discussion, null boundary condition is assumed as default boundary.
\[
T_{k}=
    \begin{bmatrix}
        d_0 & b_0 & 0 & \cdots & 0 & 0 \\
        a_1 & d_1 & b_1 & \ddots &  & 0 \\
        0 & a_2 & d_2 & \ddots &  & \vdots \\
        \vdots & \ddots & \ddots & \ddots & \ddots & \vdots \\
        0 &  &  & \ddots & d_{k-2} & b_{k-2} \\
        0 & 0 & \cdots & \cdots & a_{k-1} & d_{k-1} \\
    \end{bmatrix}_{k \times k}
\]
\begin{lemma}
\label{lemma:recurrence}
Let $T_{k}$ denote the $k\times k$ submatrix of $T$ and ${\Delta}_{k-1}=\det(xI+T_k)$. Then,
${\Delta}_{k-1}$ satisfies the following recurrence:\\ 
\hspace*{9em}${\Delta}_{-1} = 1$, \\
\hspace*{9em}${\Delta}_0 = (x+d_0)$ \\
\hspace*{9em}${\Delta}_{k-1} = (x+d_{k-1}){\Delta}_{k-2} + b_{k-2} a_{k-1} {\Delta}_{k-3}$
\end{lemma}

\begin{proof}
By expanding $(xI+{T_k})$ along the last row, we have \\
\hspace*{4em} ${\Delta}_{k-1} = (x+d_{k-1}){\Delta}_{k-2} + a_{k-1} \det(B)$,\\
where
\[
B =
    \begin{bmatrix}
        x+d_0 & b_0 & 0 & \cdots & 0 & 0 \\
        a_1 & x+d_1 & b_1 & \ddots &  & 0 \\
        0 & a_2 & x+d_2 & \ddots &  & \vdots \\
        \vdots & \ddots & \ddots & \ddots & \ddots & \vdots \\
        0 & 0 & \cdots & a_{k-3} & x+d_{k-3} & b_{k-3} \\
        0 & 0 & \cdots & \cdots & a_{k-2} & b_{k-2} \\

    \end{bmatrix}
\]
Now, $\det(B) = b_{k-2}(xI+T_{k-2}) + a_{k-2}\cdot 0$\\
Therefore, \\
\hspace*{4em} ${\Delta}_{k-1} = (x+d_{k-1}){\Delta}_{k-2} + b_{k-2} a_{k-1} {\Delta}_{k-3}$
\end{proof}

The characteristic polynomial can be calculated using the above recurrence relation. Lets take an example to get the characteristic polynomial of a CA.

%The characteristic polynomial can be calculated using the following recurrence relation~\cite{Cattell2}, where $\Delta = \det(xI-T) \pmod 2$.

% \begin{multline}
% \label{equation:polynomial}
% \begin{array}{l} 
% {\Delta}_{-1} = 1, \\
% {\Delta}_0 = (x-d_0) \\
% {\Delta}_1 = (x-d_1){\Delta}_0 - b_0 a_1 {\Delta}_{-1}, \\
% {\Delta}_2 = (x-d_2){\Delta}_1 - b_1 a_2 {\Delta}_0, \\   
% \vdots \\
% {\Delta}_k = (x-d_{k}){\Delta}_{k-1} - b_{k-1} a_{k} {\Delta}_{k-2}, \\
% \vdots \\
% {\Delta} = {\Delta}_{n-1} = (x-d_{n-1}){\Delta}_{n-2} - b_{n-2} a_{n-1} {\Delta}_{n-3} \\
% \end{array}
% \end{multline}

\begin{example}
\label{Matrix_example-4}
Let us take a 4-cell CA with rule vector $(90, 150, 90, 150)$ represented by matrix $T$:
\[
T=
    \begin{bmatrix}
        0   & 1   & 0   & 0   \\
        1   & 1   & 1   & 0   \\
        0   & 1   & 0   & 1    \\
        0   & 0   & 1   & 1   \\
    \end{bmatrix}_{4 \times 4}
\]
We can find out the characteristic polynomial of any submatrix of $T\pmod{2}$ according to recurrence relation given above: \\
\hspace*{6em} ${\Delta}_{-1} = 1$, \\
\hspace*{6em} ${\Delta}_0 = (x+0) = x,$ \\
\hspace*{6em} ${\Delta}_1 = (x+1)x + 1 = x^2+x+1,$ \\
\hspace*{6em} ${\Delta}_2 = (x+0)(x^2+x+1) + x = x^3+x^2,$ \\   
\hspace*{6em} ${\Delta}_3 = (x+1)(x^3+x^2) + (x^2+x+1) = x^4+x+1$
\end{example}
Hence, the characteristic polynomial of the given CA is $x^4+x+1$.

%\begin{example}
%\label{Matrix_example-4}
%Take a 4-cell CA with rule vector $(90, 60, 170, 150)$ represented by matrix $T'$. The characteristic polynomial of this matrix is $\det(xI-T') \pmod 2$ = $x^4+x^2+1$ = $(x^2+x+1)^2$.
%\[
%T'=
%    \begin{bmatrix}
%        0   & 1   & 0   & 0   \\
%        1   & 1   & 0   & 0   \\
%        0   & 0   & 0   & 1    \\
%        0   & 0   & 1   & 1   \\
%    \end{bmatrix}_{4 \times 4}
%\]
%Similarly, in another example of a $4$-cell CA of $(90, 150, 90, 150)$, the matrix is \[
%T''=
%    \begin{bmatrix}
%        0   & 1   & 0   & 0   \\
%        1   & 1   & 1   & 0   \\
%        0   & 1   & 0   & 1    \\
%        0   & 0   & 1   & 1   \\
%    \end{bmatrix}_{4 \times 4}
%\]
%We can calculate the characteristic polynomial according to recurrence relation by stepwise. \\
%\hspace*{6em} ${\Delta}_{-1} = 1$, \\
%\hspace*{6em} ${\Delta}_0 = (x-0) = x,$ \\
%\hspace*{6em} ${\Delta}_1 = (x-1)x - 1 = x^2-x-1,$ \\
%\hspace*{6em} ${\Delta}_2 = (x-0)(x^2-x-1) - x = x^3-x^2,$ \\   
%\hspace*{6em} ${\Delta}_3 = (x-1)(x^3-x^2) - (x^2-x-1) = x^4+x+1$
%\end{example}

\subsubsection{Characteristic Polynomial}
\label{subsection:relation}

The characteristic polynomials of $T$-matrices can be reducible, irreducible and primitive. Types of such polynomials depend on two factors - the rules used in rule vectors and sequence of rules. 

%Furthermore, the characteristic polynomials has two options, either its reducible or irreducible (see Definition~\ref{definition:irreducible}). If its irreducible, then there is a chance to be primitive (see Definition~\ref{definition:primtive}). Now, we show that, six linear rules except rules 90 and 150 produce characteristic polynomials which are always reducible. These six linear rules has either $a_i=0$ or $b_i=0$. Beyond this, we make the following theorem (Theorem~\ref{theorem:reduciblepolynomial}).

\begin{defnn}
\label{definition:irreducible}
A nonzero polynomial $P(x)$ over GF($q$) is said to be \emph{irreducible} if it cannot be factored into two non-constant polynomials $G(x)$ and $H(x)$ over the same field, that is, $P(x)\ne G(x)\times H(x)$ for any non-constant $G(x)$ and $H(x)$; otherwise, it is reducible.
\end{defnn}

\begin{defnn}
\label{definition:primtive}
The polynomial $P(x)$ over GF($q$) is \emph{primitive}, if it is irreducible and $\min\limits_{n\in \mathbb{N}}$ $\{n | P(x) \\ \text{ divides } x^n-1\}=q^k-1$. In this case, $P(x)$ has a root $\alpha$ in GF$(q^k)$ such that, it generates all the elements of the extension field GF$(q^k)$ from the base field GF($q$) as successive power of $\alpha$. That is, $\{0,1,\alpha$, $\alpha^2$, $\cdots$, $\alpha^{q^k-2}\}$ is the entire field GF$(q^k)$.
\end{defnn}

\begin{theorem}
\label{theorem:reduciblepolynomial}
In an $n$-cell CA, if $a_i=0$ for any $i \in\{1,2,\cdots,n-1\}$ or $b_j=0$ for any $j \in\{0,1,\cdots,n-2\}$, then ${\Delta}$ is always reducible~\cite{Serra90c,Cattell1}.
\end{theorem} 

\begin{proof}
Suppose, $a_{k}=0$ ($0 < k \le n-1$). Then, according to Lemma~\ref{lemma:recurrence}, we get ${\Delta}_k = (x+d_{k}){\Delta}_{k-1}$. Hence, ${\Delta}_{k-1}$ is a factor of ${\Delta}_k$. Therefore, ${\Delta}_{k+1} = (x+d_{k+1}){\Delta}_{k} + {\Delta}_{k-1}$ is divisible by ${\Delta}_{k-1}$. Continuing by induction, ${\Delta}$ is also divisible by ${\Delta}_{k-1}$. Since ${\Delta}_{k-1}$ is a polynomial of degree $k\geq 1$, it is the case that ${\Delta}$ is reducible. If we now similarly consider $b_{k}=0$, for any $k\in \{0, 1, \cdots, n-2\}$, we get the similar result. Hence the proof.
\end{proof}

From Table~\ref{table:linearrules}, we find that all the linear rules except rules 90 and 150 either $a=0$ or $b=0$. Then, according to Theorem~\ref{theorem:reduciblepolynomial}, these rules cannot participate in maximal length CAs. Hence, we have only two rules, rule 90 and rule 150, which are capable to generate characteristic polynomials as irreducible and primitive. 

\begin{theorem}
\label{theorem:sequence90150}
Only a specific sequence of rules 90 and 150 can form a maximal length CA.
\end{theorem}
\begin{proof}
By contradiction, consider that any sequence of 90 and 150 of length $n\geq 2$ form a maximal length CA. The automaton of Figure~\ref{figure:st-150-90-90-90} contradicts the proposition. Hence, a specific sequence of 90 and 150 forms rule vector of a maximal length CA. 
\end{proof}

An efficient algorithm has been developed by Cattell and Muzio. to synthesize maximal length CA from a primitive polynomial \cite{CattellM96} which will be discussed in Section~\ref{section:synthesis}. This implies that deciding primitive polynomials and deciding maximal length CAs are two equivalent problems. According to the discussions of Section~\ref{section:analysis} and Section~\ref{section:synthesis}, we get the following theorem which relates maximal length CA and primitive polynomials. The proof of this theorem can be derived from these discussion.

%After summarized all the works of Cattell et al.~\cite{CattellM96,Cattell1,CattellTh,Cattell2}, we get the following theorem which relates maximal length CA and primitive polynomials and prove that these two problems are equivalent.

\begin{theorem}
\label{theorem:pp-mlca}
A linear CA is a maximal length CA under null boundary condition iff its characteristic polynomial is primitive over GF(2).
\end{theorem}

\begin{example}
\label{example:pp}
For the CA of $(150, 150, 90, 150)$ under null boundary condition, we get the characteristic polynomial as $x^4+x^3+1$ which is primitive over GF(2). The CA is also maximal length CA (see Figure~\ref{figure:st-150-150-90-150}). Now, for the CA $(150, 90, 90, 90)$, we get the characteristic polynomial $x^4+x^3+x^2+1$, which is not primitive but irreducible. The CA is also not maximal length (see Figure~\ref{figure:st-150-90-90-90}).
\end{example}

Furthermore, the CAs that involve only rules 90 and 150 are generally called {\bf 90-150 CAs}. A cell with rule 90 depends on left and right neighbors, whereas a cell with rule 150 depends on left, self and right neighbors. So, the upper and lower diagonals of a matrix representing a CA with rules 90 and 150 are all 1. Further, the main diagonal is 0 if the rule of cell is 90, and 1 if the rule of cell is 150. 

%Furthermore, we get the maximal length CAs under only in null boundary condition using of rule 90 and rule 150. Therefore, we have only rule 90 and rule 150 which generate maximal length CAs. A cell with rule 90 depends on left and right neighbors, whereas a cell with rule 150 depends on left, self and right neighbors. So, the upper and lower diagonals of a matrix representing a CA with rules 90 and 150 are all 1. Further, the main diagonal is 0 if the rule of cell is 90, and 1 if the rule of cell is 150. 

\[
d_{i}=\left\{ 
\begin{array}{cl}
0, & \mbox{if cell $i$ uses rule 90}\\
1, & \mbox{if cell $i$ uses rule 150}\\
\end{array} 
\right.
\] 

%\[
%A=
%    \begin{bmatrix}
%        d_0 & 1 & 0 & \cdots & 0 & 0 \\
%        1 & d_1 & 1 & \ddots &  & 0 \\
%        0 & 1 & d_2 & \ddots &  & \vdots \\
%        \vdots & \ddots & \ddots & \ddots & \ddots & \vdots \\
%        0 &  &  & \ddots & d_{n-2} & 1 \\
%        0 & 0 & \cdots & \cdots & 1 & d_{n-1} \\
%
%        
%    \end{bmatrix}_{n \times n}
%\]

The characteristic polynomial of 90-150 CAs can be calculated using the following recurrence relation which is simplified from previous recurrence relation (see Lemma~\ref{lemma:recurrence}).

%The matrix representation of an $n$-cell CA using only rules 90 and 150 are described as $A$ under null boundary condition. This matrix is extract from the general matrix $T$ which is define in previous. We get the characteristic polynomial of this matrix is $P(x)=\det(xI-A)\pmod 2$. The characteristic polynomial can be calculated using the following recurrence relation which is simplified from previous recurrence relation (see Lemma~\ref{lemma:recurrence}).

 \begin{multline}
 \label{equation:polynomialgf(2)}
 \begin{array}{l} 
 {\Delta}_{-1} = 1, \\
 {\Delta}_0 = (x+d_0) \\
 {\Delta}_1 = (x+d_1){\Delta}_0 + {\Delta}_{-1}, \\ 
 \vdots \\
 {\Delta}_k = (x+d_{k}){\Delta}_{k-1} + {\Delta}_{k-2}, \\
 \vdots \\
 {\Delta} = {\Delta}_{n-1} = (x+d_{n-1}){\Delta}_{n-2} + {\Delta}_{n-3} \\
 \end{array}
 \end{multline}

However, there is no maximal length CA under periodic boundary condition (this will be discussed in next section). 

\section{Analysis of Maximality}
\label{section:analysis}

%From 1980s, this problem has been considered by several researchers. In particular, Cattell et al. have exploited in this area~\cite{CattellTh,Cattell2,Cattell1,Cattell500,CattellM96,Cattel99}. They were done by developing some algorithms to search and find maximal length CAs. In this section, we discuss about the results of maximal length CAs. Also we describe the strategy behind the results. But, some of researchers generate maximal length CAs only for applications which we will discuss in Section~\ref{section:applications}. 
%However, in the main framework, there are possible two options for generating maximal length CAs. 

Although the first maximal length cellular automaton was identified through an experimental approach \cite{Pries86}, a polynomial theory based approach has been taken later to analyse maximality of a cellular automaton. In fact, these are the only two approaches available to decide maximality of a CA.

\begin{enumerate}
\item[1.] [Exhaustive Search] Test reversibility of a given 90-150 CA. If the CA is reversible, then test if the cycle length is maximal or not.
\item[2.] [Primitivity Based] Find characteristic polynomial of a given 90-150 CA. If the polynomial is primitive, the CA is of maximal length.
\end{enumerate}

%\begin{enumerate}
%\item First one is, we can get the cycle structure of a CA by evolving the configurations from $0^{n-1}1$. And check the cycle length is maximal or not.
%\item In second one, get a primitive polynomial and then extract the CA from this polynomial.
%\end{enumerate}

%Both the approaches needs exponential time. But, few researchers produce some elegant results to generate maximal length CAs. Before showing the results, take the general procedure for checking a given CA is maximal length or not. We take $0^{n-1}1$ as initial configuration for starting the evolution of configurations. Following is the algorithm.

Testing maximality and testing primitivity - both the approaches take exponential time by the best-known algorithms. However, if primitive polynomials are known a priori, the second approach becomes efficient. Some researchers have used known primitive polynomials to analyze as well as to synthesize maximal length CAs. We shall present those results. Before that, we present a generic procedure to test maximality of a 90-150 CA.

\begin{algorithm}
\caption{Decide Maximality}
\label{algorithm:decidemaximality}
\small
%\hspace*{\algorithmicindent} \textbf{Input} $n$-cell CA $\mathcal{R} = (\mathcal{R}_0, \mathcal{R}_1, \cdots, \mathcal{R}_{n-1})$\\
\hspace*{\algorithmicindent} \textbf{Input} $n$-cell 90-150 CA \\
%\hspace*{\algorithmicindent} \textbf{Output} $\mathcal{R}$ is Maximal length or not (\textbf{True/False})
\hspace*{\algorithmicindent} \textbf{Output} Maximal length CA or not (\textbf{True/False})

\begin{algorithmic}[1]
	
\IF{the CA is not reversible} %according to Theorem~\ref{theorem:Reversibility90-150}}
\STATE \textbf{return} False (Not Maximal)
\ENDIF
\STATE $x\leftarrow 0^{n-1}1$, $y\leftarrow G(x)$ and $count\leftarrow 1$
\WHILE{$x\neq y$}
\STATE $y\leftarrow G(y)$
\STATE $count\leftarrow count+1$
\ENDWHILE
\IF{$count\neq 2^n-1$}
\STATE \textbf{return} False (Not Maximal)
\ELSE
\STATE \textbf{return} True (that is, the CA is a maximal length CA)
\ENDIF
\end{algorithmic}
\end{algorithm}

For example, the CA of Figure~\ref{figure:st-150-150-90-150} with $\mathcal{R}$ = (150, 150, 90, 150) is a maximal length CA. If we start counting from $x=0001$, we get $count=15$ after covering all configurations.

%\begin{example}
%Consider the CA of Figure~\ref{figure:st-150-150-90-150} where $\mathcal{R}$ = (150, 150, 90, 150). Here, $\mathcal{R}$ is reversible. Further, we choose initial configuration as $x=0001$. Next, we update the configuration until reach to $x$ with increase $count$. Finally, we get $count=15$ and the CA is maximal length. 
%\end{example}

Hortensius et al.~\cite{Horte89a,Horte89c} added some extra points in maximal length CAs that are used to generate pseudo-random number. Using computer simulation, they proposed a table which gives the non-uniform constructions necessary to achieve maximality for size 4 to 50. The table can be find in Ref. \cite{Horte89a}. These maximal length CAs eventually lead to sequences which closely resemble a random sequence from a regular starting pattern. %These maximal length CAs are also somewhat autoplectic since a regular starting pattern eventually leads to sequences which closely resemble a random sequence. 
Furthermore, we can generate maximal length CAs from primitive polynomials. Few research works have been done by using primitive polynomials~\cite{Barde87,Bardell1992,Hansen1992}.

%Hortensius et al.~\cite{Horte89a,Horte89c} added some extra points in maximal length CAs that are used to generate pseudo-random number. Using computer simulation, they proposed a table which gives the non-uniform constructions necessary to achieve maximality. Table~\ref{Table:mlcaupto53} (see Appendix, page~\pageref{Table:mlcaupto53}) shows maximal length CAs for size 4 to 50. These maximal length CAs eventually lead to sequences which closely resemble a random sequence from a regular starting pattern. %These maximal length CAs are also somewhat autoplectic since a regular starting pattern eventually leads to sequences which closely resemble a random sequence. 
%Furthermore, we can generate maximal length CAs from primitive polynomials. Few research works have been done by using primitive polynomials~\cite{Barde87,Bardell1992,Hansen1992}.

The maximal length CAs that we have seen till now, use null boundary condition. Researchers took effort to find maximal length CAs with other boundary conditions.

\subsection{Boundary Condition Dependence}
\label{subsection:boundarycondition}

Based on the results of extensive simulations, Bardell~\cite{Barde90,Bardell} proposed a conjecture that there is no non-uniform CA with periodic boundary condition that has characteristic polynomial as primitive. The conjecture has later been proved by Nandi et al.~\cite{Nandi96}. They have shown that, for an $n$-cell periodic boundary CA with rules 90 and 150, either $x$ or $(x+1)$ is a factor of the characteristic polynomial, or the characteristic polynomial is a square. Hence, its characteristic polynomial is always reducible. Outline of the proof can be found in Theorem \ref{theorem:synthesize-2}.    

\begin{theorem}
\label{theorem:periodicboundaryMLCA}
There is no 90-150 CA under periodic boundary condition that has the characteristic polynomials as primitive.
\end{theorem}

Nandi et al.~\cite{Nandi96} have also shown that there exist maximal length CAs against each primitive polynomial under intermediate boundary condition. Following theorem establishes the fact that there exists at least one CA in intermediate boundary condition, corresponding to a CA under null boundary condition. Both of the CAs are having same characteristic polynomial.

%Furthermore, Nandi et al.~\cite{Nandi96} show that there exist CAs for all primitive polynomials of different degrees under intermediate boundary condition. That means, we can generate maximal length CAs in intermediate boundary condition. Following theorem establishes the fact that there exists at least one CA in intermediate boundary corresponding to a CA under null boundary condition, both have the same characteristic polynomial. 

\begin{theorem}
\label{theorem:intermediateboundarymlca}
For every 90-150 CA in null boundary, there exists at least one CA in intermediate boundary having the same characteristic polynomial. 
\end{theorem}  

\begin{proof}
The characteristic matrix of a 90-150 CA under null boundary condition is represented as

\[
T_{NB}=
    \begin{bmatrix}
        d_0 & 1 & 0 & 0 & \cdots & 0 & 0 & 0 \\
        1 & d_1 & 1 & 0 & \cdots & 0 & 0 & 0 \\
        0 & 1 & d_2 & 1 & \cdots & 0 & 0 & 0 \\
        0 & 0 & 1 & d_3 & \cdots & 0 & 0 & 0 \\
        \vdots & \vdots & \vdots & \vdots & & \vdots & \vdots & \vdots \\
        0 & 0 & 0 & 0 & \cdots & d_{n-3} & 1 & 0\\
        0 & 0 & 0 & 0 & \cdots & 1 & d_{n-2} & 1\\
        0 & 0 & 0 & 0 & \cdots & 0 & 1 & d_{n-1}\\
                
    \end{bmatrix}_{n \times n}
\]
where $d_i=0$ (resp. 1), if the rule of $i^{th}$ cell is 90 (resp. 150). The characteristic polynomial of this is  $\det(xI+T_{NB})$
{\small
\[
T_{NB}+Ix=
    \begin{bmatrix}
        d_0+x & 1 & 0 & 0 & \cdots & 0 & 0 & 0 \\
        1 & d_1+x & 1 & 0 & \cdots & 0 & 0 & 0 \\
        0 & 1 & d_2+x & 1 & \cdots & 0 & 0 & 0 \\
        0 & 0 & 1 & d_3+x & \cdots & 0 & 0 & 0 \\
        \vdots & \vdots & \vdots & \vdots & & \vdots & \vdots & \vdots \\
        0 & 0 & 0 & 0 & \cdots & d_{n-3}+x & 1 & 0\\
        0 & 0 & 0 & 0 & \cdots & 1 & d_{n-2}+x & 1\\
        0 & 0 & 0 & 0 & \cdots & 0 & 1 & d_{n-1}+x\\
                
    \end{bmatrix}_{n \times n}
\]}
Let us number the rows and columns of the above matrix as $Row(1)$, $Row(2)$, $\cdots$, $Row(n)$, and $Col(1)$, $Col(2)$, $\cdots$, $Col(n)$ respectively. We now perform row and column operations on the above matrix, and get the following matrices.  
{\small
\[
\Rightarrow
    \begin{bmatrix}
        1+d_0+x & 1+d_1+x & 1 & 0 & \cdots & 0 & 0 & 0 \\
        1 & d_1+x & 1 & 0 & \cdots & 0 & 0 & 0 \\
        0 & 1 & d_2+x & 1 & \cdots & 0 & 0 & 0 \\
        0 & 0 & 1 & d_3+x & \cdots & 0 & 0 & 0 \\
        \vdots & \vdots & \vdots & \vdots & & \vdots & \vdots & \vdots \\
        0 & 0 & 0 & 0 & \cdots & d_{n-3}+x & 1 & 0\\
        0 & 0 & 0 & 0 & \cdots & 1 & d_{n-2}+x & 1\\
        0 & 0 & 0 & 0 & \cdots & 1 & 1+d_{n-2}+x & 1+d_{n-1}+x\\
                
    \end{bmatrix}_{n \times n}
\]}
\begin{center}
$Row(1)\leftarrow Row(1)+Row(2)$\\
$Row(n)\leftarrow Row(n)+Row(n-1)$
\end{center}
{\small
\[
\Rightarrow
    \begin{bmatrix}
        1+d_0+x & d_0+d_1 & 1 & 0 & \cdots & 0 & 0 & 0 \\
        1 & 1+d_1+x & 1 & 0 & \cdots & 0 & 0 & 0 \\
        0 & 1 & d_2+x & 1 & \cdots & 0 & 0 & 0 \\
        0 & 0 & 1 & d_3+x & \cdots & 0 & 0 & 0 \\
        \vdots & \vdots & \vdots & \vdots & & \vdots & \vdots & \vdots \\
        0 & 0 & 0 & 0 & \cdots & d_{n-3}+x & 1 & 0\\
        0 & 0 & 0 & 0 & \cdots & 1 & 1+d_{n-2}+x & 1\\
        0 & 0 & 0 & 0 & \cdots & 1 & d_{n-2}+d_{n-1} & 1+d_{n-1}+x\\
                
    \end{bmatrix}_{n \times n}
\]}

\begin{center}
$Col.(2)\leftarrow Col.(1)+Col.(2)$\\
$Col.(n-1)\leftarrow Col.(n)+Col.(n-1)$
\end{center}

{\small
\[
\Rightarrow
    \begin{bmatrix}
        1+d_0 & d_0+d_1 & \fbox{1} & 0 & \cdots & 0 & 0 & 0 \\
        1 & 1+d_1 & 1 & 0 & \cdots & 0 & 0 & 0 \\
        0 & 1 & d_2 & 1 & \cdots & 0 & 0 & 0 \\
        0 & 0 & 1 & d_3 & \cdots & 0 & 0 & 0 \\
        \vdots & \vdots & \vdots & \vdots & & \vdots & \vdots & \vdots \\
        0 & 0 & 0 & 0 & \cdots & d_{n-3} & 1 & 0\\
        0 & 0 & 0 & 0 & \cdots & 1 & 1+d_{n-2} & 1\\
        0 & 0 & 0 & 0 & \cdots & \fbox{1} & d_{n-2}+d_{n-1} & 1+d_{n-1}\\
                
    \end{bmatrix}_{n \times n}
\]}

\begin{center}
$\Rightarrow T_{IB}+Ix$
\end{center}
where $T_{IB}$ is the matrix representation of CA under intermediate boundary condition. Hence the proof.
\end{proof}

We shall see in detail that for each primitive polynomial, there exists at least one CA under null boundary condition. Hence, for every primitive polynomial, there exist at least one CA in intermediate boundary condition. From the exhaustive simulation study, they have proposed following conjecture for intermediate boundary condition.

%Thus, from the above theorem it is clear that every primitive polynomial can be realized by at least one CA in intermediate boundary. From the exhaustive simulation study they proposed the following conjecture under intermediate boundary.

\begin{conjecture}
\label{conjecture:IBCA18}
There exist exactly 18 $n$-cell CAs ($n\ge 4$) with intermediate boundary conditions for every primitive polynomial of degree $n$~\cite{Nandi96}.
\end{conjecture}

\subsection{Palindromic and Uniform CA}
\label{subsection:palindromic}

This section explores some special patterns of rule vectors, for which the CA is always non-maximal. The first of these is palindromic, where the rule vector is identical to its own reversal. That means, it has a symmetric rule vector: ($\mathcal{R}_0$, $\mathcal{R}_1$, $\cdots$, $\mathcal{R}_{n-1}$) = ($\mathcal{R}_{n-1}$, $\mathcal{R}_{n-2}$, $\cdots$, $\mathcal{R}_{0}$). As an example, the CA with rule vector $(150, 90, 150, 90, 150)$ is a palindromic CA. In~\cite{CattellTh}, it has been shown that the characteristic polynomials of an even length palindromic CA is a perfect square, and of an odd length palindromic CA is the product of a monomial and a perfect square. Hence, the characteristic polynomial of a palindromic CA is always reducible. To prove this statement, at first we have to prove the concatenation relation which gives characteristic polynomial of a CA formed by concatenating two CAs. The result is stated in terms of breaking an $n$-cell CA into components 0, 1, $\cdots$, $k-1$ and $k$, $k+1$, $\cdots$, $n-1$ (see Figure~\ref{figure:concatenation}). like before, we use ${\Delta_{i,j}}$ to denote characteristic polynomial of $\mathcal{R}_{i,j}$, and abbreviate ${\Delta}_{0,j}$ as ${\Delta}_{j}$. Note that a characteristic polynomial can be denoted as ${\Delta}_{0,n-1}$, ${\Delta}_{n-1}$ or ${\Delta}$. Later material requires that the characteristic polynomial of an $n$-cell CA for $n=0$ and $n=-1$ be defined: ${\Delta_{0,-1}}$=${\Delta_{-1}}$=1 and ${\Delta_{0,-2}}$=${\Delta_{-2}}$=0.

\begin{figure}
	\centering
	\includegraphics[width= 4.6in, height = 0.8in]{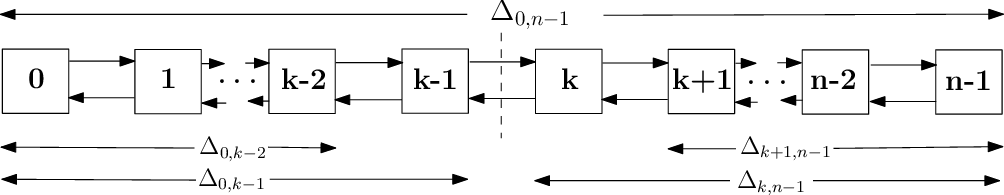}		
	\caption{Sub-rule vector related by the concatenation relation}        		
\label{figure:concatenation}        		
\end{figure}

\begin{theorem}
\label{theorem:concatenation}
For any $k$ with $0\leq k\leq n-1$,\\
\hspace*{1.5in} ${\Delta}_{0,n-1}={\Delta}_{0,k-1}{\Delta}_{k,n-1} + {\Delta}_{0,k-2}{\Delta}_{k+1,n-1}$
\end{theorem}

\begin{proof}
We prove this theorem by induction. The base case of the inductive proof is for $k=0$. By definitions of ${\Delta_{i,j}}$ and from Lemma~\ref{lemma:recurrence} we get that ${\Delta_{0,-1}}$=${\Delta_{-1}}$=1, and ${\Delta_{0,-2}}$=${\Delta_{-2}}$=0. Then the theorem claims that\\
%Using the definitions of ${\Delta}_{i,i-2}$ and ${\Delta}_{i,i-1}$ as 0 and 1 respectively, the theorem claims that\\
\hspace*{1.5in} ${\Delta}_{0,n-1}={\Delta}_{0,-1}{\Delta}_{0,n-1} + {\Delta}_{0,-2}{\Delta}_{1,n-1}$\\
\hspace*{2.0in} $=1\cdot {\Delta}_{0,n-1} + 0\cdot {\Delta}_{1,n-1}$\\
\hspace*{2.0in} $={\Delta}_{0,n-1}$\\
and hence is satisfied trivially.

Assuming that the theorem holds for $k$, it is to show that it holds for $k+1$. By the inductive hypothesis,\\
\hspace*{1.5in} ${\Delta}_{0,n-1}={\Delta}_{0,k-1}{\Delta}_{k,n-1} + {\Delta}_{0,k-2}{\Delta}_{k+1,n-1}$ \\

%A CA has left-right symmetry, in that the polynomial is essentially unchanged if the cell labelings are reversed which are discussed and proved in Section~\ref{subsection:quadratic} (see Corollary \ref{corollary:exactlytwoCA}). Hence, recurrence relation \ref{equation:polynomialgf(2)} can be stated equally well in terms of right-side sub-rule vector. We get\\
%Applying recurrence relation \ref{equation:polynomialgf(2)}, we get\\
A CA has left-right symmetry with respect to its characteristics polynomial. It implies that the characteristic polynomial remains unchanged if the cell labelings are reversed. This result is formally proved in Section~\ref{subsection:quadratic} (see Corollary \ref{corollary:exactlytwoCA}). Hence, the recurrence relation \ref{equation:polynomialgf(2)} can be rewritten considering right-most rule as the first rule. Then we get\\
\hspace*{1.5in} ${\Delta}_{k,n-1}= (x+d_k){\Delta}_{k+1,n-1} + {\Delta}_{k+2,n-1}$ \\
Hence, \\
\hspace*{1.2in} ${\Delta}_{0,n-1}={\Delta}_{0,k-1}{\Delta}_{k,n-1} + {\Delta}_{0,k-2}{\Delta}_{k+1,n-1}$ \\
\hspace*{1.7in} $= {\Delta}_{0,k-1}((x+d_k){\Delta}_{k+1,n-1} + {\Delta}_{k+2,n-1}) + {\Delta}_{0,k-2}{\Delta}_{k+1,n-1}$\\
\hspace*{1.7in} $= {\Delta}_{0,k-1}{\Delta}_{k+2,n-1} + (x+d_k){\Delta}_{0,k-1}{\Delta}_{k+1,n-1} +  {\Delta}_{0,k-2}{\Delta}_{k+1,n-1}$\\
%Factoring out ${\Delta}_{k+1,n-1}$,\\
%\hspace*{1.2in} ${\Delta}_{0,n-1}= ((x+d_k){\Delta}_{0,k-1}+{\Delta}_{0,k-2}){\Delta}_{k+1,n-1} + {\Delta}_{0,k-1}{\Delta}_{k+2,n-1}$ \\
\hspace*{1.7in} $= {\Delta}_{0,k-1}{\Delta}_{k+2,n-1} + ((x+d_k){\Delta}_{0,k-1}+{\Delta}_{0,k-2}){\Delta}_{k+1,n-1} $ \\
Applying (\ref{equation:polynomialgf(2)}) where ${\Delta}_{0,k}={\Delta}_{k}=(x+d_k){\Delta}_{0,k-1}+{\Delta}_{0,k-2}$\\
\hspace*{1.5in} ${\Delta}_{0,n-1}= {\Delta}_{0,k}{\Delta}_{k+1,n-1} + {\Delta}_{0,k-1}{\Delta}_{k+2,n-1}$ \\
Hence the formula holds for $k+1$, given that it holds for $k$. Hence the proof.
\end{proof}

\begin{example}
\label{example:concatenation}
Consider the $4$-cell CA in Example~\ref{Matrix_example-4} with the rule vector $(90, 150, 90, 150)$. With $k=2$\\
\hspace*{1.5in}${\Delta}_{0,3}={\Delta}_{0,1}{\Delta}_{2,3} + {\Delta}_{0,0}{\Delta}_{3,3}$\\
\hspace*{1.75in}$=(x^2+x+1)(x^2+x+1) + x(x+1)$\\
\hspace*{1.75in}$=x^4+x+1$

Hence, the characteristic polynomial of the given CA is $x^4+x+1$ which is same as Example~\ref{Matrix_example-4}.
\end{example}

\begin{theorem}
\label{theorem:palindromicCA}
The characteristic polynomial $\Delta$ (${\Delta}_{0,n-1}$) of a palindromic CA is given by
\[
{\Delta}_{0,n-1} =\left\{ 
\begin{array}{cl}
({\Delta}_{0,\frac{n}{2}-1} + {\Delta}_{0,\frac{n}{2}-2})^2, & \mbox{if $n$ is even}\\
(x+d_{\frac{n-1}{2}})({\Delta}_{0,\frac{n-1}{2}-1})^2, & \mbox{if $n$ is odd}\\
\end{array} 
\right.
\] 
\end{theorem}

\begin{proof}
Suppose that $n$ is even, and let $k=(n/2)-1$. The symmetry of the rule vector $\mathcal{R}$ implies that - ($\mathcal{R}_0$, $\mathcal{R}_1$, $\cdots$, $\mathcal{R}_{n-1}$) = ($\mathcal{R}_0$, $\cdots$, $\mathcal{R}_k$, $\mathcal{R}_k$, $\cdots$, $\mathcal{R}_0$). 

Due to left-right symmetry of CA with respect to its characteristic polynomial, the polynomial remains unchanged if the cell labelings are reversed. Applying it into the palindromic CA, we can get ${\Delta}_{0,k}$ = ${\Delta}_{k+1,n-1}$ and ${\Delta}_{0,k-1}$ = ${\Delta}_{k+2,n-1}$. By substituting these into Theorem~\ref{theorem:concatenation}\\
\hspace*{8em} ${\Delta}_{0,n-1}$ = ${\Delta}_{0,k}{\Delta}_{k+1,n-1} + {\Delta}_{0,k-1}{\Delta}_{k+2,n-1}$\\
\hspace*{11.3em} = ${\Delta}_{0,k}{\Delta}_{0,k} + {\Delta}_{0,k-1}{\Delta}_{0,k-1}$\\
\hspace*{11.3em} = $({\Delta}_{0,k})^2 + ({\Delta}_{0,k-1})^2$\\
\hspace*{11.3em} = $({\Delta}_{0,k} + {\Delta}_{0,k-1})^2$\\
\hspace*{11.3em} = $({\Delta}_{0,\frac{n}{2}-1} + {\Delta}_{0,\frac{n}{2}-2})^2$\\

Now, Suppose $n$ is odd, and let $k=((n-1)/2)-1$. Then the rule vector $\mathcal{R}$ has the form ($\mathcal{R}_0$, $\mathcal{R}_1$, $\cdots$, $\mathcal{R}_{n-1}$) = ($\mathcal{R}_0$, $\cdots$, $\mathcal{R}_{k-1}$, $\mathcal{R}_k$, $\mathcal{R}_{k+1}$, $\mathcal{R}_k$, $\mathcal{R}_{k-1}$, $\cdots$, $\mathcal{R}_0$). Hence, ${\Delta}_{0,k}$ = ${\Delta}_{k+2,n-1}$ and ${\Delta}_{0,k-1}$ = ${\Delta}_{k+3,n-1}$. Now, \\
\hspace*{6em} ${\Delta}_{0,n-1}$ = ${\Delta}_{0,k}{\Delta}_{k+1,n-1} + {\Delta}_{0,k-1}{\Delta}_{k+2,n-1}$\\
\hspace*{9.3em} = ${\Delta}_{0,k}((x+d_{k+1}){\Delta}_{k+2,n-1} + {\Delta}_{k+3,n-1}) + {\Delta}_{0,k-1}{\Delta}_{0,k}$\\
\hspace*{9.3em} = ${\Delta}_{0,k}((x+d_{k+1}){\Delta}_{0,k} + {\Delta}_{0,k-1}) + {\Delta}_{0,k-1}{\Delta}_{0,k}$\\
\hspace*{9.3em} = $(x+d_{k+1}){\Delta}_{0,k}{\Delta}_{0,k} + {\Delta}_{0,k}{\Delta}_{0,k-1} + {\Delta}_{0,k-1}{\Delta}_{0,k}$\\
\hspace*{9.3em} = $(x+d_{k+1})({\Delta}_{0,k})^2$ \\
\hspace*{9.3em} = $(x+d_{\frac{n-1}{2}})({\Delta}_{0,\frac{n-1}{2}-1})^2$  
\end{proof}

From the above theorem, we get that the characteristic polynomial of any palindromic CA is always reducible. Therefore, it can not produce a maximal length cycle. A {\em uniform} CA is trivially palindromic, and so has a reducible characteristic polynomial. %According to this, we develop a corollary for uniform CAs (see Corollary~\ref{corollary:uniformCA}). 
So, uniform CAs cannot have maximal length cycles. Hence, we need non-uniform CAs. 

\begin{corollary}
\label{corollary:uniformCA}
The characteristic polynomial of an $n$-cell uniform CA is always reducible.
\end{corollary}   

\begin{example}
Let us consider a $5$-cell CA with rule vector $(150, 90, 150, 90, 150)$. This is a palindromic CA. The characteristic polynomial of this CA is $x^5+x^4+x^3+x^2+x+1$ = $(x+1)(x^2+x+1)^2$. Hence, it is a reducible polynomial. Let us now take an $8$-cell uniform CA of rule 90. Its characteristic polynomial is $x^8+x^6+x^4+1$, which is reducible as $(x+1)^2(x^3+x+1)^2$.
\end{example}

\subsection{Minimal Cost Maximal Length CA}
\label{subsection:minimalcost}

It is difficult to measure the cost of a CA, as it involves assigning relative costs to addition, multiplication by constants, and storage of elements from the field. However, using the formula - ($\sum_{i=1}^{n-1}{a_i} + \sum_{i=0}^{n-1}{d_i} + \sum_{i=0}^{n-2}{b_i}$), we can calculate the cost of a CA. When it is minimal, then we call the CA as minimal cost CA. As an example, costs of two CAs (150, 90, 150, 150, 90) and (150, 90, 90, 90, 90) are 11 and 9 respectively. Which means the cost of the second CA is minimum here.

%However, for binary CAs, minimal cost maximal length CA means a maximal length CA with minimum number of 150 rule. In~\cite{Zhang91}, Zhang et al. present an algorithm for determining if a given CA has maximal length cycle. This algorithm is used to identify one such CAs for each length upto 150 cells. Following are the steps of the algorithm for $n$-cell.

In case of binary CAs, the above formula implies that minimal cost maximal length CA is a maximal length CA with minimum number of rule 150. In~\cite{Zhang91}, Zhang et al. has presented an algorithm for determining minimal cost maximal length CA. This algorithm is used to identify a such CAs for each length upto 150 cells. Further, the results have been extended by Cattell et al. \cite{Cattell500} upto degree 500. They have reported an algorithm of deciding whether a given $n$-cell CA has a maximal length cycle by checking if the corresponding characteristic polynomial is primitive; If so, the CA has the maximal length cycle.

%Following are the steps of the algorithm for an $n$-cell CA.

%\noindent\rule{8cm}{0.6pt}
%\begin{algorithmic}[1]
%
%\STATE To process factors (performed once for a given $n$)
%\STATE Set $F$ to be the set of prime cofactors of $2^n-1$
%\STATE Remove all $f$ from $F$ where $f|(2^i-1)$, $1<i<n$
%\STATE Add 1 to each remaining element of $F$
%\STATE For each candidate CA, construct the matrix $M$
%\STATE If $M$ is singular, or its main diagonal is a palindrome, then reject the CA
%\STATE For $i=1, 2, \cdots, n-1$, compute and retain $M^{2^i}$ (note $M^{2^i} = (M^{2^{i-1}})^2$);
%\STATE If $M^{2^i}= M$ for any $i$, then the CA is rejected
%\STATE If $M^{2^n}\neq M$, then reject the CA
%\STATE For each $f\in F$, $f = \sum_{i=0}^{n-1} f_i2^i$, if $\pi_{f_i\neq 0} M^{2^i}=M$, then the CA is rejected
%\STATE If the CA is not rejected in one of steps above, it has a maximal length cycle.
%
%\end{algorithmic}
%%\vspace{-0.75em}
%\noindent\rule{8cm}{0.6pt}

In this procedure, an $n$-cell maximal length CA with a single 150 is searched. If this is not successful, then a maximal length CA with two 150 is searched. 
The search was stopped at the first CA with maximal length cycle. This search has never failed, meaning that for each degree upto 500, there is a CA with the maximal length cycle that has either one or two rule 150~\cite{Zhang91,Cattell500}. The results can be find in Ref. \cite{Cattell500}. From the experimental results, following result has been conjectured.

\begin{conjecture}
\label{conjecture:two150rule}
For every $n$, there exist a maximal length CA using at most two 150 rules \cite{Cattell500}.
\end{conjecture}

\subsection{Phase Shift Operation}
\label{subsection:phaseshift}

Let us now turn our attention to an interesting property of a maximal length CA. Consider Table~\ref{table:phaseshift} which notes down consecutive configurations of a maximal length CA (150, 150, 90, 150) (see Figure~\ref{figure:st-150-150-90-150}) against time. Column II shows the bit sequence generated by cell 0. Observe that the same sequence is repeated in column III from $4^{th}$ time step. Start of the sequence in both columns are marked by bold face. The same sequence is also repeated in column IV and column V (start of the sequence is marked by bold face). Let us consider that $(S^0_i)_{1\leq i\leq 15}$ is the bit sequence generated by cell $0$, and $\sigma$ is a left shift operator. Then we get, 

\hspace*{0.8in} $(S^0_i)_{1\leq i\leq 15}$ $=\sigma^{3}((S^1_i)_{1\leq i\leq 15})$ $=\sigma^{13}((S^2_i)_{1\leq i\leq 15})$ $=\sigma^{10}((S^3_i)_{1\leq i\leq 15})$ 

These shifts in sequences with respect to cell $0$ are named as {\em phase shifts}. Hence, the phase shift of 1st, 2nd and 3rd cell with respect to cell 0 are 3, 13 and 10 respectively. Following is the formal definition of phase shift.
\begin{defnn}
\label{definition:phaseshift}
For an $n$-cell maximal length CA, $(S^a_i)_{1\leq i\leq 2^n-1}$ $=\sigma^{k}((S^b_i)_{1\leq i\leq 2^n-1})$ where $a$ and $b$ are two cells of the CA ($0\leq a,b\leq n-1$) and $k$ is the number of shift ($1\leq k\leq 2^n-1$). Such a shift in the bit sequence is called phase shift.
\end{defnn}

%Actually, we get the sequence as \\
%\hspace*{1in} $(S^0_i)_{1\leq i\leq 15}$ $\xrightarrow{3}$ $(S^1_i)_{1\leq i\leq 15}$ $\xrightarrow{10}$ $(S^2_i)_{1\leq i\leq 15}$ $\xrightarrow{12}$ $(S^3_i)_{1\leq i\leq 15}$ \\
%which can be written for an $n$-cell maximal length CA as following -\\
%$(S^0_i)_{1\leq i\leq 2^n-1}$ $\xrightarrow{k_1}$ $(S^1_i)_{1\leq i\leq 2^n-1}$ $\xrightarrow{k_2}$ $\cdots$ $(S^j_i)_{1\leq i\leq 2^n-1}$ $\xrightarrow{k_{j+1}}$ $\cdots$  $\xrightarrow{k_{n-1}}$ $(S^{n-1}_i)_{1\leq i\leq 2^n-1}$

\begin{table}[!t]
	\begin{center}	
		\caption{Phase shift with respect to cell $0$ of the CA (150, 150, 90, 150)}	
		\label{table:phaseshift}
		\begin{tabular}{c||cccc}\hline
		Step & \multicolumn{4}{c}{Cell} \\ \cline{2-5}
		 & 0 & 1 & 2 & 3 \\ \hline\hline
		1 & {\bf \boxit{0.12in} 1} & 0 & 0 & 0 \\\hline
		2 & 1 & 1 & 0 & 0 \\\hline
		3 & 0 & 0 & 1 & 0 \\\hline
		4 & 0 & {\bf \boxit{0.12in} 1} & 0 & 1 \\\hline
		5 & 1 & 1 & 0 & 1 \\\hline
		6 & 0 & 0 & 0 & 1 \\\hline
		7 & 0 & 0 & 1 & 1 \\\hline
		8 & 0 & 1 & 1 & 0 \\\hline
		9 & 1 & 0 & 1 & 1 \\\hline
		10 & 1 & 0 & 1 & 0 \\\hline
		11 & 1 & 0 & 0 & {\bf \boxit{0.12in} 1} \\\hline
		12 & 1 & 1 & 1 & 1 \\\hline
		13 & 0 & 1 & 0 & 0 \\\hline
		14 & 1 & 1 & {\bf \boxit{0.12in} 1} & 0 \\\hline
		15 & 0 & 1 & 1 & 1 \\\hline
\end{tabular}
	\end{center}
\end{table}

Finding of phase shifts of a maximal length CA was investigated by Bardell~\cite{Barde90}. He calculated the phase shifts between the output sequences generated by different stages of a maximal length CA by using discrete logarithms of a binary polynomial. Nandi and Chaudhuri~\cite{Nandi93d} proposed a method for the study of phase shift analysis based on matrix algebra. In~\cite{Sarkar2003}, Sarkar has given an algorithm to find phase shifts. This was achieved by developing a proper algebraic framework for the study of CA sequences. The algorithm is implemented following the algorithm by Tezuka and Fushimi~\cite{Tezuka} which is also based on a result of Mesirov and Sweet~\cite{Mesirov1987}.

%In~\cite{Sarkar2003}, Sarkar gave an algorithm to compute phase shifts. This was achieved by developing the proper algebraic framework for the study of CA sequences. The algorithm is implemented based on the algorithm by Tezuka and Fushimi~\cite{Tezuka} which is also based on a results by Mesirov and Sweet~\cite{Mesirov1987}. In the earlier work, Bardell had provided an example of computing shift in CA sequences of a $6$-cell CA - (150, 90, 90, 90, 90, 90). The characteristic polynomial is $f(x) = x^6+x^5+x^4+x+1$ which is primitive. The shifts were computed to be 0, 39, 35, 47, 33, 32. On the other hand, for this same CA, the computed shifts are 0, 24, 28, 16, 30, 31 by using the Sarkar's algorithm. Basically, Bardell's shifts calculation and Sarkar's shifts calculation are actually obtained in opposite directions. However, the outputs are same whether the shifted technique methods are different. Further, an improved method to compute phase shift has been described by Cho et al.~\cite{Cho2004}. This method is slightly different from the other methods whether we get the same output as previous. Next, we describe the method of phase shift and how is it work.

In an earlier work, Bardell had provided an example of computing shift in CA sequences of a $6$-cell CA - (150, 90, 90, 90, 90, 90). The characteristic polynomial is $P(x) = x^6+x^5+x^4+x+1$ which is primitive. The shifts were computed to be 0, 39, 35, 47, 33, 32. On the other hand, for the same CA, the computed shifts are 0, 24, 28, 16, 30, 31 by using the Sarkar's algorithm. Basically, Bardell's shift calculation and Sarkar's shifts calculation are obtained in opposite directions. So, the phase shift operation of a CA with respect to a given cell position is not uniquely determined by the charateristic polynomial of the CA. The reason is, there exist at least two CAs with respect to a charateristic polynomial and for those two CAs we can obsereved different phase shift properties.

In the above background, Nandi and Chaudhuri~\cite{Nandi93d} proposed a simple method based on matrix algebra to analysis the phase shift operation. Following theorem gives a compact formulation of the method to calculate the phase shift values where $T$ is the correspondence matrix of a given CA.

\begin{theorem}
\label{theorem:phaseshift}
If $T^k$ has a 1 at the $j^{th}$ column of the $i^{th}$ row (where $i\neq j$), with all other elements of the row are 0, then there exist a phase shift of $k$ bit for the bit sequence generated by the $i^{th}$ cell position with respect to that of the $j^{th}$ cell position of the CA having $T$ as its characteristic matrix.  
\end{theorem}

\begin{proof}
Let the state vector at a particular instant of time be $S$ and after $k$ time its $S^k$. Then, assume that the $i^{th}$ row has a single 1 at $j^{th}$ column position. Hence, we can write from matrix algebra tools \cite{das1993vector},
 \begin{center}
 $S^k = T^k \times S$ $\hspace{3em} i.e.$ 
 \end{center}
\[
    \begin{bmatrix}
        s^k_0    \\
        s^k_1    \\
        \vdots   \\
        s^k_i   \\
        \vdots   \\
        s^k_{n-1} \\
    \end{bmatrix}
    \hspace{1em}
   =
   \hspace{1em}
     \begin{bmatrix}
        \vdots & &  \vdots & \vdots & \vdots & & \vdots \\
        0 & \cdots & 0 & 1 & 0 & \cdots & 0 \\
        \vdots & &  \vdots & \vdots & \vdots & & \vdots \\
    \end{bmatrix}
     \begin{bmatrix}
      s_0    \\
      s_1    \\
      \vdots   \\
      s_j   \\
      \vdots   \\
      s_{n-1} \\
    \end{bmatrix}
\]
Hence we have, $s^k_i=s_j$.\\
This proves that the sequence appearing at the $j^{th}$ cell position appears at the $i^{th}$ cell position after $k$ time. Hence the proof. 
\end{proof}

Let's take an example of the above theorem where the CA is same as of Table \ref{table:phaseshift}. As per the above theorem, the position of matrix $T^k$ have been marked to identify shifting characteristic.
\[
T =
    \begin{bmatrix}
        1   & 1   & 0   & 0   \\
        1   & 1   & 1   & 0   \\
        0   & 1   & 0   & 1   \\
        0   & 0   & 1   & 1   \\
    \end{bmatrix}
    \hspace{1.5em}
    T^{2} =
     \begin{bmatrix}
        0   & 0   & \fbox{1}   & 0   \\
        0   & 1   & 1   & 1   \\
        1   & 1   & 0   & 1   \\
        0   & 1   & 1   & 0   \\
    \end{bmatrix}
     \hspace{1.5em}
    T^{3} =
     \begin{bmatrix}
        0   & 1   & 0   & 1   \\
        \fbox{1}   & 0   & 0   & 0   \\
        0   & 0   & 0   & 1   \\
        1   & 0   & 1   & 1   \\
    \end{bmatrix}
    \hspace{1.5em}
    T^{5} =
     \begin{bmatrix}
        0   & 0   & 0   & \fbox{1}   \\
        0   & 0   & 1   & 1   \\
        0   & 1   & 1   & 0   \\
        1   & 0   & 0   & 1   \\
    \end{bmatrix}
\] 

Thus, we can calculate the shifting amount of sequences for different cell positions which means value of $k$. Amount of shift: $k=2$ for shift of cell position 0 with respect to 2, $k=3$ for shift of cell position 1 with respect to 0 and $k=5$ for shift of cell positioncell 0 with respect to 3. As this CA is a $4$-cell maximal length CA having a cycle length of 15; so phase shift with respect to $0^{th}$ cell are 3, 13 and 10 for $1^{st}$, $2^{nd}$ and $3^{rd}$ cell respectively. Based on the above characterization, we outline a method of finding the phase shift in a given $n$-cell maximal length CA (see Algorithm \ref{algorithm:phaseshiftoperation}) with respect to a {\bf predefined  cell position (pivot cell)}. For details of this algorithm and example, see~\cite{Nandi93d}.  

%Based on the above characterization, further, a slightly different method of phase shifts has been given by Cho et al.~\cite{Cho2004}. Next, we outline a method of finding the phase shift in a given $n$-cell maximal length CA (see Algorithm \ref{algorithm:phaseshiftoperation}). For details of this algorithm and example, see~\cite{Cho2004}. 

%Based on the above characterization, we outline a method of finding the phase shift in a given $n$-cell maximal length CA (see Algorithm \ref{algorithm:phaseshiftoperation}). For details of this algorithm and example, see~\cite{Cho2004}. 

%Now, we describe the algorithm to find the phase shifts in a given $n$-cell maximal length CA. 

\begin{algorithm}
\caption{Finding Phase Shifts}
\label{algorithm:phaseshiftoperation}
\small
\hspace*{\algorithmicindent} \textbf{Input} The $n\times n$ matrix $T$ for an $n$-cell maximal length CA\\ 
\hspace{4.2em} pivot cell $p$\\
\hspace*{\algorithmicindent} \textbf{Output} phaseshift[n] /* stores shift of cell positions with respect to pivot cell */

\begin{algorithmic}[1]
	
\STATE $M$ $\leftarrow$ $T$, mark$_{1\times n}$ $\leftarrow$0; mark[p]$\leftarrow$1; power$\leftarrow$1; phaseshift[p]$\leftarrow$ $2^n-1$;
\STATE while (all mark $\neq$ 1) do step 3 to step 6.
\STATE count the number of 1's in the $p^{th}$ row of $M$.\\
if $p^{th}$ row contains single 1, then $j$ $\leftarrow$ column position with 1; otherwise $j$ $\leftarrow$ (-1). \\ 
\STATE if ($j$ $\neq$ 1), then \\
mark[j] $\leftarrow$ 1 \\
phaseshift[j] $\leftarrow$ (($2^n-1$) - power);\\
goto step 6.
\STATE for $i=0$ to $n-1$, do\\
if (mark[i]$\neq$1 and $M[i,p]=1$ and $i\neq p$), then \\
count number of 1's in the $i^{th}$ row of $M$ \\
if $i^{th}$ row contains single 1, then $j$ $\leftarrow$ column position with 1; otherwise $j$ $\leftarrow$ (-1). \\ 
if ($j$ $\neq$ 1), then \\
mark[i] $\leftarrow$ 1 \\
phaseshift[i] $\leftarrow$  power;\\
\STATE $M$ $\leftarrow$ $T\times M$\\
power $\leftarrow$ power+1

\end{algorithmic}
\end{algorithm}

\subsection{Searching for Pattern}
\label{subsection:searchingpattern}

According to the previous discussion, only rules 90 and 150 can take part in the rule vectors of maximal length CAs. As there is no efficient algorithm to decide a maximal length CA, they search for a pattern, if exists, in the rule vectors of maximal length CAs. For a given $n$, there are $2^n$ rule vectors, some of which correspond to maximal length CAs. At first, they efficiently exclude the rule vectors (that is, the CAs) which correspond to some reducible polynomials. In the set of remaining CAs, which obviously correspond to the irreducible polynomials, they search for a pattern in the rule vectors which may indicate the maximal length CAs. %If such a pattern exists, then we can decide efficiently if the given rule vector can (or cannot) be maximal length CA.

In~\cite{JCA19}, they proposed three experimental approaches to observe any pattern in the CAs. The approaches are primarily experimental. In the first approach, the authors have undertaken the machine learning approach. A standard open-source software tool {\sl Weka} is used for that purpose~\cite{Weka-book}. In this approach, they search for a pattern to differentiate maximal length CAs from the other CAs whose characteristic polynomials are irreducible (but not primitive). There are many classifiers, implemented in Weka, from which three well-known algorithms are choosen. They are {\em ZeroR}, {\em LibSVM} and {\em J48} (see \cite{Weka-book} for details of these algorithms). One can find the experimental results in \cite{JCA19}, which have indicated that there is no pattern by observing which one can decide a CA as a maximal length CA.

In the second approach, the concept of ratio of rules were used. They experiment to see if the ratio of 90 and 150 matters in differentiating the maximal length CAs from the CAs with irreducible polynomials. To start with, they first see the distribution of CAs having characteristic polynomial as irreducible (including primitive) against the number of 150s. In this experiment, they observe that for any size $n$, the number of CAs having $k$ 150s is equal to the number of CAs having $(n-k)$ 150s, where $1\le k\le n-1$. For example, for the size 15, there are 388 CAs using five 150s. And, the number of CAs with ten 150s is also 388. To explore this interesting property, the term {\bf conjugate} is introduce.

\begin{defnn}
\label{definition:conjugate}
Two CAs with rule vectors $\mathcal{R}=( \mathcal{R}_0,\mathcal{R}_1,\cdots,\mathcal{R}_{n-1})$ and $\mathcal{R}'=(\mathcal{R}'_0,\mathcal{R}'_1,\cdots,\mathcal{R}'_{n-1})$ are said to be {\bf conjugate} to each other if the following condition is satisfied:\\
$\mathcal{R}_i$ is 90 (resp. 150) iff $\mathcal{R}'_i$ is 150 (resp. 90), for each $i\in\{0,1,\cdots,n-1\}$.
\end{defnn}
As an example, consider a CA $\mathcal{R}=(90, 90, 90, 150, 150)$, then its conjugate CA is $\mathcal{R}'= (150, 150,$ $150, 90, 90)$. We next experiment with the conjugate CAs. And, as a result of this experiment, we get the following property.

\begin{theorem}
\label{theorem:conjugate}
The characteristic polynomial of a CA is irreducible iff the characteristic polynomial of its conjugate CA is irreducible.
\end{theorem}

\begin{proof}
If $p(x)=det(T+xI)$ is the characteristic polynomial of CA $\mathcal{R}$ then the characteristic polynomial of its conjugate $\mathcal{R}'$ is $det(T+I+xI) = det(T+(x+1)I) = p(x+1)$. If $p(x)$ is not irreducible, so that $p(x)=a(x)b(x)$ then also $p(x+1)$ is not irreducible as $p(x+1)=a(x+1)b(x+1)$. Symmetrically, if $p(x+1)$ is not irreducible then its conjugate $p(x)$ is not irreducible either.
\end{proof}

Let us take the following pair of conjugate CAs: (90, 90, 90, 150, 150) and (150, 150, 150, 90, 90). The characteristic polynomial of the first CA is $x^5+x^3+x^2+x+1$, which is irreducible. The polynomial for the second CA is $x^5+x^4+x^3+x+1$, which is also irreducible.

Although it appears to us that conjugate CAs may follow some pattern. The authors have adopted many techniques on conjugate CAs and as well as ratio of rules 90 and 150. But, it has failed to produce any convincing result. However, the idea of conjugate CAs and their corresponding polynomials would be an area of work in future.

In the third approach, they concatenate the small sized CAs and check the final CA is maximal length or not. Let $\mathcal{R}'$ and $\mathcal{R}''$ be rule vectors of two CAs, and $\mathcal{R}=(\mathcal{R}'\mathcal{R}'')$ be another rule vector obtained by concatenating $\mathcal{R}'$ and $\mathcal{R}''$. Let us name the rule vectors $\mathcal{R}'$ and $\mathcal{R}''$ as {\em component} CAs of $\mathcal{R}$. %What is the chance that $\mathcal{R}$ is a maximal length CA? In the experiment, they put some constraints on $\mathcal{R}'$ and $\mathcal{R}''$. However, the $\mathcal{R}$ may be an irreversible CA, or a reversible CA with characteristic polynomial as reducible. In our study, we exclude these cases, because these issues can be resolved in polynomial time. 

In the first experiment, they concatenate two or more CAs which are maximal length and get the final CA. Finally, check the chance of final CA to be maximal length and also comparing with actual percentage of maximal length CAs. In another experiment, they concatenate two or more CAs which are not maximal length but their characteristic polynomials are irreducible and do the same. Observe that, when some maximal length CAs are concatenated, then the chance of the final CA to be maximal length becomes high. The experimental results can be find in Ref. \cite{JCA19}. From the results, we can conclude that the chances improves when the component CAs are maximal length CAs. But, it is little disappointing that finally ther are no clear pattern in the rule vectors of maximal length CAs.

%Table~\ref{table:concatenate} (see Appendix, page~\pageref{table:concatenate}) shows the result of this experiment. The second column is the actual percentage of maximal length CAs and the third column shows the number of times they have synthesized CAs using the scheme. Whereas next two columns note that how many times the CAs with irreducible polynomial, and out of these CAs how many are maximal length CAs respectively. The last column of this table shows the percentage of maximal length CAs, which get by dividing the fifth column by the fourth column.

%It is encouraging to note that the percentage of the last column is always greater than that of column 2. By concatenating two maximal length CAs, and if the characteristic polynomial of the CA is irreducible, then the chance of the CA to be of maximal length is more than the actual ratio (second column of Table~\ref{table:concatenate}). Hence, the chance improves when the component CAs are maximal length CAs. But, it is little disappointing that finally ther are no clear pattern in the rule vectors of maximal length CAs.

We finish this section by posing an open problem about the complexity issue of maximal length cellular automata:

\begin{oproblem}
\label{oproblem:analysis}
Decide a given CA as maximal length in polynomial time.
%Develop an algorithm which can find a given CA as maximal length in polynomial time.
\end{oproblem}

\section{Synthesis of Maximal Length CA from Primitive Polynomial}
\label{section:synthesis}

Let us now synthesize a maximal length CA from a given primitive polynomial. By ``synthesis'' we mean to find out a rule vector of a maximal length CA from a given primitive polynomial. The method that we are going to present next synthesizes CAs from an irreducible polynomial in GF(2). When the polynomial is primitive, then synthesized CA is of maximal length. It is also shown here that there are exactly two CAs against each irreducible polynomial. One can find the details of the algorithm in~\cite{CattellM96}.

%Let us now synthesize a maximal length CA from a given primitive polynomials. Now its time to reverse, that means, calculation of a CA from a given polynomial. Next, we present a method for the calculation of a one-dimensional linear non-uniform CA from a given irreducible or primitive polynomial. Further, it is shown that there are exactly two CA corresponding to any irreducible polynomial. In here we present the algorithm pseudocode with example. The complete technical theory of the algorithm can be found in~\cite{CattellM96}.

The pair of polynomials $\Delta_{n-1}$ and $\Delta_{n-2}$ (defined in Section~\ref{subsection:matrix}) uniquely determine the CA using Euclid's GCD algorithm. Conversely, a CA uniquely determines $\Delta_{n-1}$ and $\Delta_{n-2}$. 
%This is why it is desirable to distinguish reversals. $\Delta_{0,n-2}$ and $\Delta_{1,n-1}$ play an important role, and except for the degenerated symmetric case, are different polynomials. This role is so important that it motivates the following definition.
If we divided $\Delta_{n-1}$ by $\Delta_{n-2}$, we get ($x+d_{n-1}$) as quotient and $\Delta_{n-3}$ as remainder as per recurrence relation \ref{equation:polynomialgf(2)}. So, if we proceed to find GCD of $\Delta_{n-1}$ and $\Delta_{n-2}$ by Euclid's algorithm, we observe a series of quotients $x+d_{n-1}$, $x+d_{n-2}$, $\cdots$, $x+d_{0}$ that come out as part of the GCD computation. This gives us the CA as ($d_0$,$d_1$,$\cdots$,$d_{n-1}$). This result is stated in Lemma~\ref{lemma:subpolynomial}. 
%To proceed further, let us define the following.

\begin{defnn}
\label{definition:subpolynomial-pair}
For an $n$-cell CA, $\Delta_{0,n-2}$ is called a CA {\em subpolynomial}, and the pair of polynomials $\Delta_{0,n-1}$, $\Delta_{0,n-2}$ are called a CA polynomial-subpolynomial pair. Note that a CA subpolynomial is itself a CA polynomial.
\end{defnn}

%CA polynomial means, the polynomial which is calculated from a CA. A degree $n$ polynomial $p$ might not be a CA polynomial. But supposing that it is, how can determine if $p$ and a given polynomial $q$ form a CA polynomial-subpolynomial pair? First of all, $q$ must be degree $n-1$, since it is the characteristic polynomial of a length $n-1$ CA. Second, the correspondence between Euclid's GCD algorithm (see Ref~\cite{CattellM96}, for details of the GCD algorithm) and the computation of a CA characteristic polynomial provides a simple method of determination, described in the following lemma.

\begin{lemma}
\label{lemma:subpolynomial}
Let $p$ and $q$ be two polynomials, with degree $n$ and $n-1$ respectively. Then, there exists a CA with characteristic polynomial $p$ and characteristic subpolynomial $q$ ($\Delta_{n-1}=p$ and $\Delta_{n-2}=q$) if and only if applying Euclid's GCD algorithm to $p$ and $q$, $n$ quotients with degree 1 are received in sequence.
\end{lemma}

\begin{proof}
Suppose there exists a CA with $\Delta_{n-1}=p$ and $\Delta_{n-2}=q$. Lemma \ref{lemma:recurrence} (relation \ref{equation:polynomialgf(2)}) ensures that the sequence of remainder polynomials obtained from Euclid's GCD algorithm is $\Delta_{n-3}$, $\Delta_{n-4}$, $\cdots$, $\Delta_{1}$, $\Delta_{0}$. As these polynomials have degrees differing by one, Euclid's algorithm produces $n$ quotients, each of which is of degree 1. For the converse, suppose that Euclid's algorithm is applied to a degree $n$ polynomial $p$ and degree $n-1$ polynomial $q$, and that $n$ quotients $x+d_{n-1}$, $x+d_{n-2}$, $\cdots$, $x+d_{0}$ are obtained in sequence. The last remainder is always zero, and the fact that there are $n$ quotients means that the last divisor (second last remainder) is one. 
Now if we apply recurrence relation \ref{equation:polynomialgf(2)} and use $d_0$, $d_1$, $\cdots$, $d_{n-1}$, we get the polynomials $\Delta_{n-1}=p$ and $\Delta_{n-2}=q$. Hence, there exists a CA with the desired properties.

%Hence the algorithm terminates with the base cases of the CA recurrence. Clearly, if the CA recurrence is applied to $[d_0,d_1,\cdots,d_{n-1}]$, it calculates the polynomials $\Delta_{n-1}=p$ and $\Delta_{n-2}=q$. Hence, there exists a CA with the desired properties.
\end{proof}

\begin{example}
\label{example:polynomialtoCA}
Let us consider the CA (90, 150, 150, 90, 90) where $\Delta_{n-1}=x^5+x^3+1$ and $\Delta_{n-2}=x^4+1$. %The constant terms of the n degree one quotients specify the CA.
When Euclid's GCD algorithm is applied on ${\Delta_{n-1}}$ and ${\Delta_{n-2}}$, following results come out in successive divisions. 

{\bf dividend} \hspace{4em} {\bf divisor} \hspace{4em} {\bf quotient} \hspace{4em} {\bf remainder} \hspace{4em} {\bf $d_i$ (Rule)}

$x^5+x^3+1$ \hspace{3.5em} $x^4+1$ \hspace{5em} $x$ \hspace{7em} $x^3+x+1$ \hspace{4.5em} $0$ $(90)$

$x^4+1$ \hspace{5em} $x^3+x+1$ \hspace{3.9em} $x$ \hspace{7em} $x^2+x+1$ \hspace{4.5em} $0$ $(90)$

$x^3+x+1$ \hspace{3.5em} $x^2+x+1$ \hspace{3.7em} $x+1$ \hspace{5.5em} $x$ \hspace{8.1em} $1$ $(150)$

$x^2+x+1$ \hspace{3.5em} $x$ \hspace{7.3em} $x+1$ \hspace{5.5em} $1$ \hspace{8.1em} $1$ $(150)$

$x$ \hspace{7em} $1$ \hspace{7.5em} $x$ \hspace{7em} $0$ \hspace{8.1em} $0$ $(90)$

\end{example}

Note that we have extracted the rules from the series of 5 quotients.

\subsection{Solutions to the Quadratic Congruence}
\label{subsection:quadratic}

It is clear now that finding a $q$ satisfying Lemma~\ref{lemma:subpolynomial} would solve the CA synthesis problem, if it could be done efficiently. The degree $n-1$ polynomial $\Delta_{n-2}$ satisfies an expression called the CA {\em quadratic congruence} \cite{Cattell91} which is a key result and central to the synthesis approach. The usefulness of the result is discussed below and its application to the synthesis algorithm is described later. We now ready to present the key result.

\begin{theorem}
\label{theorem:congruence}
Suppose we have a CA with characteristic polynomial ${\Delta}_{0,n-1}$ and characteristic subpolynomials ${\Delta}_{0,n-2}$ and ${\Delta}_{1,n-1}$. The both $y={\Delta}_{0,n-2}$ and $y={\Delta}_{1,n-1}$ satisfy the congruence \\
\hspace*{9.3em} $y^2 + (x^2+x){\Delta}'_{0,n-1}y + 1\equiv 0\pmod{{\Delta}_{0,n-1}}$\\
where ${\Delta}'_{0,n-1}$ is formal derivative of ${\Delta}_{0,n-1}$.
\end{theorem}

%The proof of the above theorem consists of several steps, with some of the steps requiring extensive case analysis. To simplify the presentation, we have omit the proof here (see \cite{CattellM96} for details about proof). 

\begin{proof}
The proof consists of several steps, with some of the steps requiring extensive case analysis. One case of each type is shown for illustrative purpose, to simplify the presentation.

%The proof starts by defining an {\em linear finite state machine} (LFSM) known as a {\em cyclic} CA. Such machines are also known in the literature as {\em periodic boundary} CA. 

%The proof starts by defining a {\em cyclic} CA which is also known {\em periodic boundary} CA. The {\em Cyclic CA Lemma} shows that the characteristic polynomial of a cyclic CA is related to the characteristic polynomial of its underlying null boundary CA. This relation leads to the Cyclic CA Theorem. This is in turn combined with some lemma concerning rule-changing in CA to obtain the {\em Sum of Subpolynomials Theorem}. The CA quadratic congruence follows from this and the {\em product of subpolynomials} relation (the latter being a standard property of Euclid's GCD algorithm).

The characteristic polynomial of a CA under periodic boundary condition is related to the characteristic polynomial of its underlying null boundary CA. This relation leads to Theorem~\ref{theorem:synthesize-2}. This is in turn combined with Lemma \ref{lemma:synthesis-1} to obtain Theorem~\ref{theorem:synthesize-3}. The CA quadratic congruence follows from this and the {\em product of subpolynomials} relation (the latter being a standard property of Euclid's GCD algorithm).

A periodic boundary CA is formed from its underlying null-boundary CA by the addition of interconnections between cells 0 and $n-1$. Hence the characteristic matrix of a periodic boundary CA contains a 1 in each of the corner entries $(0, n-1)$ and $(n-1, 0)$. The characteristic polynomial of a periodic boundary CA, say $\Phi_{0,n-1}$ can be obtained from that of null boundary CA. 
%Also, $\Phi_{i,j}$ denotes the characteristic polynomial of the periodic boundary CA obtained by (1) taking the sub-rule vector $\mathcal{R}_{i,j}$ (see Section~\ref{subsection:matrix}) of the underlying null boundary CA, and (2) adding connections between cells $i$ and $j$. By carefully chosen minor expansions of the periodic boundary CA transition matrix, it is shown that the characteristic polynomial of a periodic boundary CA is related to its null boundary CA.

\begin{lemma}
\label{lemma:synthesis-1}
The characteristic polynomial of a periodic boundary CA ($\Phi_{0,n-1}$) is related to the characteristic polynomial of its underlying null boundary CA as follows\\
\hspace*{12em} $\Phi_{0,n-1} = \Delta_{0,n-1} + \Delta_{1,n-2}$\\
\end{lemma}

\begin{proof}
Let A be the characteristic matrix of a periodic boundary CA.%Consider the characteristic polynomial of a periodic boundary CA
\[
A=
    \begin{bmatrix}
        a_0 & 1 & 0 & \cdots & \cdots & 0 & 1 \\
        1 & a_1 & 1 & \ddots & &  & 0 \\
        0 & 1 & a_2 & \ddots & \ddots &  & \vdots \\
        \vdots & \ddots & \ddots & \ddots & \ddots & \ddots & \vdots \\
        \vdots &  & \ddots & \ddots & a_{n-3} & 1 & 0 \\
        0 &  &  & \ddots & 1 & a_{n-2} & 1 \\
        1 & 0 & \cdots & \cdots & 0 & 1 & a_{n-1} \\        
    \end{bmatrix}_{n \times n}
\]
where $a_i$ denotes $x+d_i$. By expanding the determinant along the first row, we get\\
\hspace*{8em} $\Phi_{0,n-1} = a_0 \Delta_{1,n-1} + \vert B\vert + \vert D\vert$, \\
where
%\[
%B=
%    \begin{bmatrix}
%        1 & 1 & 0 & \cdots & 0 \\
%        0 & & & &  \\
%        \vdots & & L_{2,n-1} & & \\
%        0 & & & &  \\
%        1 & & & &  \\
%    \end{bmatrix} 
%\hspace*{3em} and \hspace*{3em}  
%D=
%    \begin{bmatrix}
%        1 & & & &  \\
%        0 & & & L_{1,n-2} &  \\
%        \vdots & & & &  \\
%        0 & & & &  \\
%        1 & 0 & \cdots & 0 & 1  \\ 
%    \end{bmatrix}
%\]

\[
B=
   \begin{bmatrix}
        1 & 1 & 0 & \cdots & \cdots & 0 & 0 \\
        0 & a_2 & 1 & \ddots & &  & 0 \\
        0 & 1 & a_3 & \ddots & \ddots &  & \vdots \\
        \vdots & \ddots & \ddots & \ddots & \ddots & \ddots & \vdots \\
        \vdots &  & \ddots & \ddots & a_{n-3} & 1 & 0 \\
        0 &  &  & \ddots & 1 & a_{n-2} & 1 \\
        1 & 0 & \cdots & \cdots & 0 & 1 & a_{n-1} \\        
    \end{bmatrix}
\hspace*{1em} and \hspace*{1em}  
D=
   \begin{bmatrix}
        1 & a_1 & 1 & 0 & \cdots & \cdots & 0 \\
        0 & 1 & a_2 & 1 & \ddots & & \vdots \\
        \vdots & \ddots & 1 & \ddots & \ddots & \ddots & \vdots \\
        \vdots &  & \ddots & \ddots & \ddots & 1 & 0 \\
        \vdots &  &  & \ddots & \ddots & a_{n-3} & 1 \\
        0 & & &  & \ddots & 1 & a_{n-2} \\
        1 & 0 & \cdots & \cdots & \cdots & 0 & 1 \\        
    \end{bmatrix}
\]

%($L_{i,j}$ is the characteristic matrix of the sub-rule vector $\mathcal{R}_{i,j}$ of the underlying null boundary CA). Now, by expanding $B$ along the first column\\
Now, by expanding $B$ along the first column and $D$ along the last row,\\
\hspace*{12em} $\vert B\vert = \Delta_{2,n-1} + \vert E\vert$, and\\ 
\hspace*{12em} $\vert D\vert = \Delta_{1,n-2}+\vert F\vert$\\
where
\[
E=
   \begin{bmatrix}
        1 & 0 & \cdots & \cdots & \cdots & \cdots & 0 \\
        a_2 & 1 & \ddots & & & & 0 \\
        1 & a_3 & \ddots & \ddots &  & & \vdots \\
        0 & 1 & \ddots & \ddots & \ddots &  & \vdots \\
        \vdots & \ddots & \ddots & a_{n-4} & 1 & \ddots & \vdots \\
        \vdots &  & \ddots & 1 & a_{n-3} & 1 & 0\\
        0 & \cdots & \cdots & 0 & 1 & a_{n-2} & 1 \\        
    \end{bmatrix}
\hspace*{1em} and \hspace*{1em}  
F=
   \begin{bmatrix}
        1 & 1 & 1 & 0 & \cdots & \cdots & 0 \\
        0 & 1 & a_2 & 1 & \ddots & & \vdots \\
        \vdots & \ddots & 1 & \ddots & \ddots & \ddots & \vdots \\
        \vdots &  & \ddots & \ddots & \ddots & 1 & 0 \\
        \vdots &  &  & \ddots & \ddots & a_{n-4} & 1 \\
        \vdots & & &  & \ddots & 1 & a_{n-3} \\
        0 & \cdots & \cdots & \cdots & \cdots & 0 & 1 \\        
    \end{bmatrix}
\]

Since $E$ and $F$ are lower and upper triangular, respectively, with all 1s on the diagonal, $\vert E\vert = \vert F\vert = 1$. Substituting these into the expressions for $B$ and $D$,\\
\hspace*{16em} $\vert B\vert = \Delta_{2,n-1} + 1$\\ 
\hspace*{16em} $\vert D\vert = \Delta_{1,n-2} + 1$\\
%where $E$ is unit lower triangular. By expanding $D$ along the first column, $\vert D\vert = \Delta_{1,n-2}+\vert F\vert$, where $F$ is unit upper triangular. Since $E$ and $F$ are lower and upper triangular respectively, with all ones on the diagonal, $\vert E\vert = \vert F\vert = 1$. Hence,\\
and hence, \\
\hspace*{6em} $\Phi_{0,n-1} = a_0\Delta_{1,n-1} + \vert B\vert + \vert D\vert$\\
\hspace*{9em} $=(x+d_0)\Delta_{1,n-1} + \Delta_{2,n-1} + 1 + \Delta_{1,n-2} +1$\\
\hspace*{9em} $=\Delta_{0,n-1} + \Delta_{1,n-2}$\\
where $\Delta_{0,n-1}=(x+d_0)\Delta_{1,n-1} + \Delta_{2,n-1}$ 

\end{proof}

A periodic boundary CA has {\em rotational symmetry}, which implies that the characteristic polynomial remains unchanged if the cell labels are rotated. %For example, the two periodic boundary CA in Figure~\ref{figure:cyclicCA} have the same characteristic polynomial. 
A crucial consequence of this is that the $n$ (typically different) null boundary CAs associated with the $n$ label rotations of a periodic boundary CAs, all have the same value for $\Delta_{0,n-1} + \Delta_{1,n-2}$. Let $w$ denote the weight of the rule vector depending on even and odd numbers of rules 90 and 150. If the rule vector contains an even number of rule 150, then $w=0$ and $1$ otherwise.

%\begin{figure}
%	\centering
%	\includegraphics[width= 4.6in, height = 0.45in]{cyclicCA.eps}		
%	\caption{Cell label rotation in a periodic boundary CA.}        		
%\label{figure:cyclicCA}        		
%\end{figure}

\begin{theorem}
\label{theorem:synthesize-2}
Consider a periodic boundary CA with characteristic polynomial $\Phi_{0,n-1}$, and let $w$ denote the weight of its rule vector. Then,
\begin{equation}
\label{equation:cyclicCA}
\Phi_{0,n-1}=\left\{ 
\begin{array}{rl}
q^2 & \mbox{if $n$ is even and $w=0$}\\
(x+1)q^2 & \mbox{if $n$ is odd and $w=1$}\\
(x)q^2 & \mbox{if $n$ is odd and $w=0$}\\
(x)(x+1)q^2 & \mbox{if $n$ is even and $w=1$}\\ 
\end{array} 
\right.
\end{equation}
for some polynomial $q$.
\end{theorem}

\begin{proof}
The proof is by induction on $n$. The base cases can be checked by calculating the characteristic polynomials of all periodic boundary CAs for $n\leq 4$. The inductive step consists of ten cases.

The first eight cases apply when the periodic boundary CA has a label rotation so that $d_0=d_{n-1}$ (i.e., the CA contains adjacent cells that use the same rule). The eight cases come about from the possible values of (1) $n$ is even or odd, (2) $w$ is zero or one, and (3) $d_0(=d_{n-1})$ is zero or one. It is straightforward to show that
\begin{equation}
\label{equation:cyclicCA1}
\begin{aligned}
\Phi_{0,n-1} & = \Delta_{0,n-1} + \Delta_{1,n-2} \\
& = (x+d_0)\Delta_{1,n-1} + \Delta_{2,n-1} + \Delta_{1,n-2} \\
& = (x+d_0)\Delta_{1,n-1} + (x+d_{n-1})\Delta_{2,n-2} + \Delta_{2,n-3} + \Delta_{1,n-2} \\
& = (x+d_0)(\Delta_{1,n-1} + \Delta_{2,n-2}) + (\Delta_{2,n-3} + \Delta_{1,n-2}) \hspace*{2em} here, \hspace*{0.5em} (d_0 = d_{n-1}) \\
& = (x+d_0)\Phi_{1,n-1} + \Phi_{1,n-2}
\end{aligned}
\end{equation}
Suppose that $n$ is even, $w=0$, and $d_0=0$. Then $[d_1,d_2,\cdots,d_{n-1}]$ has odd length and weight zero, and so by induction $\Phi_{1,n-1}=xq^2_{y}$ for some $q_y\in GF(2)[x]$. Similarly, $[d_1,d_2,\cdots,d_{n-2}]$ has even length and weight zero (recall that $d_0=d_{n-1}$), and so $\Phi_{1,n-2}=q^2_z$ for some $q_z$. Applying Equation~\ref{equation:cyclicCA1} \\
\hspace*{6em} $\Phi_{0,n-1} = (x+d_0)\Phi_{1,n-1} + \Phi_{1,n-2}$\\
\hspace*{9em} $=x(xq^2_y) + q^2_z$\\
\hspace*{9em} $=(xq_y + q_z)^2$\\
\hspace*{9em} $=q^2$\\
where $q=xq_y+q_z$. Thus the inductive hypothesis is satisfied. The other seven cases are proven along the same lines.

The last two cases come about when no two adjacent cells use the same rule. For these cases, $n$ is necessarily even, but $w$ could be zero or one. However, such a periodic boundary CA has in the form of $(90,150,90,150,\cdots,90,150)$. This special structure gives the equation\\
\hspace*{8em} $\Phi_{0,n-1} = x(x+1)\Phi_{2,n-1}+\Phi_{4,n-1}$\\
By an argument similar to that for the eight cases above, the inductive hypothesis can be shown to hold. This completes the proof of the theorem.
\end{proof}

We now return our attention to null boundary CA. We need to consider the CA along with three closely related CAs, obtained by changing the rules used in one or both of the end cells. The rule vectors of the CAs and their characteristic polynomials (stated in terms of subpolynomials of the original CA) are
\begin{center}
\begin{tabular}{ll}
\underline{case} & \underline{characteristic polynomial}\\
original $[d_0,d_1,\cdots,d_{n-2},d_{n-1}]$ & $\Delta_{0,n-1}$ \\
complement $d_{n-1}$ $[d_0,d_1,\cdots,d_{n-2},d_{n-1}+1]$ & $\Delta_{0,n-1} + \Delta_{0,n-2}$ \\
complement $d_{0}$ $[d_0+1,d_1,\cdots,d_{n-2},d_{n-1}]$ & $\Delta_{0,n-1} + \Delta_{1,n-1}$ \\
complement both $[d_0+1,d_1,\cdots,d_{n-2},d_{n-1}+1]$ & $\Delta_{0,n-1} + \Delta_{0,n-2} + \Delta_{1,n-1} + \Delta_{1,n-2}$ \\
\end{tabular}
\end{center}

where $d_i=0$ when cell $i$ rule is 90 and $d_i=1$ when rule is 150. And, complement of rules means replace the rule 90 (150) by rule 150 (90). 
\begin{theorem}
\label{theorem:synthesize-3}
The sum of subpolynomials states that\\
\[
\Delta_{0,n-2}+\Delta_{1,n-1}=\left\{ 
\begin{array}{rl}
(x^2+x)\Delta_{0,n-1}' & \mbox{if $n$ is even, $w=0$}\\
\Delta_{0,n-1}+(x^2+x)\Delta_{0,n-1}' & \mbox{if $n$ is even, $w=1$}\\
(x+1)\Delta_{0,n-1}+(x^2+x)\Delta_{0,n-1}' & \mbox{if $n$ is odd, $w=0$}\\
x\Delta_{0,n-1}+(x^2+x)\Delta_{0,n-1}' & \mbox{if $n$ is odd, $w=1$}\\
\end{array} 
\right.
\]
\end{theorem}

\begin{proof}
The proof requires four cases, the first of which is as follows. Suppose that $n$ is even and $w$ is zero. Equation~\ref{equation:cyclicCA} states that there exists a polynomial $q_1$ such that\\
\hspace*{8em} $\Phi_{0,n-1}=\Delta_{0,n-1}+\Delta_{1,n-2} = q_1^2$\\
If $d_{n-1}$ is complemented, $[d_1,\cdots,d_{n-2}]$ remains unchanged and the length is still even, but the weight has become one. Hence, again by Equation~\ref{equation:cyclicCA}, there exists $q_2$ such that
\begin{equation}
\label{equation:cyclicCA2}
(\Delta_{0,n-1}+\Delta_{0,n-2})+\Delta_{1,n-2}=(x^2+x)q_2^2
\end{equation}
Similarly, if $d_0$ is complemented, there is a $q_3$ with
\begin{equation}
\label{equation:cyclicCA3}
(\Delta_{0,n-1}+\Delta_{1,n-1})+\Delta_{1,n-2}=(x^2+x)q_3^2
\end{equation}
If both $d_0$ and $d_{n-1}$ are complemented, the resulting CA has even length and weight zero, and so there exists $q_4$ such that
\begin{equation}
\label{equation:cyclicCA4}
\begin{aligned}
(\Delta_{0,n-1}+\Delta_{0,n-2}+\Delta_{1,n-1}+\Delta_{1,n-2})+\Delta_{1,n-2}=q_4^2 \\
\Rightarrow \Delta_{0,n-1}+\Delta_{0,n-2}+\Delta_{1,n-1} = q_4^2
\end{aligned}
\end{equation}
By adding Equations~\ref{equation:cyclicCA2} and~\ref{equation:cyclicCA3}, and differentiating both sides, we respectively get\\
\hspace*{8em} $\Delta_{0,n-2}+\Delta_{1,n-1} = (x^2+x)(q_2+q_3)^2$\\
\hspace*{10em} $\Rightarrow (\Delta_{0,n-2}+\Delta_{1,n-1})' = (q_2+q_3)^2$\\
Hence, we have the differential equation
\begin{equation}
\label{equation:cyclicCA5}
(\Delta_{0,n-2}+\Delta_{1,n-1}) = (x^2+x)(\Delta_{0,n-2}+\Delta_{1,n-1})'
\end{equation}
Equation~\ref{equation:cyclicCA4} says that $\Delta_{0,n-1}+\Delta_{0,n-2}+\Delta_{1,n-1}$ is a perfect square, and so 
\begin{equation}
\label{equation:cyclicCA6}
\Delta_{0,n-1}' = (\Delta_{0,n-2}+\Delta_{1,n-1})'
\end{equation}
Combining Equations~\ref{equation:cyclicCA5} and~\ref{equation:cyclicCA6} \\
\hspace*{8em} $\Delta_{0,n-2} + \Delta_{1,n-1} = (x^2+x)\Delta_{0,n-1}'$\\
which completes the proof for this case. The other cases are different, but the same argument structure holds (see, Ref. \cite{CattellTh} for details).
\end{proof}

\noindent{\em CA Quadratic Congruence:} Theorem~\ref{theorem:synthesize-3} immediately provides the congruence
\begin{equation}
\label{equation:cyclicCA7}
\Delta_{0,n-2}+\Delta_{1,n-1} \equiv (x^2+x)\Delta_{0,n-1}' \pmod{\Delta_{0,n-1}}
\end{equation}
The product of subpolynomials relation comes from Euclid's GCD algorithm (one can find the details in \cite{CattellTh}). It states that 
\begin{equation}
\label{equation:cyclicCA8}
\Delta_{0,n-2}\Delta_{1,n-1} \equiv 1 \pmod{\Delta_{0,n-1}}
\end{equation}
Solving Equation~\ref{equation:cyclicCA7} for $\Delta_{1,n-1}$ and substituting into Equation~\ref{equation:cyclicCA8}\\
\hspace*{4em} $(\Delta_{0,n-2})^2 + (x^2+x)\Delta_{0,n-1}'\Delta_{0,n-2} + 1 \equiv 0 \pmod{\Delta_{0,n-1}}$\\
Equivalently,

\hspace*{4em} $(\Delta_{1,n-1})^2 + (x^2+x)\Delta_{0,n-1}'\Delta_{1,n-1} + 1 \equiv 0 \pmod{\Delta_{0,n-1}}$

This completes the proof of Theorem \ref{theorem:congruence}.
\end{proof}

We demonstrate the theorem on the CA (90, 90, 150, 150, 90), which has characteristic polynomial $x^5+x^3+1$. For this CA, $\Delta_{0,3}=x^4+x$ and $\Delta_{1,4}=x^4+1$.  Substituting $y=x^4+1$ into the left-hand side of Theorem \ref{theorem:congruence}.
 
\hspace*{5.3em} $y^2 + (x^2+x){\Delta}'_{0,n-1}y + 1$\\
\hspace*{6.3em} $=(x^4+1)^2+(x^2+x)(x^4+x^2)(x^4+1)+1$\\
\hspace*{6.3em} $=(x^8+1)+(x^{10}+x^9+x^8+x^7+x^6+x^5+x^4+x^3)+1$\\
\hspace*{6.3em} $=0 \pmod{x^5+x^3+1}$

Similarly, substituting $y=x^4+x$ into the left-hand side results in zero. By combining Lemma~\ref{lemma:subpolynomial} and Theorem~\ref{theorem:congruence}, we have a characterization of CA polynomials.

\begin{corollary}
\label{corollary:congruence}
Let $p$ be a degree $n$ polynomial. Then $p$ is a CA polynomial if and only if there exists some solution $q$ for $y$ of the congruence\\
\hspace*{9.3em} $y^2 + (x^2+x)p'y + 1\equiv 0\pmod{p}$\\
where $p'$ is the formal derivative of $p$.
\end{corollary}

Theorem~~\ref{theorem:congruence} is very useful for irreducible polynomials. As discussed above, when $\Delta_{0,n-1}$ is irreducible, solutions to quadratic congruence (Theorem~\ref{theorem:congruence}) can be found elegantly which extracts the following corollary.

\begin{corollary}
\label{corollary:exactlytwoCA}
If $p$ is a degree $n$ irreducible polynomial, then $p$ has exactly two CA realizations.
\end{corollary}

\begin{proof}
\label{corollary:twoCArealization}
As per Theorem~\ref{theorem:congruence}, we get two solutions for a characteristic polynomial. This implies, two CAs have the same characteristic polynomial. Hence, for each irreducible polynomial, there are exactly two CAs.
\end{proof}

As an example, suppose we have an irreducible polynomial which is $x^5+x^2+1$. Then, we get two CAs where both of which characteristic polynomials are $x^5+x^2+1$. For this polynomial the CAs are  - $(150, 150, 150, 150, 90)$ and $(90, 150, 150, 150, 150)$ where the CAs are reversal of each other.

%\subsection{Calculation of CA from a Polynomial}
\subsection{The Formula for the Solutions}
\label{subsection:CAfrompolynomial}

Let $p$ be a degree $n$ polynomial. Then according to previous theorem, we can write the congruence according to $p$. To simplify the statement of the formula, we use the following denotations.\\
\hspace*{9.3em} $y^2 + (x^2+x)p'y + 1\equiv 0\pmod{p}$\\
Let us take,  $f(x) = (x^2+x)p'$ and $g(x)= \Big(\frac{1}{f(x)}\Big)^2$.
Some definitions are necessary to describe the algorithm. At first, we define the {\em trace}~\cite{McEliece87} of a polynomial $a(x)$ with respect to the irreducible polynomial $p$ as 
\vspace*{-1.0em}
\begin{align*}
{Tr}(a) & = (a + a^2 + a^4 + \cdots + a^{2^{n-1}}) \pmod p \\
& =\bigg(\sum_{i=0}^{n-1} a^{2^i}\bigg)\pmod p
\end{align*}

At this point, we require a polynomial $\theta(x)$ with trace one. If $n$ is odd then we may use $\theta(x)=1$, but if $n$ is even we have to find such a polynomial. Interestingly, any such $\theta(x)$ suffices. Now, we define the polynomial $\beta(x)$ by 
\vspace*{-0.5em}
\begin{align*}
\beta(x) & = g{{\theta}^2} + (g+g^2){\theta}^2 + \cdots + (g+g^2+\cdots+g^{2^{n-2}}){\theta}^{2^{n-1}}\\
& =\sum_{i=1}^{n-1}\bigg(\sum_{j=0}^{i-1} g^{2^j}\bigg) {\theta}^{2^i}
\end{align*}

If $n$ is odd, the use of $\theta(x)=1$ simplifies the expression for $\beta(x)$ to\\
% \begin{align*}
%     \beta(x) & = g^2 + g^8 + \cdots + g^{2^{n-2}} \\
%     & = \sum_{i=1}^{(n-1)/2} g^{2^{2i-1}}
% \end{align*}

\hspace*{12.8em} $\beta(x) = g^2 + g^8 + \cdots + g^{2^{n-2}}$\\
\hspace*{12.0em} $$=\sum_{i=1}^{(n-1)/2} g^{2^{2i-1}}$$

We now show how the solutions of the quadratic equations are used to obtain the CA. Because $p$ is irreducible, we know that the quadratic has two distinct solution (according to Theorem~\ref{theorem:congruence}), and we denote them $q_1(x)$ and $q_2(x)$. In fact, the CA obtained by using $q_2$ is the reversal of the CA obtained by using $q_1$. In practice, we need only one solution. We apply Euclid's GCD algorithm~\cite{Knuth81} to get the solutions. The pseudocode listed below summarizes the process of computing a CA for a given irreducible polynomial $p$.

\noindent\rule{8cm}{0.6pt}
\begin{algorithmic}[1]
\STATE \textbf{procedure} {$CalcCAFromPolys$}
\STATE let $f = (x^2+x)p'(x)$
\STATE let $g = (1/f)^2$
\STATE let $q=\beta f$

\IF{$n$ is even}
\STATE find $\theta$ with trace one
\STATE calculate $\beta$
\ELSE
\STATE calculate $\beta$
\ENDIF

\STATE calculate $q$
\STATE calculate gcd(p,q)
\STATE \textbf{return} the quotients form the CA from the constant terms of the quotients.

\end{algorithmic}
\vspace{-0.75em}
\noindent\rule{8cm}{0.6pt}

We present a complete example of computing the CA for a degree 5 irreducible polynomial. It shows the determination of the solution of the CA quadratic congruence, and the subsequent application of Euclid's GCD algorithm. The following steps mirror the algorithm description above.

Let $p$ be an irreducible (in this case, $p$ is primitive) polynomial ($n=5$)
\begin{center}
%\vspace{-1em}
$p = x^5+x^2+1$
%\vspace{-1em}
\end{center}

Now, we calculate the formal derivative of $p$, reducing the coefficients modulo $2$ 
\begin{center}
%\vspace{-1em}
$p' = 5x^4+2x = x^4$
%\vspace{-1em}
\end{center}

Compute $f$, recalling that all polynomials are reduced in modulo $p$
\begin{center}
%\vspace{-1em}
$f = (x^2+x)p' = x^6+x^5 = x^3+x^2+x+1$
%\vspace{-1em}
\end{center}

Applying the extended Euclidean GCD algorithm \cite{CattellM96} to compute the inverse of $f$, we get 
\begin{center}
%\vspace{-1em}
$1/f = x^3+x^2+1$
%\vspace{-1em}
\end{center}

To get $g$, we calculate
\begin{center}
%\vspace{-1em}
$g = (1/f)^2 = x^6+x^4+1 = x^4+x^3+x+1$
%\vspace{-1em}
\end{center}

The trace of $g$ can be verified to be zero, as expected. Since $n$ is odd, we use the formula to getting $\beta$ where $\theta=1$ 
% \begin{center}
% %\vspace{-2em}
% $$\beta=\sum_{i=1}^{2} g^{2^{2i-1}}$$
% $= g^2 + g^8$
% %\vspace{-1em}
% \end{center}

\vspace*{-1.5em}
\begin{align*}
\beta & =\sum_{i=1}^{2} g^{2^{2i-1}} \\
& = g^2 + g^8
\end{align*}

Further, we have to compute the powers of $g$.
% \begin{center}
% %\vspace{-1em}
% $g = x^4+x^3+x+1$\\
% $g^2 = (x^4+x^3+x+1)^2 = x^8+x^6+x^2+1 = x$\\
% $g^4 = (g^2)^2 = x^2$\\
% $g^8 = (g^4)^2 = x^4$
% %\vspace{-1em}
% \end{center}

\vspace*{-1.5em}
\begin{align*}
g & = x^4+x^3+x+1\\
g^2 & = (x^4+x^3+x+1)^2 = x^8+x^6+x^2+1 = x\\
g^4 & = (g^2)^2 = x^2\\
g^8 & = (g^4)^2 = x^4    
\end{align*}

Therefore, we get
\vspace*{-1.0em}
\begin{align*}
\beta & = g^2 + g^8\\
& = x+x^4
\end{align*}

Finally, we compute a solution $q$ to the quadratic
\vspace*{-0.2em}
\begin{align*}
q & = (x^2+x)p'\beta\\
& = (x^2+x)(x^4)(x+x^4)\\
& = x^4+x^2+1
\end{align*}

Euclid's algorithm is now applied to determine the CA. At this point, we know that $q$ is ${\Delta}_{n-2}$
\vspace*{-0.2em}
\begin{align*}
x^5+x^2+1 & = {\color{red}(x)}(x^4+x^2+1)+x^3+x^2+x+1\\
x^4+x^2+1 & = {\color{red}(x+1)}(x^3+x^2+x+1)+x^2\\
x^3+x^2+x+1 & = {\color{red}(x+1)}(x^2)+x+1\\
x^2 & = {\color{red}(x+1)}(x+1)+1\\
x+1 & = {\color{red}(x+1)}(1)+0
\end{align*}

The algorithm results in the five degree one quotients $[x, x+1, x+1, x+1, x+1]$, and so the CA is $(150, 150, 150, 150, 90)$. %Further, there is an another different method from Cattell and Muzio method. In~\cite{Cho2007}, Cho et al. proposed a new method for synthesis of 1-dimensional CA for a given irreducible polynomial. 
%The algorithm is extracted from the algorithm of Serra and Slater~\cite{serra1990lanczos}. By the proposed algorithm, they obtained the 9689 cell CA for the primitive polynomial $x^{9689}+x^{4187}+x^{2444}+x^{1836}+x^{471}+x^{84}+1$~\cite{Zierler1969} using about 156 min.

\section{Synthesis of Primitive Polynomial using CA}
\label{section:synthesis-PP}

In this section, we use CA as tool of generating primitive polynomials over GF(2). Till date, there is no efficient algorithm known to decide primitivity of a polynomial in polynomial time. Also, there is no proof claiming that it is an NP-complete or NP-hard problem. Till date the complexity for deciding a primitive polynomial is exponential \cite{Nirmal04}. In \cite{Makato98, Adak-TCS-2021}, the authors consider some special type of CAs to generate maximal length CAs. These CAs can be used as generators of primitive polynomials of a set of degrees. However, they do not always get maximal length CAs by those special CAs. Hence, these generators are imperfect.

Matsumoto~\cite{Makato98} has first worked with such a special type of CAs. He has chosen a simple and explicit construction of CAs which consist of only one rule 150 (for the first cell) and rest of the rules are rule 90. This CA may be used for generating maximal length CAs. Further, this work has been extended in \cite{Adak-TCS-2021}.

%The drawback of the this method is a strong limitation on the size. He has given a list of the sizes $\leq$ 300 that attain the maximality. In our recent works, we extended the list of the sizes from 300 to 1200 that attain maximality \cite{Adak-TCS-2021,Sumit-EAIT}.

\subsection{Imperfect Strategies}
\label{subsection:strategies}

In \cite{Adak-TCS-2021}, the authors develope three strategies for finding maximal length CAs, hence primitive polynomials over GF(2) of large degree. The proposed strategies greedily synthesize (linear) CAs of different sizes, which are almost always of maximal length. Since the characteristic polynomial of a maximal length cellular automaton is primitive, characteristic polynomials of the synthesized CAs are claimed as primitive. The main drawback is the proposed strategies do not generate CAs of every size, though it can generate CAs of arbitrary large size. Follwoing are the definition of two simple special type of CAs which they are used to generate maximal length CAs.

\begin{defnn}
\label{definition:specialCAs}
An $n$-cell CA with rule vector $\mathcal{R}=(\mathcal R_0, \mathcal{R}_1, \cdots, \mathcal{R}_{n-1})$
\begin{itemize}
\item is called $CA(90')$ with $n$-cell if $\mathcal{R}_0$ = 150 and $\mathcal{R}_i$ = 90 where $1\leq i\leq n-1$.
\item is called $CA(150')$ with $n$-cell if $\mathcal{R}_0$ = 90 and $\mathcal{R}_i$ = 150 where $1\leq i\leq n-1$.
\end{itemize}
\end{defnn}

%To fulfil our need, we carefully choose two special types of CAs, named

Note that $CA(90')$ and $CA(150')$ are {\em conjugate} (discussed in Section~\ref{subsection:searchingpattern}) to each other. Next to identify some sizes for which the CAs do not show maximality. These greedy procedures are backed by the dynamics of the CAs and run in polynomial time. The strategies for $CA(90')$ and $CA(150')$ greedily work on the configurations that consist of all but one non-zero element. Decimal equivalents of these configurations are always power-of-2 numbers, namely {\em p-configuration}.

%The procedure for the strategy takes the CA size $n$ as input. Based on some logic, the procedure returns false to report that the CA is non-maximal, and true to indicate that the CA may be of maximal length.

\begin{defnn}
\label{definition:pconfig}
A configuration is a $p$-configuration if only one cell is in state 1 and others are in state 0. 
\end{defnn}

The configuration 0001, for example, is a $p$-configuration. In an $n$-cell CA, there are $n$ $p$-configurations, which are named as $p^0$, $p^1$, $p^2$, $\cdots$, $p^{n-1}$. Here, $p^i$ indicates that cell $n-i$ is at state 1 and the rest are in state 0. For example, the configuration 0010 is $p^1$ and the configuration 1000 is $p^3$. It is, however, obvious that if a CA with $n$ cells is a maximal length CA, then all the $p$-configurations lie in a single cycle.

\subsubsection{Properties of $CA(90')$ and $CA(150')$}

%Space-time diagram is a well-known tool which has been used since long to understand dynamical behaviour a CA \cite{Wolfr83}. Here we note down some properties of $CA(90')$ and $CA(150')$ by observing its space-time diagram. Needless to say that maximality of a CA can also be observed in its space-time diagram. Following is an important property of $CA(90')$ which notes down the behaviour of $p$-configurations in a cycle and we take $p^0$ as initial configuration. To state the property, we use a function $length$ defined as following: $length (x,y)=k$ for two configurations $x,y$ if $y=G^k(x)$, where $G$ is the global transition function of the CA.

Here we note down some properties of $CA(90')$ and $CA(150')$ with respect to maximality. Following is an important property of $CA(90')$ which notes down the behaviour of $p$-configurations in a cycle and we take $p^0$ as initial configuration. To state the property, a function $length$ is defined as following: $length (x,y)=k$ for two configurations $x,y$ if $y=G^k(x)$, where $G$ is the global transition function of the CA.

\begin{theorem}
\label{property:CA90'strategy}
In $CA(90')$ with $n\ge 2$ cells, following relation holds:\\
%\begin{enumerate}
\indent\indent\indent\indent\indent\indent $length(p^{m_0},p^{m_1})=2^0$,\\
\indent\indent\indent\indent\indent\indent $length(p^{m_1},p^{m_2})=2^1$,\\
\indent\indent\indent\indent\indent\indent\indent\indent\indent\indent $\vdots$\\
\indent\indent\indent\indent\indent\indent $length(p^{m_i},p^{m_{i+1}})=2^i$,\\
\indent\indent\indent\indent\indent\indent\indent\indent\indent\indent $\vdots$\\
\indent\indent\indent\indent\indent\indent $length(p^{m_k},p^{m_0})=2^{k}$\\
%\end{enumerate}
where $m_0=0$ and the sequence $(m_i)_{0\le i\le k}$ with some $k\le n-1$ maintains the following relation:
\begin{center}
$m_{i+1}=
\left\{ 
\begin{tabular}{cl}
$2m_i+1$ & ~~\mbox{if}~~ $2m_i+1<n$\\
$2(n-1-m_i)$ & ~~\mbox{otherwise}
\end{tabular}
\right. $ \cite{Adak-TCS-2021}
\end{center}
\end{theorem}

%Now, we show efficacy of the theorem by giving two examples followed by Figure \ref{figure:spacetime-max}. Figure \ref{figure:spacetime-max} shows the space-time diagram of $CA(90')$ and $CA(150')$. Space-time diagram is a well-known tool which has been used since long to understand the transitions and dynamical behaviour of a CA \cite{Wolfr83}. Let us first consider that $n=5$ (see, Figure~\ref{figure:spacetime-max}(b)). Here, sequence of $p$-configurations is $(p^0,p^1,p^3,p^2,p^4)$, and $length(p^0,p^1)=2^0$, $length(p^1,p^3)=2^1$, $length(p^3,p^2)=2^2$, $length(p^2,p^4)=2^3$ and $length(p^4,p^0)=2^4$. All $p$-configurations are covered in one cycle, and this CA is a maximal length CA. Let us now take another example of $n=4$ (see, Figure~\ref{figure:spacetime-max}(a)). The $p$-configurations covered in this diagram are $p^0$, $p^1$ and $p^3$, and $length(p^0,p^1)=2^0$, $length(p^1,p^3)=2^1$ and $length(p^3,p^0)=2^2$. However, this CA is not a maximal length CA.  

Lets take an example of the above theorem. The authors used the space-time diagram instead of transition diagram in here. Space-time diagram is a well-known tool which has been used since long to understand the transitions and dynamical behaviour of a CA \cite{Wolfr83}. Figure \ref{figure:spacetime-max} shows the space-time diagram of $CA(90')$ and $CA(150')$.

%Space-time diagram is a well-known tool which has been used since long to understand the transitions and dynamical behaviour of a CA \cite{Wolfr83}. The authors used the space-time diagram instead of transition diagram in here. Figure \ref{figure:spacetime-max} shows the space-time diagram of $CA(90')$ and $CA(150')$. Lets take an example of the above theorem.

Let us first consider that $n=5$ (see, Figure~\ref{figure:spacetime-max}(b)). Here, sequence of $p$-configurations is $(p^0,p^1,p^3,p^2,p^4)$, and $length(p^0,p^1)=2^0$, $length(p^1,p^3)=2^1$, $length(p^3,p^2)=2^2$, $length(p^2,\\p^4)=2^3$ and $length(p^4,p^0)=2^4$. All $p$-configurations are covered in one cycle, and this CA is a maximal length CA. Let us now take another example of $n=4$ (see, Figure~\ref{figure:spacetime-max}(a)). The $p$-configurations covered in this diagram are $p^0$, $p^1$ and $p^3$, and $length(p^0,p^1)=2^0$, $length(p^1,p^3)=2^1$ and $length(p^3,p^0)=2^2$. However, this CA is not a maximal length CA.

\begin{figure}
				\subfloat[$n$=4]{\includegraphics[width= 0.75in, height = 1.6in]{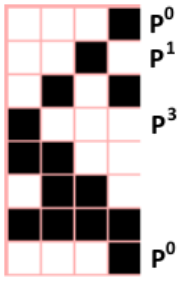}}\hspace*{3em}
				\subfloat[$n$=5]{\includegraphics[width= 0.9in, height = 4.5in]{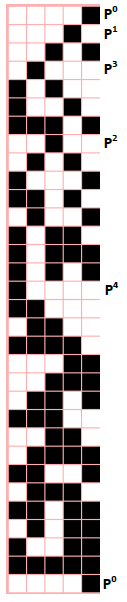}}\hspace*{3em}
				\subfloat[$n$=5]{\includegraphics[width= 0.9in, height = 4.5in]{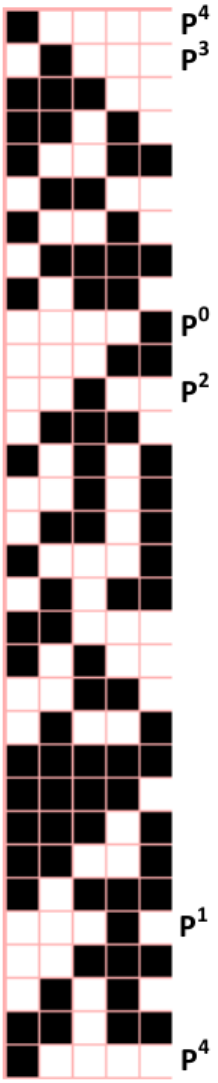}} \hspace*{3em}
				\subfloat[$n$=6]{\includegraphics[width= 0.9in, height = 3.1in]{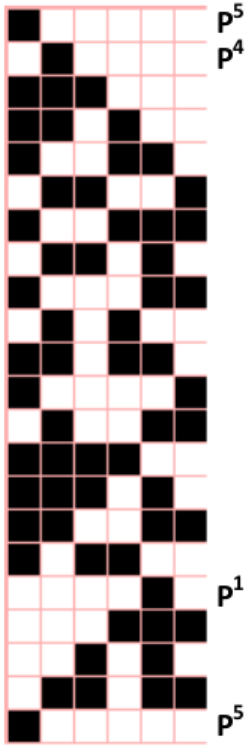}}		
\caption{Space-time diagram of $CA(90')$ and $CA(150')$. Figures (a) and (b) for $CA(90')$ where the evolutions have started from configuration $p^0(0^{n-1}1)$ and figures (c) and (d) for $CA(150')$ where evolutions have started from $p^{n-1}(10^{n-1})$. Here, white is for state $0$ and black is for state $1$.}        		
\label{figure:spacetime-max}  
%\vspace{-1.0em}      		
\end{figure}

\begin{corollary}
\label{property:totallength90'}
For an $n$-cell $CA(90')$, if all the $p$-configurations lie in a single cycle and two consecutive $p$-configurations $p^x$ and $p^y$ maintain the relation: $y=2x+1$ if $(2x+1)<n$; otherwise, $y=2(n-1-x)$, then the CA is a maximal length CA.  
\end{corollary}

\begin{proof}
%According to Definition~\ref{definition:pconfig}, there are $n$ number of $p$-configurations are present in an $n$-cell $CA(90')$. 
Let us assume that the $n$ $p$-configurations of an $n$-cell $CA(90')$ are in a single cycle and two consecutive $p$-configurations $p^x$ and $p^y$ of the cycle maintain the following relation: $y=2x+1$ if $(2x+1)<n$; otherwise, $y=2(n-1-x)$. Hence, we get from Theorem~\ref{property:CA90'strategy} that lengths between two consecutive $p$-configurations in the cycle, starting from $p^0$, are - $2^0$, $2^1$, $\cdots$, $2^{n-1}$. Therefore, length of the cycle is $length(p^0,p^0)=2^0+2^1+\cdots +2^{n-1} = 2^n-1$. Hence, the CA is a maximal length CA. 
\end{proof}

One may observe that in many occasions the sequence of $p$-configurations ($p^0,p^{m_0},p^{m_1},\cdots,$ $p^{m_{k-2}},p^{n-1},p^0$) maintaining the stated condition of two consecutive $p$-configurations in Theorem~\ref{property:CA90'strategy} does not include all the $p$-configurations. These CAs are not maximal length CAs. On the other hand, in most of the cases when all the $p$-configurations are covered maintaining the above stated relation between two consecutive $p$-configurations, then it is a maximal length CA.

In the case of $CA(150')$, similar behaviour is observed. Unlike previous, $p^{n-1}$ is considered as initial configuration. The relation between the $p$-configurations are noted below.

\begin{theorem}
\label{property:CA150'strategy}

In $CA(150')$ with $n\ge 2$ cells, following relation holds:\\
\indent\indent\indent\indent\indent\indent $length(p^{m_0},p^{o_0})=2^{0}$,\\
\indent\indent\indent\indent\indent\indent $length(p^{m_1},p^{o_1})=2^{1}$,\\
\indent\indent\indent\indent\indent\indent\indent\indent\indent\indent $\vdots$\\
\indent\indent\indent\indent\indent\indent $length(p^{m_i},p^{o_i})=2^{i}$,\\
\indent\indent\indent\indent\indent\indent\indent\indent\indent\indent $\vdots$\\
\indent\indent\indent\indent\indent\indent $length(p^{m_k},p^{o_k})=2^{k}$\\
where $m_0=n-1$, $o_0=n-2$ and the sequences $(m_i)_{0\le i\le k}$ and $(o_i)_{0\le i\le k}$ with some $k\le n-1$ obey the following relation:
\begin{center}
$m_{i+1}=
\left\{ 
\begin{tabular}{cl}
$2m_i+1$ & ~~\mbox{if}$~~$ $2m_i+1<n$\\
$2(n-1-m_i)$ & $~~$ \mbox{otherwise}
\end{tabular}
\right. $
\vspace*{1in}
$o_{i+1}=
\left\{ 
\begin{tabular}{cl}
$2o_i+1$ & ~~\mbox{if}$~~$ $2o_i+1<n$\\
$2(n-1-o_i)$ & $~~$\mbox{otherwise}
\end{tabular}
\right. $ \cite{Adak-TCS-2021}
\end{center}
\vspace{-8em}
\end{theorem}

Let us take examples to illustrate the theorem. At first, consider $n=5$. For two sequences of $p$-configurations $(p^4,p^0,p^1,p^3,p^2,p^4)$ and $(p^3,p^2,p^4,p^0,p^1,p^3)$, we observe that $length(p^4,p^3)=2^0$, $length(p^0,p^2)=2^1$, $length(p^1,p^4)=2^2$, $length(p^3,p^0)=2^3$ and $length(p^2,\-p^1)=2^4$. From these two sequences, we can extract a sequence of $p$-configurations $p^4$, $p^3$, $p^0$, $p^2$, $p^1$, $p^4$. This sequence is observed in a single cycle with length between two consecutive $p$-configurations are $2^0$, $2^3$, $2^1$, $2^4$ and $2^2$ respectively. Hence, it is a maximal length CA (see, Figure~\ref{figure:spacetime-max}(c)). In the next example, consider $n=6$. Two sequences of $p$-configurations are - ($p^5,p^0,p^1,p^3,p^4,p^2$) and ($p^4,p^2,p^5,p^0,p^1,p^3$). From where, we can get a sequence $p$-configurations as $p^5, p^4, p^1, p^5$ where $length(p^5,p^4)=2^0$, $length(p^4,p^1)=2^4$, $length(p^1,p^5)=2^2$. Since all the $p$-configurations are not in a single cycle and the CA is not maximal length (see, Figure~\ref{figure:spacetime-max}(d)).

Formal proofs of Theorem \ref{property:CA90'strategy} and Theorem \ref{property:CA150'strategy} are quite lengthy. One can find the details of the proofs in Ref.\cite{Adak-TCS-2021}.

Let us now state the greedy strategy for $CA(90')$ and $CA(150')$. Based on some logic, the procedure returns false to report that the CA is non-maximal, and true to indicate that the CA may be of maximal length. The logic is developed based on Theorem \ref{property:CA90'strategy} for $CA(90')$ and Theorem \ref{property:CA150'strategy} for $CA(150')$. The details of the procedures - {\bf procedure $CA(90')$} and {\bf procedure $CA(150')$} can be find in Ref. \cite{Adak-TCS-2021}.

\subsubsection{Generators}

Whenever {\bf procedure $CA(90')$} and {\bf procedure $CA(150')$} returns ``True'', we understand from its dynamical behaviour that there is a good chance for the CA to be a maximal length CA. Table~\ref{table:resultCA90'} and Table~\ref{table:resultCA150'} shows the effectiveness of the procedurs. The numbers in the table denote the size of $CA(90')$ and $CA(150')$ for which {\bf procedure $CA(90')$} and {\bf procedure $CA(150')$} returns ``True''. 
%In our experiment, we first find out characteristic polynomials of $CA(90')$ with mentioned sizes, and then use the available data of Ref.\cite{connor} to check if these characteristic polynomials are primitive over GF(2). If yes, then corresponding CA is a maximal length CA, otherwise not. 
In Table~\ref{table:resultCA90'} and Table~\ref{table:resultCA150'} some numbers are marked by `*'. These numbers are the sizes for which $CA(90')$ and $CA(150')$ are not maximal length CAs. That is, for these sizes, {\bf procedure $CA(90')$} and {\bf procedure $CA(150')$} returns ``True'', but the CA is not maximal length CA. Hence, these numbers are {\em false positive} if we use the procedure as maximal length CAs generator.

\begin{table}[h]
	\setlength{\tabcolsep}{5.5pt}
	\begin{center}
		%\vspace{-0.9em}	
		\caption{List of sizes when procedure of $CA(90')$ returns ``True''. *-marked sizes are false positive}
		%\vspace{-1.0em}	
		\label{table:resultCA90'}
		\resizebox{0.80\textwidth}{!}{
		\begin{tabular}{|c|c|c|c|c|c|c|c|c|c|c|}\hline
2 & 3 & 5 & 6 & 9 & 11 & 14 & $18^*$ & 23 & 26 & 29  \\\hline
30 & 33 & 35 & 39 & 41 & $50^*$ & 51 & 53 & 65 & 69 & 74  \\\hline
81 & 83 & 86 & 89 & 90 & 95 & $98^*$ & $99^*$ & 105 & 113 & 119  \\\hline
131 & $134^*$ & 135 & 146 & 155 & 158 & 173 & $174^*$ & 179 & 183 & $186^*$  \\\hline
189 & 191 & $194^*$ & 209 & 210 & 221 & 230 & 231 & 233 & 239 & 243  \\\hline
245 & 251 & 254 & 261 & $270^*$ & 273 & $278^*$ & 281 & 293 & 299 & 303  \\\hline
306 & 309 & 323 & 326 & 329 & 330 & $338^*$ & $350^*$ & $354^*$ & 359 & 371  \\\hline
375 & $378^*$ & 386 & $393^*$ & 398 & $410^*$ & 411 & 413 & $414^*$ & 419 & 426  \\\hline
429 & 431 & $438^*$ & 441 & 443 & 453 & 470 & 473 & 483 & 491 & 495  \\\hline
509 & 515 & 519 & 530 & 531 & 543 & 545 & 554 & 558 & 561 & 575  \\\hline
585 & 593 & $606^*$ & 611 & $614^*$ & 615 & 618 & 629 & 638 & 639 & 641 \\\hline
$645^*$ & $650^*$ & 651  & 653 & 659 & 683 & $686^*$ & $690^*$ & 713 & 719 & 723 \\\hline
725 & 726 & $741^*$ & 743 & 746 & 749 & 755 & 761 & 765 & 771 & 774 \\\hline
779 & $783^*$ & 785 & 791 & $803^*$ & 809 & 810 & 818 & 831 & 833 & 834 \\\hline
846 & 866 & 870 & 873 & 879 & 891 & 893 & 911 & 923 & $930^*$ & 933 \\\hline
935 & 938 & 939 & $950^*$ & 953 & 965 & $974^*$ & 975 & $986^*$ & 989 & 993 \\\hline
998 & 1013 & $1014^*$ & 1019 & $1026^*$ & 1031 & $1034^*$ & 1041 & 1043 & 1049 & 1055 \\\hline
1065 & 1070 & 1103 & $1106^*$ & $1110^*$ & 1118 & 1119 & 1121 & 1133 & $1134^*$ & 1146 \\\hline
1154 & 1155 & 1166 & 1169 & $1178^*$ & 1185 & 1194 & & & & \\\hline
\end{tabular}
		}
	\end{center}
	%\vspace{-1.5em}
\end{table}

The procedures are heavily dependent on the theorems (Theroem \ref{property:CA90'strategy} and \ref{property:CA150'strategy}). However, there are some cases when the procedures cannot detect their non-maximality, but are actually non-maximal. The strategies target to eliminate these cases by adding some filtering methods. Based on filtering, three strategies have been proposed. First two strategies directly use $CA(90')$ and $CA(150')$ respectively, whereas the last strategy uses both of them jointly. Empirical studies show that most of the time, generated CAs are maximal length CAs. Since, for few cases, generated CAs are not of maximal length, the strategies are called as imperfect. %From the above understanding, we report our strategies to generate maximal length CAs (hence primitive polynomials).

Further, to generate primitive polynomial using the {\bf procedure $CA(90')$}, then the success rate of this scheme is above 80\%. Further to improve this success rate, the authors exclude the even $n$. Observe that in Table~\ref{table:resultCA90'}, occurrence of error is higher for the cases of even $n$ than odd $n$. In fact, percentage of primitive polynomials for odd $n$ is higher than that of even $n$.
This fact can also be understood from the theory of polynomials. Note that number of primitive polynomials against a degree ($n$) is dependent on {\em Totient function} $\phi (2^n-1)$ \cite{handbookmath}, which is related to prime factors of $2^n-1$. The number of prime factors of $2^n-1$ is less when $n$ is odd, and this fact actually implies that the number of primitive polynomial against odd $n$ is higher than that of even $n$. In the case of $CA(90')$, we can improve the success rate upto 95.41\%. From the above understanding, the authors reported the first strategy. In similar fashion, the authors developed the second startegy followed by $CA(150')$.\\

%Table~\ref{table:strategy1} extracts only the odd sizes from Table~\ref{table:resultCA90'}. Here, we get 125 maximal length CAs from a set of 131 CAs. That is, only for six $n$, the CAs are not of maximal length (for $n=$ 99, 393, 645, 741, 783 and 803). If we claim all such CAs ($n$ is odd and {\bf procedure $CA(90')$} returns ``True'') as maximal length CAs, then the claim is correct in 95.41\% cases. From the above understanding, we report our first strategy to generate maximal length CAs. In similar fashion with $CA(150')$, we develop our second startegy. And the last startegy uses both ($CA(90')$ and $CA(150')$) of them jointly.\\

\noindent {\bf STRATEGY 1}: Declare $CA(90')$ of size $n$ as maximal length CA if {\bf procedure $CA(90')$} returns ``True'' for $n$ and odd $n$.\\

\noindent {\bf STRATEGY 2}: Declare $CA(150')$ of size $n$ as maximal length CA if {\bf procedure $CA(150')$} returns ``True'' for $n$ and odd $n$.\\

The authors explored the relationship between $CA(90')$ and $CA(150')$ to develop the next strategy. These two CAs of size $n$ are {\em conjugate} to each other and in Ref. \cite{JCA19} the authors published an interesting results about conjugate CAs. That is, the characteristic polynomial of $CA(90')$ with size $n$ is irreducible iff that of $CA(150')$ of same size is irreducible. These facts indicate that there are some cases where both procedures may help to detect non-maximality (one can compare the Table~\ref{table:resultCA90'} and Table~\ref{table:resultCA150'} to get an overview). 
For example, {\bf procedure $CA(90')$} and {\bf procedure $CA(150')$} both return ``True'' when $n=11$, and the characteristic polynomials of both CAs are primitive. On the other hand, for $n=15$, {\bf procedure $CA(150')$} returns ``True'' but {\bf procedure $CA(90')$} returns ``False'', and the characteristic polynomials of these CAs are reducible. Which means, there are some cases where false positive of $CA(90')$ in guessing maximality can be filtered by {\bf procedure $CA(150')$}. \\

\noindent {\bf STRATEGY 3}: Declare $CA(90')$ of size $n$ as maximal length CAs if {\bf procedure $CA(90')$} and {\bf procedure $CA(150')$} both return ``True'' for $n$.\footnote{It is interesting to note that if one interchanges the role of $CA(90')$ and $CA(150')$ in this strategy, then she can get similar result.} \\

\begin{table}[h]
	\setlength{\tabcolsep}{5.5pt}
	\begin{center}	
		\caption{List of sizes when procedure of $CA(150')$ returns ``True''. *-marked sizes are false positive} %for non-maximal length CAs	
		\label{table:resultCA150'}
		\resizebox{0.80\textwidth}{!}{
		\begin{tabular}{|c|c|c|c|c|c|c|c|c|c|c|}\hline
2 & 3 & 5 & 9 & 11 & 14 & 23 & 26 & 29 & 35 & 39 \\\hline
41 & $50^*$ & 53 & 65 & 69 & 74 & 81 & 83 & 86 & 89 & 95  \\\hline
$98^*$ & $99^*$ & 105 & 113 & 119 & 131 & $134^*$ & 146 & 155 & 158 & 173\\\hline
179 & 189 & 191 & $194^*$ & 209 & 221 & 230 & 231 & 233 & 239 & 243 \\\hline
251 & 254 & $278^*$ & 281 & 293 & 299 & 303 & 323 & 326 & 329 & $338^*$  \\\hline
$350^*$ & 359 & 371 & 375 & 386 & 398 & $410^*$ & 411 & 413 & 419 & 429  \\\hline
431 & 443 & 453 & 470 & 473 & 491 & 509  & 515 & 519 & 530 & 531  \\\hline
543 & 554 & 561 & 575 & 593 & $614^*$ & 615 & 629 & 638 & 639 & 641\\\hline
$645^*$ & $650^*$ & 653 & 659 & 683 & $686^*$ & 713 & 719 & 723 & 725 & $741^*$\\\hline
743 & 746 & 749 & 761 & 779 & $783^*$ & 785 & $803^*$ & 809 & 818 & 831\\\hline
833 & 866 & 873 & 893 & 911 & 923 & $950^*$ & 953 & 965 & $974^*$ & 975\\\hline
$986^*$ & 989 & 993 & 998 & 1013 & 1019 & 1031 & $1034^*$ & 1041 & 1043 & 1049\\\hline
1070 & 1103 & $1106^*$ & 1118 & 1119 & 1121 & 1133 & 1154 & 1155 & 1166 & 1169 \\\hline
$1178^*$ & 1185 & & & & & & & & & \\\hline		

\end{tabular}
		}
	\end{center}
\end{table}

%As like the result of {\bf procedure $CA(90')$}, the experimental results of these three strategies are shown in Table~\ref{table:strategy1}, Table~\ref{table:strategy2} and Table~\ref{table:mergeresult} (see Appendix, page~\pageref{table:strategy1}). 

The results of above strategies can be find in Ref. \cite{Adak-TCS-2021}. It has been conjectured that for Sophie Germain Primes, the strategies always generate maximal length CAs.

\begin{conjecture}
\label{conjecture-1}
Characteristic polynomials of $CA(90')$ and $CA(150')$ of size $n$ are primitive over GF(2) if $n$ and $2n+1$ both are primes (i.e. $n$ is a \textbf{Sophie Germain prime}) \cite{Adak-TCS-2021}.
\end{conjecture} 

%The proof of Conjecture \ref{conjecture-1} is still unknown. So, it is an interesting open question to determine the proof of this conjecture.
%
%\begin{oproblem}
%\label{oproblem:conjecture}
%The proof of Conjecure \ref{conjecture-1}.
%\end{oproblem}

\section{Complemented Maximal Length CA}
\label{section:complemented}

%Previous defined, if at least one rule is complemented and rest are linear, then the CA is complemented. As like linear CA, complemented CA is another special case. Here, we report the complete characterization of complemented maximal length CA. An analytical framework is developed to characterize the cyclic vector subspaces of a complemented CA that can be derived from careful analysis of the vector subspaces covered by the linear CA. In \cite{NiloyPhD,NiloyCycle}, a schemed is proposed to explore the complemented CA structures having different configuration transition behaviour than that of its linear CA counterpart.

%Linear maximal length CAs are studied so far. Advantage of linear CAs is, they can be characterized by matrix algebra. As like linear CAs, we can characterize complemented CAs by matrix algebra (see Section~\ref{subsection:matrix}). The difference between linear and complemented CA to represent by matrix algebra is mainly in inversion vector ($F$). The inversion vector $F$ of a linear/complemented CA is defined as - 

Linear maximal length CAs have been studied so far. Advantage of linear CAs is, they can be characterized by matrix algebra. Matrix algebra, however, can be extended to characterize complemented CAs by introducing an inversion vector ($F$) \cite{Datta,kolintc,KolinPhd}. The inversion vector $F$ of a complemented CA is defined as -
\[
F_i=\left\{ 
\begin{array}{cl}
1, & \mbox{if the next state of the $i^{th}$ cell is resulted from inversion (complement)}\\
0, & \mbox{otherwise (linear).}\\
\end{array} 
\right.
\]
If $c_t$ represents the configuration of the CA at $t^{th}$ instant of time, then the next configuration - that is, the configuration at $(t+1)^{th}$ time instant, is given by 
\begin{equation}
\label{equation:nextstate}
c_{(t+1)} = T.c_t + F \Rightarrow c_{(t+p)} = T^p.c_t + (I+T+T^2+\cdots + T^{p-1})F,
\end{equation}
where $c_{(t+p)}$ is the configuration of CA at $(t+p)^{th}$ instant of time. For an $n$-cell CA, $F$ is the $n$ bit inversion vector with its $i^{th}$ ($0\leq i\leq n-1$) bit as 1, if complemented rule is applied on the $i^{th}$ cell; 0 if linear rule is applied. Linear CA is the special case where the inversion vector $F$ is an all 0s vector, the next state function for the linear CA gets simplified to 
\begin{equation}
\label{equation:nextstate1}
c_{(t+1)} = T.c_t \Rightarrow c_{(t+p)} = T^p.c_t
\end{equation}

\subsection{Vector Space Theoretic Analysis of Complemented CA}
\label{subsection:ACA}

A complemented CA is represented by the characteristic matrix $T$ and the non-zero inversion vector $F$. The complemented CA exhibits more varieties in cyclic behaviour than that of linear CA. %In this section, we also highlight the variation of complemented CA cycle structure with that of its linear counterpart. 
A complemented CA is reversible if and only if its linear counterpart is a reversible CA \cite{ppc1}. It signifies that the {\em cycle structure} \cite{Adak-JCA16,ACRI14,Supreeti-JCA19} of a complemented CA can be figured out from the analysis of configuration transition behaviour of the corresponding linear CA. Cycle structure of a CA represents number of cycles present in the CA along with their lengths. To illustrate the cycle structure of a CA, let us take an example. Suppose, a CA has 4 cycles of length 1, 2 cycles of length 2, 4 cycles of length 3 and 2 cycles of length 6. Then, the cycle structure of the CA is denoted as $CS=[4(1),2(2),4(3),2(6)]$. Moreover, the concept of null space and its relationship with cycle length of a linear CA can be employed for further analysis of cycle structure of a complemented CA. These concepts are reproduced below from \cite{Datta,kolintc}.

\begin{defnn}
\label{definition:nullspace}
The null space of a matrix ($T$) consists of all such vectors that can be transformed to the all-zero vector when pre-multiplied by the $T$.
\end{defnn}

\begin{theorem}
\label{theorem:submultiple}
If a linear CA, represented by $T$, has a cycle of length $k$, then the cardinality of null space of $(T^k + I)$ denotes the number of configurations forming cycles of length $k$ or sub-multiple of $k$ \cite{KolinPhd}.
\end{theorem}

\begin{theorem}
\label{theorem:commonfactor}
If a linear CA is represented by $T$ and for any configuration $\chi\neq 0$, $g(T)\cdot\chi = 0$, then $g(x)$ and the characteristic polynomial $f(x)$ of $T$ have a common factor $h(x)$ \cite{Datta}.
\end{theorem}

The analysis of a complemented CA's configuration transition behaviour, reported in \cite{NiloyCycle}, targets solutions to the following issues:
\begin{itemize}
\item[I.] To check whether a particular cycle of length $k$ is present in the cycle structure of a complemented CA.
\item[II.] To find the special class of complemented CA $(C')$ with the cycle structure as that of its linear counterpart $C$, irrespective of its inversion vector $F$.
\item[III.] To identify the class of $C'$ whose cycle structure differs from that of $C$ as well as the properties of $F$ vectors which impart this difference.
\item[IV.] To compute the cycle structure of $C'$.
\end{itemize}

%Now, we have to identify of a cycle of length $k$ in complemented CA cycle structure. 
The following theorem enables us to determine whether a cycle of length $k$ exists in the cycle structure of a complemented CA \cite{NiloyCycle}.

\begin{theorem}
\label{theorem:lengthk}
In a complemented CA with characteristic matrix $T$ and inversion vector $F$, a cycle of length $k$ exists if \\
$rank([T^k+I]) = rank([T^k+I, \mathcal{F}])$, where $\mathcal{F}=[I+T+T^2+\cdots +T^{k-1}]F$
\end{theorem}

\begin{proof}
Let $\chi$ be a configuration that falls on a cycle of length $k$ in a complemented CA $(C')$. Hence, as per Equation~\ref{equation:nextstate},\\
\hspace*{5em} $\chi = [I+T+T^2+\cdots +T^{k-1}]F + T^k\cdot \chi$\\
It can be written as
\begin{equation}
\label{equation:eq1}
[T^k+I]\cdot \chi = \mathcal{F},\quad where\quad \mathcal{F}=[I+T+T^2+\cdots +T^{k-1}]F
\end{equation}
If a cycle of length $k$ is to exist in $C'$, Equation~\ref{equation:eq1} should be consistent. The condition for consistency is  
\begin{equation}
\label{equation:eq2}
rank([T^k+I]) = rank([T^k+I,\mathcal{F}])
\end{equation}
Hence the proof.
\end{proof}

\subsection{Complemented CA with cycle structure identical to its linear counterpart}
\label{subsection:identical}

The complemented maximal length CAs can be derived from linear maximal length CAs. Against an $n$-cell linear maximal length CA, there are $2^n-1$ complemented maximal length CAs. Using linear maximal length CAs, so, we can get complemented maximal length CAs. To establish the fact theoretically, we have to identify the complemented CA that has {\em cycle structure} identical to that of a linear CA irrespective of its inversion vector $F$. %Cycle structure of a CA represents number of cycles present in the CA along with their lengths. To illustrate the cycle structure of a CA, let us take an example. Suppose, a CA has 4 cycles of length 1, 2 cycles of length 2, 4 cycles of length 3 and 2 cycles of length 6. Then, the complete cycle structure is denoted as $CS=[4(1),2(2),4(3),2(6)]$. %Following theorem establishes the fact of identical cycle structure.

Theorem~\ref{theorem:lengthk} guides identification of the complemented CA $(C')$ that has cycle structure identical to that of a linear CA $(C)$ irrespective of its inversion vector $F$. The following theorem defines the class of $C'$.

\begin{theorem}
\label{theorem:identicalcycle}
The cycle structures of complemented CA $(C')$ and linear CA $(C)$ are identical if the characteristic polynomial $f(x)$ of $T$ matrix of the complemented CA does not have a factor $(x+1)$ \cite{NiloyCycle}.
\end{theorem}

\begin{proof}
Let $k$ be the length of a cycle of the linear CA $(C)$ with characteristic matrix $T$, characteristic polynomial $f(x)$ of which does not have a factor $(x+1)$. To have a cycle of length $k$ in the corresponding complemented CA $(C')$, Equation~\ref{equation:eq2} has to be satisfied.

The number of vectors, forming a cycle of length $k$ or sub-multiple of $k$, in the complemented CA or linear CA are derived from the computation of null space of $\alpha_1 = (T^k + I)$ \cite{kolintc}.\\
Let $(x^k+1) = g(x)\cdot \phi_c(x)$; where $\phi_c(x)$ is the largest factor of the characteristic polynomial $f(x)$. Therefore, $g(x)$ and $f(x)$ don't have any common factor. Hence, for each configuration $\chi$ with $(T^k+I)\cdot\chi = 0$, there is a corresponding unique configuration $\overline{\chi}$, where $\phi_c(T)\cdot\overline{\chi} = 0$ (from Theorem~\ref{theorem:commonfactor}). Hence, the cardinality of null space of $\alpha_2=\phi_c(T)$ is same as $\alpha_1$.

Since $f(x)$ does not have a factor $(x+1)$ and $x^k+1=(x+1)\cdot (1+x+x^2+\cdots +x^{k-1})$, therefore, the cardinality of the null space of $\alpha_3=[I+T+T^2\cdots T^{k-1}]$ is same as that of $\alpha_1$ and $\alpha_2$. Hence,\\
\hspace*{6em} $rank(T^k+I) = rank(T^k+I, I+T+T^2\cdots T^{k-1})$\\
that is,\\
\hspace*{8em} $rank(T^k+I) = rank(T^k+I, \mathcal{F})$\\
directly follows for any $F$.

Therefore, all the cycle lengths of $C$ also exist in $C'$. Since the number of vectors forming each cycle length is same for both of $C$ and $C'$ (directly derived from the cardinality of null space), the cycle structures for both the CAs are identical. Hence the proof.
\end{proof}

\begin{example}
This example illustrates the result of Theorem~\ref{theorem:identicalcycle}. Let us consider a $5$-cell complemented CA $(C')$ with characteristic matrix $T$ and the inversion vector $F\neq 0$, where
\[
T=
    \begin{bmatrix}
        1  & 1   & 0  & 0   & 0 \\
        1  & 0   & 1  & 0   & 0 \\
        0  & 0   & 1  & 1   & 0 \\
        0  & 0   & 1  & 1   & 1 \\
        0  & 0   & 0  & 1   & 0 \\
    \end{bmatrix}
\]

The characteristic polynomial of $T$ is $(x^3+x+1)(x^2+x+1)$. The cycle structure of the corresponding linear CA $(C)$ is $[1(1),1(3),1(7),1(21)]$. The complemented CA as per Theorem~\ref{theorem:identicalcycle}, has the identical cycle structure as that of the linear CA, irrespective of $F$.\\
Let us consider a particular cycle of length 7. As per the theorem, we compute $\alpha_1=T^7+I$, $\alpha_2=T^3+T+1$ and $\alpha_3=T^6+T^5+T^4+T^3+T^2+T+I$: 
\[
T^7+I=
    \begin{bmatrix}
        0  & 1   & 0  & 0   & 0 \\
        1  & 1   & 0  & 1   & 0 \\
        0  & 0   & 0  & 0   & 0 \\
        0  & 0   & 0  & 0   & 0 \\
        0  & 0   & 0  & 0   & 0 \\
    \end{bmatrix}
\]

\[
T^3+T+1=
    \begin{bmatrix}
        1  & 1   & 0  & 1   & 0 \\
        1  & 0   & 0  & 0   & 1 \\
        0  & 0   & 0  & 0   & 0 \\
        0  & 0   & 0  & 0   & 0 \\
        0  & 0   & 0  & 0   & 0 \\
    \end{bmatrix}
\]

\[
T^6+T^5+T^4+T^3+T^2+T+I=
    \begin{bmatrix}
        1  & 0   & 0  & 0   & 1 \\
        1  & 1   & 0  & 1   & 1 \\
        0  & 0   & 0  & 0   & 0 \\
        0  & 0   & 0  & 0   & 0 \\
        0  & 0   & 0  & 0   & 0 \\
    \end{bmatrix}
\]

The cardinality of null space for all the three matrices (with $3^{rd}$, $4^{th}$ and $5^{th}$ rows as all zeros) is 8. Therefore, $rank(\alpha_1)=rank(\alpha_2,\alpha_3)$ and $rank(\alpha_1)=rank(\alpha_2,\mathcal{F})$; $\mathcal{F}=(I+T+T^2+T^3+T^4+T^5+T^6)\cdot F$. That is, both the $C'$ and $C$ have cycles of length 7. The number of configurations in the cycles of length 7 or sub-multiple of 7 (here it is 1) is 8. Therefore, the number of cycles of length 7 is 1 (as one configuration forms a single cycle). The complete cycle structure $CS'=[1(1),1(3),1(7),1(21)]$ is same as that of $C$.
\end{example}

The characteristic polynomials of maximal length CAs are primitive; that means polynomials are irreducible. Irreducible polynomials does not have any factor. Therefore, we can apply Theorem~\ref{theorem:identicalcycle} on maximal length CA. Hence complemented maximal length CAs can be obtained by exploiting the linear maximal length CAs. Further, we can extract the following corollary from Theorem~\ref{theorem:identicalcycle}.

\begin{corollary}
\label{corollary:complemented}
For every linear maximal length CA, we can get complemented maximal length CAs.
\end{corollary}

To get complemented maximal length CAs, we have to use the complemented rules of rule 90 and rule 150. The respective complemented rules are 165 and 105. Only using these two rules, we can design complemented maximal length CAs from a linear maximal length CA \cite{NiloyPhD,NiloyCycle}. According to the above discussion, if any 90 or 150 in the rule vector of a linear maximal length CA is replaced by its complement (that is, 165 or 105), the new CA remains as a maximal length CA. For example, the CA (90, 150, 90, 150) is a linear maximal length CA, and so (165, 105, 90, 150) (Figure~\ref{figure:st-165-105-90-150}), (90, 150, 165, 150), etc. are also maximal length CAs. Any number of rules in a rule vector of a linear maximal length CA can be replaced by their complements to get a new complemented maximal length CA. Therefore, against an $n$-cell linear maximal length CAs, there are $2^n-1$ number of complemented CAs which are maximal length.

\begin{figure}
	\centering
	\includegraphics[width= 4.3in, height = 1.7in]{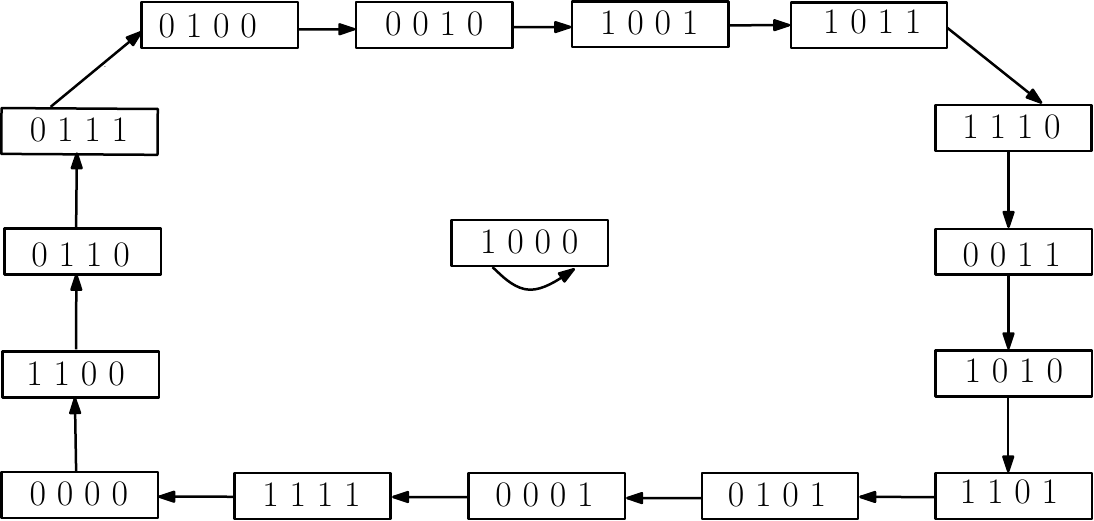}
	\caption{Transition diagram of CA $(165, 105, 90, 150)$}
	%\vspace{-0.5em}  
	\label{figure:st-165-105-90-150}
\end{figure}

\section{Non-linearity in Maximal Length CAs}
\label{section:nonlinear}

%Provisionally, maximal length CAs are linear and complemented and have been used in different applications. Further, linear maximal length CAs suffer from some drawbacks. Firstly, the availability of $n$-degree primitive polynomials is limited. Besides, linear maximal length sequences are not secure. So, there is necessity of a construction that can provide both non-linearity and as well as maximal length sequence. Basically, the question is can we add some non-linearity in maximal length CA? And the answer is yes, that means we can get maximal length CA where the CA is non-linear. This section is based on 1-dimensional 2-state 3-neighborhood CA.

The maximal length CAs, studied so far, are linear and complemented, where all the rules in rule vectors are either linear or complemented rules. Out of 256 rules, there are eight linear rules and eight complemented rules. There are 240 rules which are non-linear. If any one rule of a CA is non-linear, then the CA is non-linear.

There have been some efforts to introduce non-linearity in maximal length CAs~\cite{ACRIDRC10,jcaDasR11,ACRIDRC14,ACRI18-Adak}. The technique referred in~\cite{ACRIDRC10} manipulates the number of clock cycles, based on inputs, in a maximal length CA. This method becomes unsynchronized for different inputs. An efficient technique is devised for generating non-linear maximal length CA from linear maximal length CA by injecting non-linearity in different cell positions \cite{ACRIDRC14}. The effect of the non-linearity can be propagated among multiple cells by shifting the non-linear function. However, such a non-linear CA increases neighborhood dependencies of some cells; for example the 3-neighborhood dependency, which is traditionally used by the maximal length CAs, may increase to 5-neighborhood dependency. %Hence the question remains: do there exist non-linear maximal length CAs within 3-neighborhood dependency? Still there is no algorithm which can generate non-linear maximal length CA within 3-neighborhood dependency. 
%Further, by ``non-linear CA'' we shall mean ``non-linear non-uniform CA''.

%There have been some researchers to introduce non-linearity in maximal length CAs~\cite{ACRIDRC10,jcaDasR11,ACRIDRC14,ACRI18-Adak}. The technique referred in~\cite{ACRIDRC10} manipulates the number of clock cycles, based on inputs, in a maximal length CA. This method becomes unsynchronized for different inputs. An efficient technique~\cite{ACRIDRC14} is devised for generating non-linear maximal length CA from linear maximal length CA by injecting non-linearity in different cell positions. The effect of the non-linearity can be propagated among multiple cells by shifting the non-linear function. However, it incurs increasing neighborhood dependency. For optimal design, the construction of non-linear maximal length CA limits upto 5 neighborhood. This motivates us to figure out if there exists a non-linear maximal length CA without exceeding the neighborhood dependency. Consequently, in~\cite{ACRI18-Adak} we classifies the rules based on some parameter values, which determine the degree of dependence of a cell on its neighbors. Then to construct the CA, we select rules randomly by assigning more priority to those rules which have higher values for those parameters. Our experiments show that non-linear reversible CAs constructed in this manner, can have cycles of maximal length. Further, by ``non-linear CA'' we shall mean ``non-linear non-uniform CA''.

%\subsection{Synthesizing Non-linear Maximal Length CA}
%\label{subsection:snthesize-non-linear}

In~\cite{ACRIDRC14}, the authors have provided an algorithm to construct non-linear maximal length CA from a given linear maximal length CA. But, there is no proof behind the correctness of the algorithm. The algorithm uses a {\em neighbor set} $\mathbf{N}(i)$, for each cell $i$, which is the set of cells on which the $i^{th}$ cell is dependent on each iteration. For a three neighborhood CA, $\mathbf{N}(i) = \{i-1, i, i+1\}$ for any $i$. Let $f_i$ be the feedback function of the $i^{th}$ cell, then define an operation called shifting as followed. 

\begin{defnn}
\label{defintition:shifting}
The one cell shifting operation, denoted by $f_i\xrightarrow[]{\text{P}}f_{i\pm 1}$ moves a set of algebraic normal form monomials $P$ from $i^{th}$ cell of a non-linear CA to all the cells from $(i-1)$ to $(i+1)^{th}$ cell, according to the dependency of the affected cells upon the $i^{th}$ cell. Each variable in $P$ is changed by their previous state. Similarly, a $k$ cell shifting is obtained by applying the one cell shifting operation for $k$ times upon the initial non-linear CA and symbolized as $f_i\xrightarrow[]{\text{P}}f_{i\pm k}$.
\end{defnn}

The details of the scheme have been presented in~\cite{ACRIDRC14}. Before injecting non-linearity to the required position of a given linear maximal length CA, a shifting operation is applied to the corresponding non-linear function at the required position to retain the maximal length. The procedure to synthesize a non-linear maximal length CA from a linear maximal length CA is depicted in following algorithm. % followed by an example.
In Ref. \cite{ACRIDRC14}, one can find the details of the algorithm and example.

\begin{algorithm}
\caption{Synthesize Non-linear Maximal Length CA}
\label{algorithm:decidenonlinearmaximality}
\small
\hspace*{\algorithmicindent} \textbf{Input} A $n$-cell linear maximal length CA ($\mathcal{R}$), position $j$ to inject non-linearity\\
\hspace*{\algorithmicindent} \textbf{Output} A non-linear maximal length CA ($\mathcal{R}'$)

\begin{algorithmic}[1]
	
\STATE $\mathcal{R}'\leftarrow\mathcal{R}$
\STATE Let $\mathcal{R}' = \{f_{n-1}, \cdots, f_0\}$
\STATE $\mathcal{X} \subset F : \forall x\in\mathcal{X}, x\notin \mathbf{N}(j)$ \hfill$\triangleright$ {select a subset from $F$}
\STATE $P\leftarrow \mathbf{f}_{\mathbf{N}}(\mathcal{X})$ \hfill$\triangleright$ {$\mathbf{f}_{\mathbf{N}}$ is a non-linear function}
\STATE $f_j\leftarrow f_j \oplus P$
\STATE ($f_j\xrightarrow[]{\text{P}}f_{j+1}$) \hfill$\triangleright$ {applying shifting operation}
\STATE $f_j\leftarrow f_j\oplus P$
\STATE {\bf return} $\mathcal{R}'$

\end{algorithmic}
\end{algorithm}

%\begin{example}
%Let us consider a linear maximal length CA synthesized from the primitive polynomial $x^5+x^2+1$ to obtain the rule set $\mathcal{R}' = \{f_{0}, \cdots, f_{4}\}$ such that:\\
%\hspace*{10em} $f_0 = x_1\oplus x_0$ \\
%\hspace*{10em} $f_1 = x_2\oplus x_1\oplus x_0$ \\
%\hspace*{10em} $f_2 = x_3\oplus x_2\oplus x_1$ \\
%\hspace*{10em} $f_3 = x_4\oplus x_3\oplus x_2$ \\
%\hspace*{10em} $f_4 = x_3$ \\
%We inject the non-linearity into the bit $j=2$. Let, $\mathcal{X} = \{0,4\}$ and $\mathbf{f}_{\mathbf{N}}(a,b) = a \& b$ where $a$ and $b$ are two boolean one bit variables. So, according to the algorithm, the function $f_j$ would be updated like $f_j = f_j\oplus \mathbf{f}_{\mathbf{N}}(x_0\&x_4)$ $= f_j\oplus (x_0\&x_4)$. Now, applying the shifting operation $f_j\xrightarrow[]{\text{$x_{0}\& x_{4}$}}f_{j+1}$, the whole rule set is : \\
%\hspace*{8em} $f_0 = x_1\oplus x_0$ \\
%\hspace*{8em} $f_1 = x_2\oplus x_1\oplus x_0 \oplus ((x_0\oplus x_1)\&x_3)$ \\
%\hspace*{8em} $f_2 = x_3\oplus x_2\oplus x_1 \oplus (x_0\&x_4) \oplus ((x_0\oplus x_1)\&x_3)$ \\
%\hspace*{8em} $f_3 = x_4\oplus x_3\oplus x_2 \oplus ((x_0\oplus x_1)\&x_3)$ \\
%\hspace*{8em} $f_4 = x_3$ \\
%This rule set produces maximal length CA which is non-linear.
%\end{example}

However, the question remains - without exceeding the neighborhood dependency can we get non-linear maximal length CAs. In \cite{ACRI18-Adak}, the authors generate non-linear maximal length CAs from linear maximal length CAs in the 3-neighborhood dependency. However, no formal procedure was presented in Ref.\cite{ACRI18-Adak} to study the generation of non-linear maximal length CA. This is mainly due to the challenge in applying mathematical tools  like linear algebra, primitive polynomial etc. to analyze the properties of non-linear CAs. In recent one \cite{Ascat2023-Sumit}, the authors address this issue by introduced few concepts that help us to identify non-linear CAs that can potentially generate maximal length CAs.

%However, the question remains - without exceeding the neighborhood dependency can we get non-linear maximal length CAs. In a preliminary study \cite{ACRI18-Adak}, the authors answered this question in an affirmative way. We had hinted that it is also possible to generate a non-linear maximal length CA in 3-neighborhood dependency. However, no formal procedure was presented in Ref.\cite{ACRI18-Adak} to study the generation of non-linear maximal length CA. This is mainly due to the challenge in applying mathematical tools  like linear algebra, primitive polynomial etc. to analyze the properties of non-linear CAs. In our recent research \cite{Adak-Thesis}, we address this issue by introduced few concepts that help us to identify non-linear CAs that can potentially generate maximal length CAs.

At first, they have investigated some essential properties of non-linear maximal length CAs. For example, if a (non-uniform) CA  has a {\em blocking word}, it cannot be a maximal length CA. Further, a maximal length CA should have only one single length cycle. Further, the number of non-linear rules cannot be more than $n/2$ for an $n$-cell non-linear maximal length CAs. These properties help in straightway identifying some of the non-linear CAs that won't generate maximal length cycles, without even computing their cycle lengths.

The main approach is to systematically explore if mathematical tools of matrix and linear algebra can be applied with minor modifications for the case of non-linear maximal length CAs. To this effect, they have studied if a non-linear maximum length CA can be generated by starting from a linear CA and then changing a very few rules of it to non-linear. Introduced the notion of {\em isomorphism}, which acts as a  similarity measure between the cycle structures of two CAs. In particular, if a non-linear CA is isomorphic (see Definition~\ref{isomorphic}) to a linear maximal length CA, then the former is also a maximal length CA. They developed theories to determine the conditions under which a non-linear CA can be isomorphic to linear maximal length CA. Finally, an algorithm has been developed for the same (details of the algorithm in Ref.\cite{Ascat2023-Sumit}).

\begin{defnn}
\label{isomorphic}
Let $G_1$ and $G_2$ be two CAs of length $n$. We say that $G_2$ is isomorphic to $G_1$ iff - 
\begin{enumerate}
\item transition diagram of $G_2$ is isomorphic~\footnote{Two graphs which contain the same number of vertices connected in the same way are said to be isomorphic.} to that of $G_1$.
\item there exist a non-empty set of configuration $C$ such that $\forall x\in C$, $G_1(x) = G_2(x)$. 
\end{enumerate}
\end{defnn}

Table~\ref{table:nonlinearresults} notes few non-linear maximal length CAs which we get using the above method. The second column of Table~\ref{table:nonlinearresults} shows the linear maximal length CAs, whereas the third and fourth columns show at which position(s) of the linear CAs the non-linear rule(s) is (are) put, and the the non-linear rule(s) respectively. The linear rule which is replaced are marked in bold face. For example, positions of bold faced rules for size = 24 are 5 and 21. In the corresponding non-linear maximal length CA ($G_2$), rules 30 and 86 are placed in there positions respectively. 

{\small 
\begin{table}
\caption{Linear maximal length CAs to non-linear maximal length CAs. Rule with bold face is replaced by non-linear one}
\label{table:nonlinearresults}
\resizebox{\textwidth}{!}{\begin{tabular}{|c|c|c|c|} \hline

   &  & \multicolumn{2}{|c|}{$G_2$} \\ \cline{3-4}
  \textbf{Size} & $G_1$ & \textbf{Positions} & \textbf{Non-linear} \\
  & & & \textbf{rules} \\\hline
  
  4 & (6,90,\textbf{150},80) & 2 & 89 \\\hline
  5 & (6,150,\textbf{150},150,80) & 2 & 75 \\\hline
  6 & (6,90,\textbf{90},150,90,20) & 2 & 86\\\hline
  7 & (10,90,150,\textbf{90},150,90,20) & 3 & 169 \\\hline
  8 & (6,150,150,90,150,\textbf{150},150,20) & 5 & 154 \\\hline
  9 & (10,\textbf{150},150,150,90,\textbf{90},90,90,20) & 1,5 & 30,58  \\ \hline
  10 & (10,150,150,\textbf{90},90,90,90,150,150,20) & 3 & 101 \\\hline
  11 & (6,90,150,150,\textbf{150},90,150,90,90,150,20) & 4 & 86  \\\hline
  12 & (10,90,\textbf{150},150,150,150,90,\textbf{90},90,150,150,20) & 2,7  & 86,149  \\\hline
  13 & (6,90,150,\textbf{90},\textbf{90},90,90,150,90,150,90,150,20) & 3,4 & 210,101  \\\hline
  14 & (10,150,90,\textbf{150},\textbf{90},150,90,90,90,90,90,90,90,20) & 3,4 & 45,169 \\\hline
  15 & (10,150,\textbf{90},90,90,150,150,90,150,90,150,90,150,150,20) & 2 & 53 \\\hline
  16 & (6,150,90,90,90,90,90,90,90,\textbf{150},90,90,150,90,90,20) & 9 & 154 \\\hline
  17 & (6,90,90,90,\textbf{90},150,90,150,90,150,90,150,90,150,90,90,20) & 4 & 101 \\\hline
  18 & (6,150,150,90,\textbf{150},90,150,90,150,90,90,150,90,90,150,90,150,80) & 4 & 225 \\\hline
  19 & (10,90,150,150,\textbf{90},\textbf{90},150,90,90,90,150,150,150,150,90,150,90,90,20) & 4,5 & 30,154  \\\hline
  20 & (10,90,150,150,150,150,150,90,90,90,90,150,150,\textbf{150},90,90,90,150,150,20) & 13 & 86\\\hline
\end{tabular}}
\end{table}}

It is also an interesting open question to determine a function which can generate non-linear maximal length CAs from a linear maximal length CA in polynomial time.
\begin{oproblem}
\label{oproblem:analysis}
Generation of non-linear maximal length CAs from a linear maximal length CA in polynomial time.
%Develop an algorithm which can find a given CA as maximal length in polynomial time.
\end{oproblem}

\section{Maximality over GF($q$)}
\label{section:MLCAGF(q)}

The topics discussed till now are related to the binary CAs. In this section, we extend the ideas to the CAs defined over the finite filed GF($q$), $q$ is a prime. %Here, $q$ is the number of states of a CA cell. 
The relevant background of linear machines over finite fields can be found in~\cite{McEliece87,Lidl1986}. The cells can be in any of the states of $\mathcal{S} = \{0,1,\cdots, q-1\}$. A configuration of the CA is a mapping ${c}$: $\mathcal{L}$ $\rightarrow$ $\mathcal{S}$. Let $\mathcal{C}$ be the set of all possible configurations of an $n$-cell CA. Then, $|\mathcal{C}|$=$q^n$. Now, we can define the linear rules in $q$-state CA which is similar to Definition~\ref{definition:linearrules}.

%Before that, all the topic which are discussed are belongs to the binary CA. But, in here, we studied the CAs which are particular linear machines over the finite filed GF($q$), $q$ is a prime. The relevant background of linear machines over finite fields can be found in~\cite{McEliece87,Lidl1986}. That means, $q$ is the number of states of a CA cell. The cells can be in any of the states $\mathcal{S} = \{0,1,\cdots, q-1\}$. Each cell $i$ of the CA updates its state depending on the present state of itself, and its left and right neighbors following a local update rule $R_i:{\mathcal{S}}^3 \rightarrow {\mathcal{S}}$. A configuration of the CA is a mapping ${c}$: $\mathcal{L}$ $\rightarrow$ $\mathcal{S}$ where $\mathcal{L}$ is a 1-dimensional lattice. Let $\mathcal{C}$ be the set of all possible configurations of an $n$-cell CA (that is $|\mathcal{C}|$=$q^n$). Then, a CA is a function $G_n$: $\mathcal{C}$ $\rightarrow$ $\mathcal{C}$, which satisfies the following conditions: $y=G_n(x)$, $x,y \in {\mathcal{C}}$, where $x=(x_i)_{0\leq i\leq {n-1}}$, $y=(y_i)_{0\leq i\leq {n-1}}$ and $y_i=R_i(x_{i-1},x_i,x_{i+1})$. The CAs have \emph{null boundary condition} where left and right neighbors of cell 0 and cell $n-1$ respectively are always in state 0; that is, $y_0=R_0(0,x_1,x_2)$ and $y_{n-1}=R_{n-1}(x_{n-2},x_{n-1},0)$. Now, we can define the linear rules in $q$-state which is similar as Definition~\ref{definition:linearrules}.

\begin{defnn}
\label{definition:linearrule}
A rule ${R}:\mathcal{S}^3\rightarrow \mathcal{S}$ is called \textbf{linear} if $R$ can be expressed as $R(x,y,z)= ax + dy + bz\pmod{q}$, for some $a,b,d\in \mathcal{S}$. 
\end{defnn}

To state a rule, we can use the triplet: ($a,d,b$). Let us represent a rule for cell $i$ as $(a_i,d_i,b_i)$. For a CA in GF($q$), there are in total $q^3$ linear rules. For example, number of linear rules in GF(2), GF(3) and GF(5) are $8$, $27$ and $125$ respectively. In this section, by ``CA'' we shall mean $1$-dimensional $q$-state $3$-neighborhood non-uniform CAs having null boundary condition. The rule vector of an $n$-cell CA is ${\mathcal R}=({\mathcal R_0},~ {\mathcal R_1}, \cdots, {\mathcal R_i}, \cdots, {\mathcal R_{n-1}})$ where cell $i$ follows ${\mathcal R_i}\equiv (a_i,d_i,b_i)$.

%Further, a CA $G_n$ is linear, if every rule $R_i$ of it is linear. Also we need a rule vector ${\mathcal R}=( {\mathcal R_0},~ {\mathcal R_1}, \cdots, {\mathcal R_i}, \cdots, {\mathcal R_{n-1}})$ to define an $n$-cell non-uniform CA, where cell $i$ follows ${\mathcal R_i}$. But, in here, the representation of rules are different from previous. We call these constants $a_i, d_i, b_i$ as triplet and represent the linear rule in terms of this triplet $( a_i, d_i, b_i)$. For a CA in GF(q), there are total $q^3$ number of possible linear rules. For example, number of linear rules for GF(2), $GF(3)$ and $GF(5)$ are $8$, $27$ and $125$ respectively. Further, in this section, by ``CA'' we shall mean $1$-dimensional $q$-state $3$-neighborhood non-uniform CAs having null boundary condition, where each rule of the rule vector is linear and represented by its triplet form.

%As previous, an $n$-cell $3$-neighborhood linear CA can also be expressed by a square matrix ($T$) of dimension $n$ which shows the dependency of a cell on other cells. This matrix is tridiagonal, where the lower, upper and main diagonals depend on the rules (triplet). It can be represented as previous which discussed in Section~\ref{subsection:relation}. However, we get the characteristic polynomial of this matrix is $P(x)=det(xI-T)$ $\pmod q$. The characteristic polynomial in $GF(q)$ can be calculated using the previous recurrence relation (see, Lemma~\ref{lemma:recurrence}), where $\Delta = det(xI-T) \pmod q$. Lets take an example.

As per previous discussion (see Section~\ref{subsection:matrix}), an $n$-cell 3-neighborhood linear CA can also be expressed by a square matrix ($T$) of dimension $n$. The matrix shows the dependency of a cell on other cells. %The matrix is tridiagonal, where the lower, upper and main diagonals depend on the rules. 
Under null boundary condition, it can be represented as the following: 
\[
T=
    \begin{bmatrix}
        d_0 & b_0 & 0 & \cdots & 0 & 0 \\
        a_1 & d_1 & b_1 & \ddots &  & 0 \\
        0 & a_2 & d_2 & \ddots &  & \vdots \\
        \vdots & \ddots & \ddots & \ddots & \ddots & \vdots \\
        0 &  &  & \ddots & d_{n-2} & b_{n-2} \\
        0 & 0 & \cdots & \cdots & a_{n-1} & d_{n-1} \\
    \end{bmatrix}_{n \times n}
\]

The characteristic polynomial of this matrix is $P(x)=det(xI+T)$ $\pmod q$. The characteristic polynomial in GF($q$) can be calculated using the previous recurrence relation (see Lemma~\ref{lemma:recurrence}), where $\Delta = det(xI+T) \pmod q$. 
\begin{equation}
\label{equation:polynomial}
{\Delta} = {\Delta}_n = (x+d_{n-1}){\Delta}_{n-1} + b_{n-2} a_{n-1} {\Delta}_{n-2}
\end{equation}

%Let us take an example.

\begin{example}
\label{Matrix_example}
Let us consider a 3-state 3-cell CA with rule vector ((1,2,1), (2,0,1), (2,0,2)). The characteristic polynomial of this matrix is $det(xI+T) \pmod 3$ = $x^3+x^2+2x+1$.
\[
T=
    \begin{bmatrix}
        2   & 1   & 0   \\
        2   & 0   & 1   \\
        0   & 2   & 0   \\
    \end{bmatrix}_{3 \times 3}
\]
%Similarly, for the 2-state $4$-cell CA of $((1,1,1), (1,0,1), (1,1,1), (1,1,1))$, the matrix is \[
%T''=
%    \begin{bmatrix}
%        1   & 1   & 0   & 0   \\
%        1   & 0   & 1   & 0   \\
%        0   & 1   & 1   & 1    \\
%        0   & 0   & 1   & 1   \\
%    \end{bmatrix}_{4 \times 4}
%\]
%and the characteristic polynomial is $x^4+x^3+1\pmod2$ .
\end{example}

Figure~\ref{figure:transition2} shows a transition diagram for the 3-state 3-cell CA ((1,2,1), (2,0,1), (2,0,2)). Here, the next configuration $c'$ of any configuration $c$ can be calculated by multiplying $c$ with $T$, where the configuration $c$ is represented by an $n\times1$ column vector. The CA of Figure~\ref{figure:transition2} is reversible; in fact it is a maximal length CA.

%However, if a CA has only two cycles in its configuration space where all configurations except $0^n$ (for linear CAs) are part of the same cycle, it is called a \emph{maximal length} CA. Further, we can recall the definition of maximal length CA - a CA is maximal length if all but one configuration of it are in a single cycle.  

\begin{figure}
	\centering
	\includegraphics[width= 4.8in, height = 1.4in]{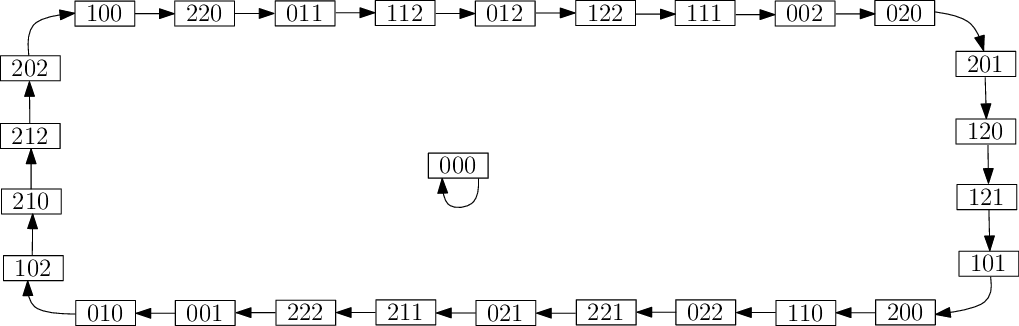}
	\caption{Transition diagram of 3-state CA ((1,2,1), (2,0,1), (2,0,2))}
	%\vspace{-1.0em}  
	\label{figure:transition2}
\end{figure}

%\begin{defnn}
%\label{definition:MLCA}
%An $n$-cell (linear) CA is maximal length if, for a configuration $x\in \mathcal{C}\setminus \{0^n\}$, $x={G_n}^{q^n-1}(x)$, but $x\neq {G_n}^t(x)$ where $1\leq t < q^n-1$.
%\end{defnn}

For a linear maximal length CA over GF($q$), the length of the cycle is $q^n-1$. Theorem~\ref{theorem:mlcagf(q)} relates the maximal length CA and primitive polynomial over GF($q$). The characteristic polynomial of the CA of Figure~\ref{figure:transition2} is $x^3+x^2+2x+1\pmod3$, which is primitive over GF(3).  

%In Theorem~\ref{theorem:mlcagf(q)}, we relate the maximal length CA and primitive polynomial over $GF(q)$ after summarize the works of Cattell et. al.~\cite{Cattell2,CattellTh}. As an example, the characteristic polynomial of Figure~\ref{figure:transition2} is $x^4+x^3+1\pmod2$ which is primitive over $GF(2)$.

\begin{theorem}
\label{theorem:mlcagf(q)}
A $q$-state linear CA is maximal length if and only if the characteristic polynomial of this CA is primitive over GF($q$) \cite{Cattell2}.
\end{theorem}

For some conditions, the characteristic polynomials of CAs over GF($q$) are always reducible. Following theorem is generalization of Theorem~\ref{theorem:reduciblepolynomial} and Theorem~\ref{theorem:palindromicCA}.
 
%which are already proved in previous (see, Theorem~\ref{theorem:reduciblepolynomial} and Theorem~\ref{theorem:palindromicCA}). These constraints are also applicable in $GF(q)$. We described the constraints in the following theorem.

\begin{theorem}
\label{theorem:reducible}
The characteristic polynomial of an $n$-cell CA is always reducible over GF($q$), if one of the following conditions is satisfied -
\begin{enumerate}
\item $a_i=0$ for any $i\in \{1, 2, \cdots, n-1\}$ or $b_j=0$ for any $j\in \{0, 1, \cdots, n-2\}$ of any rule.
\item the CA is palindromic.
\end{enumerate}
\end{theorem}

%\begin{enumerate}
%\item If $a_i=0$ for any $0<i\leq n-1$ or $b_j=0$ for any $0\leq j<n-1$ of any rule (see Theorem~\ref{theorem:reduciblepolynomial}).
%\item If the CA is palindromic (see Theorem~\ref{theorem:palindromicCA}).
%\item If the CA is uniform (see Corollary~\ref{corollary:uniformCA}).
%\end{enumerate} 

%Therefore, for an $n$-cell CA, if any rule $(a_i,d_i,b_i)$ has $a_i=0$ ($0< i \le n-1$) or $b_i=0$ ($0\le i < n-1$), then the characteristic polynomial of this CA is reducible. 

%From Theorem~\ref{theorem:reducible}, we can exclude $2q^2-q$ (two positions of the triplet to hold any of the $q$ number of states considering left-most or right-most position as $0$ and excluding the common $q$ triplets that repeat in the process) number of rules out of the total $q^3$ rules, which can not be part of any rule vector for generating maximal length CAs over $GF(q)$. Hence, $q(q-1)^2$ number of linear rules are participate to generate maximal length CAs over $GF(q)$. 

In \cite{Cattell2}, Cattell and Muzio have provided one minimal cost maximal length CA for each degree upto size 40, for the fields of order 2, 3, 5, 7 and 11. Only rules (1,0,1) and (1,1,1) are used to generate maximal length CAs where minimal of rule (1,1,1) are used.

%Table~\ref{table:minimalcostgf(q)} (see Appendix, page~\pageref{table:minimalcostgf(q)}) contains the results. The table lists those cells that have rule (1,1,1) and rest of the rules are rule (1,0,1). For example, the entry for $n=6$ for GF(5) is 0, 1, 4, which represents the CA -  ((1,1,1), (1,1,1), (1,0,1), (1,0,1), (1,1,1), (1,0,1)).

%In \cite{Cattell2}, Cattell and Muzio have provided one minimal cost maximal length CA for each degree upto size 40, for the fields of order 2, 3, 5, 7 and 11. Only rules (1,0,1) and (1,1,1) are used to generate maximal length CAs. Table~\ref{table:minimalcostgf(q)} (see Appendix, page~\pageref{table:minimalcostgf(q)}) contains the results. The table lists those cells that have rule (1,1,1) and rest of the rules are rule (1,0,1). For example, the entry for $n=6$ for GF(5) is 0, 1, 4, which represents the CA -  ((1,1,1), (1,1,1), (1,0,1), (1,0,1), (1,1,1), (1,0,1)).

Given a polynomial $p(x)$, does there exist a CA with characteristic polynomial $p(x)$? If this question does not hold for all $p(x)$, then does it hold for a subclass of polynomials? Cattell and Muzio have addressed this question for CA over GF(2) \cite{CattellM96} by applying the results of Mesirov and Sweet~\cite{Mesirov1987}. Unfortunately, the approach can not be extended to the CA over GF($q$). But, a large body of empirical evidence suggests the following conjecture for CA over GF($q$). The evidence is obtained by exhaustive enumeration of all CA for various fields and various $n$, finding the characteristic polynomials, and checking if all irreducible polynomials appear.

Further, we can decide it to the new cases where we find the details of our work related to the EDA and MMAS with degree 2 terminology.

%Further, there is a basic question that given a polynomial $p(x)$, does there exist a CA with characteristic polynomial $p(x)$. If this question does not hold for all $p(x)$, does it hold for a natural subclass of polynomials (e.g., irreducible polynomials). Furthermore, also in~\cite{CattellM96}, Cattell and Muzio solved this question for CA over $GF(2)$ by applying the results of Mesirov and Sweet~\cite{Mesirov1987}. Unfortunately, neither approach has a theoretical basis that extends to the general case of $GF(q)$. But, a large body of empirical evidence suggest the existence conjecture for $GF(q)$ (Conjecture~\ref{conjecture:irrpolynomial}). The evidence is obtained by exhaustive enumeration of all CA for various fields and various $n$, calculating the characteristic polynomials, and checking that all irreducible polynomials appear. 

\begin{conjecture}
\label{conjecture:irrpolynomial}
If $p(x)$ is an irreducible polynomial over GF($q$), then there exists a CA with characteristic polynomial $p(x)$~\cite{Cattell2}.
\end{conjecture}

We have synthesized CA from a given irreducible polynomial over GF(2) (see Section~\ref{subsection:CAfrompolynomial}). However, there is no efficient method of synthesizing CAs over GF($q$). So, to find maximal length CAs in GF($q$), the known way is following \cite{Adak-IJMPC2019}.

%In GF($q$), the actual difficulty lies in the reverse process, that means, calculation of CA from a given primitive polynomial. Therefore, to generate maximal length CAs in GF($q$), the only way to deal is following~\cite{Adak-IJMPC2019}.   

\begin{enumerate}
\item For a given size, randomly generate a candidate CA.
\item Find characteristic polynomial of the CA using recurrence relation.
\item Test whether the polynomial is primitive using prime factorization method \cite{Rabin-irreducible}.
\end{enumerate}

In Ref. \cite{Adak-IJMPC2019}, the authors experimentally generated maximal length CAs in GF($q$) using the above approach. Table~\ref{table:primitiveGF} records a demo of this process to find one maximal length CA of size 10, 15 and 20 over GF(2) to GF(19) where the last column shows the total number of attempts after which we can get a maximal length CA. We conclude this section with some open problems related to maximal length CAs over GF($q$). Some of them are quite old and difficult. These problems are solved in GF(2), but remain unsolved over GF($q$). Through experiments, Problem~\ref{oproblem:existCAgfq} has been answered affirmatively. But, the proof is still not known.

%We randomly choose the rules from only the valid rules, then finding a maximal length CA of desired size and correspondingly its primitive polynomial. 
%Table~\ref{table:primitiveGF} (see Appendix, page~\pageref{table:primitiveGF}) records a demo of this process to find one maximal length CA of size 10, 15 and 20 over GF(2) to GF(19) in one random sample run. The $3^{rd}$ column records the total execution time in seconds and the last column shows the total number of attempts after which we can get a maximal length CA.

\begin{oproblem}
\label{oproblem:synthesisCAgfq}
Synthesize a CA from a given polynomial over GF($q$).
\end{oproblem}

\begin{oproblem}
\label{oproblem:noofCAgfq}
How many CAs are there for an irreducible polynomial over GF($q$)? 
\end{oproblem}

%\begin{oproblem}
%\label{oproblem:nlinearCAgfq}
%Existence of non-linear maximal length CA over GF($q$).
%\end{oproblem}

\begin{oproblem}
\label{oproblem:existCAgfq}
Does there exist a CA for every irreducible polynomial in GF($q$)?
\end{oproblem}

\begin{table}[t]
\begin{center}
	\caption{Finding one maximal length CA of a particular degree in GF($q$) in a sample run}		
		{\begin{tabular}{c||c|c|c}\hline
		Field & CA Size & Time (sec) & Actual number\\ 
		& & & of execution\\ \hline
		
		 & 10 & 3.697707 & 5 \\
		GF(2) & 15 & 6.421378 & 11 \\
		 & 20 & 21.839431  & 20 \\ \hline
		
		& 10 & 2.441540 & 3 \\
		GF(3) & 15 & 7.400775 & 6 \\
		 & 20 & 6.092859 & 5 \\ \hline
		 
		 & 10 & 2.461736  & 2 \\
		GF(5) & 15 & 3.680341 & 3 \\
		 & 20 & 18.409685 & 16 \\ \hline
		 
		 & 10 & 32.476975 & 31 \\
		GF(7) & 15 & 18.126331 & 15  \\
		 & 20 & 56.014876  & 55 \\ \hline
		
		 & 10 & 52.554778 & 47 \\
		GF(11) & 15 & 12.102284  & 9 \\
		 & 20 & 12.099299 &  9 \\ \hline
		 
		 & 10 & 3.577851 & 3 \\
		GF(13) & 15 & 11.738276  & 10 \\
		 & 20 & 16.560623 & 14 \\ \hline
		 
		 & 10 & 9.406888 & 7 \\
		GF(17) & 15 & 6.121513 & 5 \\
		 & 20 & 102.453698 & 92  \\ \hline
		 
		 & 10 & 3.397518 & 2 \\
		GF(19) & 15 & 84.848327 & 71 \\
		 & 20 & 34.926098  & 26 \\ \hline
\end{tabular}
		
		\label{table:primitiveGF}}
		\end{center}
		\end{table}

\section{Applications: PRNGs and Others}
\label{section:applications}

One advantage of the binary CAs is that they can easily be implemented in hardware. Figure~\ref{Figure:Hardware-ECA-Null} shows a schematic diagram of hardware implementation of an $n$-cell non-uniform CA under null boundary condition. Here, each cell consists of a flip-flop (FF) to store the state of a cell and a combinational logic circuit to find the next state of the cell. Due to similarity of this structure with finite state machine, many researchers, who have been working with finite non-uniform CAs, consider these CAs as finite state machine (FSM) and call the configurations of CAs as states.

\begin{figure}
\begin{center}
     \includegraphics[scale=0.55]{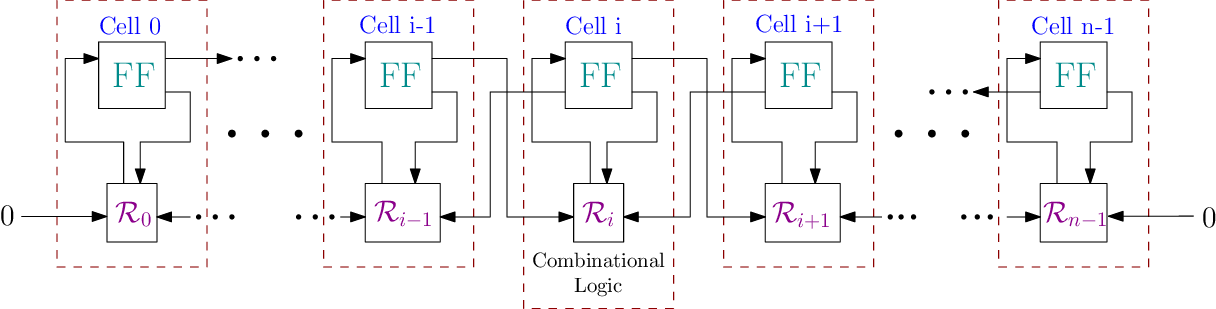}
     \caption{Implementation of an $n$-cell non-uniform CA under null boundary condition}
     \label{Figure:Hardware-ECA-Null}
\end{center}
\end{figure}

The CAs are useful in many application areas due to their simplicity, modularity and regularity \cite{comer2012random, formenti2014advances, leporati2014cryptographic, tomassini2000generation, sipper1996generating, wang2008generating, ganguly2002evolutionary, ats03, ats04d}. In the last several years, a group of researchers started characterizing maximal length CA based models in diverse fields including pseudo-random number generation \cite{Horte89a,Horte89b,Wolfram84b,Horte89c,Barde90,COMPAGNER1987391,Tezuka}, cryptography~\cite{nandi1994theory,CHANG97,ACRIDRC10,jcaDasR11}, signature analysis~\cite{Horte90b,Das90e}, error correcting codes~\cite{tsalides1990cellular,Chowd94a}, test-pattern generator~\cite{das1993vector} and etc. Next, we discuss some of the applications briefly.

\subsection{Pseudo-random Number Generators (PRNGs)}

In $1980$s, cellular automata (CAs) were introduced as an alternative source of randomness \cite{Wolfr85a}. CAs are simple, regular, locally interconnected, and modular. Further, using of maximal length CAs as the source of randomness offers some inherent advantages. Due to these properties, many researchers have used these CAs as their pseudo-random number generators \cite{Horte89a,Horte89c,Barde90,COMPAGNER1987391,Tezuka}.

%The description of 1-dimensional CA is very simple and capable of very wide range of global behaviour. Wolfram has classified four basic classes of behaviour in 1-dimensional CA \cite{Wolfram84b}. Class 1 automata evolve to homogeneous final global states, class 2 to periodic structure, class 3 exhibit chaotic behaviour, and class 4 yield complicated localized and propagating structures. The first two classification readily predictable and show no properties of interests for PRNG. The third class yields much more complex behaviour in that the detailed patterns can no longer be predicted (it may still be possible to make statements about global behaviour) and often seem random in nature. Wolfram consider class 3 CAs to be an abstract model of randomness in nature and so are suitable for pseudo-random generation \cite{Wolfr85a}. This is because the cumulative effect of many iterations of many class 3 CAs is equivalent to performing very complicated transformations on the initial starting value. This evolution often becomes so complicated that its outcome can be found only by observation or simulation. Further, rules 90 and 150 are belongs to class 3 which can compute maximal length CAs. 

We must first consider a suitable definition of randomness in order to properly evaluate any PRNG. The randomness tests of Knuth \cite{Knuth81} were used as the metric against which a PRNG can be measured for randomness. This metric describes the result of the average performance of sequences produced by a PRNG on all tests. The weighting allows some tests to have an increased importance. For example, if it were crucial that a pseudo-random sequence have equidistribution, then the weighting for the equidistribution test could be made large. Here we assign equal weight to all tests.

Hortensius et al.~\cite{Horte89a,Horte89c} added some extra points in maximal length CAs to generate pseudo-random number. They used the technique of ``site spacing'' between output sites, as in Figure~\ref{figure:sitespacing} with the maximal length CA.  They have defined a site spacing parameter, $\gamma$ , where the value of $\gamma$ is the number of sites between outputs in the CA. For example, in Figure~\ref{figure:sitespacing}, $\gamma=3$. If $\gamma$ is increased, the authors have expected that the cross correction between adjacent bits in the pseudo-random numbers is to be reduced. Also, one can find the experimental results of the pseudo-random number tests for various values of site spacing in Ref. \cite{Horte89c}. Maximal length CAs with $\gamma=0$ posses ``poor'' pseudo-random properties because of the distribution problems, but when single spacing site is introduced the maximal length CA sequences pass as many tests as the sequences produced by the algorithmic PRNG. Therefore, a site spacing of $\gamma=1$ will be considered adequate for pseudo-random sequence generation. These are very encouraging results as the application of site spacing may prove to be of academic interest only. At this point, the maximal length CA based PRNG is performing like a standard software based PRNG.

\begin{figure}
	\centering
	\includegraphics[width= 3.9in, height = 0.4in]{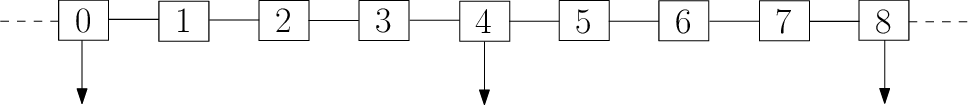}		
	\caption{Definition of site spacing : here $\gamma=3$}        		
\label{figure:sitespacing}        		
\end{figure}

%In Table~\ref{table:PRNGMLCA} (see Appendix, page~\pageref{table:PRNGMLCA}), the results of the pseudo-random number tests for various values of site spacing, $\gamma$ are provided. The 30-bit maximal length CA is taken from Table~\ref{Table:mlcaupto53} (see page~\pageref{Table:mlcaupto53}) for this experiment. Maximal length CAs with $\gamma=0$ posses ``poor'' pseudo-random properties because of the distribution problems, but when single spacing site is introduced the maximal length CA sequences pass as many tests as the sequences produced by the algorithmic PRNG of Table~\ref{table:keytests}. This behaviour extends over all the site-spaced maximal length CA PRNGs presented in Table~\ref{table:PRNGMLCA}. In addition, the other test metrics are well satisfied. Therefore, a site spacing of $\gamma=1$ will be considered adequate for pseudo-random sequence generation. These are very encouraging results as the application of site spacing may prove to be of academic interest only. At this point, the maximal length CA based PRNG is performing like a standard software based PRNG.

%\begin{figure}
%				\subfloat[$n$=24]{\includegraphics[width= 0.8in, height = 2.5in]{GF(2)-24_0_1.png}}\hspace*{4em}
%				\subfloat[$n$=15]{\includegraphics[width= 0.65in, height = 2.5in]{Gf(3)-15_0.png}} 		
%\caption{Space-time diagram of maximal length CAs. For each of the figures, the evolutions have started from configuration $0^{n-1}1$. Left one in GF(2) and right one in GF(3).}        		
%\label{figure:space-time}  
%%\vspace{-1.0em}      		
%\end{figure}

Usually, latest PRNG technique claims to be better from previous one. This claim is based on the PRNG's performance in some statistical tests, like Diehard~\cite{diehard}, TestU01~\cite{L'Ecuyer2007}, NIST~\cite{Soto99} etc. battery of tests, which empirically detect non-randomness in the generated numbers. In a PRNG, based on CA, the pseudo-random numbers are generated out of the configurations of a finite CA, so, to be unpredictable, the CA has to \emph{destroy} the information of the initial configuration, step by step, during its evolution \cite{Supreeti_2018_chaos,eisele1991long, Mitchell1993chaos, cattaneo1999dynamical, cervelle2009chaotic, acerbi2009conservation,Sukanya-IJMPC-2021}. In the space-time diagram of a CA, one can find this type of unpredictability. This unpredictability, if satisfied in a finite CA, has a correspondence to the independence property, which is essential for a good PRNG. Maximal length CAs have those properties including information flow on both sides (right and left). %It has been established that the patterns generated by maximal length CAs meet all the criteria and the quality of randomness of the patterns generated by CAs is good. %Also, we will discuss the topic in Chapter~\ref{chapter:PRNG}. 

It has been established that the pattern generated by maximal length CAs meet all the criteria and the quality of randomness of the patterns generated by CAs is good. In \cite{Adak-IJMPC2019}, the authors explores the randomness quality of maximal length CAs in GF($q$) and verified the randomness quality of these maximal length CAs as PRNGs by using Diehard battery of tests. They claim that, if we increase $q$ from $2$ to $11$, the randomness quality of the CAs for some specific size improve (see Figure~\ref{figure:comparison-chart}). %However, after GF(11), further increment in $q$ does not improve the randomness quality of the CAs. So, we have concluded that, GF(11) is the saturation point of randomness quality for these maximal length CAs (see Figure~\ref{figure:comparison-chart}).

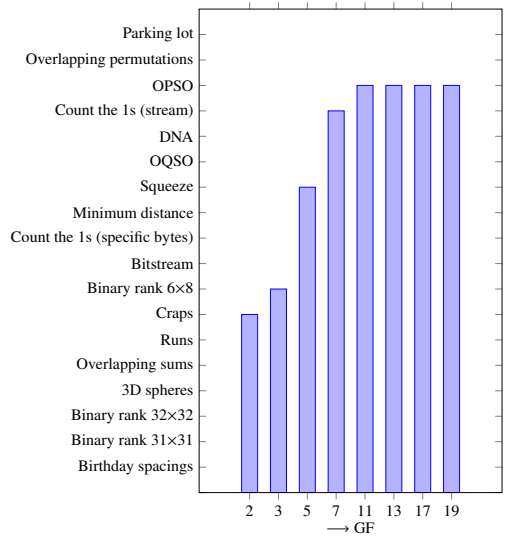
\begin{figure}[t]
\centering
\begin{tikzpicture}[scale=0.6]
\begin{axis}[
    every axis plot post/.style={/pgf/number format/fixed},
	xlabel={$\longrightarrow$ GF},
    ybar=100pt,
    bar width=10pt,
    x=18pt,
    y=16pt,
    ymin=0,
    axis on top,
    ymax=19,
    xtick=data,
    ytick={0,...,18},
    yticklabels={,Birthday spacings,Binary rank 31$\times$31,Binary rank 32$\times$32,3D spheres,Overlapping sums,Runs,Craps,Binary rank 6$\times$8,Bitstream,Count the 1s (specific bytes),Minimum distance,Squeeze,OQSO,DNA,Count the 1s (stream),OPSO,Overlapping permutations,Parking lot},
    enlarge x limits=0.25,
    symbolic x coords={2,3,5,7,11,13,17,19},
]
\addplot coordinates {(2,7) (3,8) (5,12) (7,15) (11,16) (13,16) (17,16) (19,16)};
\end{axis}
\end{tikzpicture}
\caption{Comparison chart of the maximal length CAs by using Diehard battery of tests}
\label{figure:comparison-chart}
\end{figure}

\subsection{Security Purpose}

{\em \textbf{Signature Analysis} :} Signature analysis is the most widely used data compression technique for test response evaluation. Built-in self test (BIST) design methodologies attempt to deal with the inherent complexity of testing digital VLSI circuits by incorporating both the test pattern generator and data compactor on-chip. Advantages of this approach over other VLSI test techniques include high speed, low pin overhead, and relatively low cost. In the traditional signature analysis method, they are well understood and are thoroughly explained in the algebraic coding theory literature as {\em syndrome detection}~\cite{Peterson1972,Lin1983} and in the digital testing literature as {\em signature analysis}~\cite{Fujiw85,Tsui1987}. But in~\cite{Horte90b}, Hortensius et al. discussed different signature analysis technique based on CAs. In this paper, they have examined in detail some of the signature analysis properties of CA. By experimentally, they proved that, the nearest neighbor communication properties is required for implementing elementary one-dimensional CA to allow the consideration of several different techniques of signature analysis. They have examined four signature analysis techniques for CAs. It was shown that, signature analysis methods using maximal length CA provide equivalent aliasing performance compared to others.

{\em \textbf{Cryptography} :} With this ever increasing growth of data communication, the need for security and privacy have become a basic necessity. Cryptography is an essential requirement for communication privacy or concealment of data. Encryption may be achieved by constructing two different types of ciphers - block ciphers and stream ciphers~\cite{Rueppel1986}. A block cipher is the one in which a massage is broken into successive blocks and they are encrypted using single key or multiple keys. On the other hand, in a stream cipher the message is broken into successive bits or characters and then the string of characters is encrypted using a key stream. The stream ciphers play an important role in cryptographic practices that protect communications in very high frequency domain. The central problem in stream cipher cryptography, however, is the difficulty of generating a long unpredictable sequence of binary signals from a short and random key. In~\cite{nandi1994theory}, they produce three basic requirements for cryptographically secure key stream generators which are as follows.

\begin{enumerate}
\item The period of the key stream must be large enough to accommodate the length of the transmitted message.
\item The output bits must be easy to generate.
\item The output bits must be hard to predict. That is, given a portion of the output sequence, the cryptanalyst should not be able to generate other bits forward or backward.
\end{enumerate}

Nandi et al.~\cite{nandi1994theory} have presented a CA based scheme for a class of block ciphers and stream ciphers. The scheme which has several attractive features use of maximal length CA. Chang et al.~\cite{CHANG97} have used, maximal length CAs for stream ciphers. Linear maximal length CAs are not secure, therefore, we need non-linear maximal length CAs. Das et al.~\cite{ACRIDRC10,jcaDasR11} put non-linearity in maximal length CAs which are cryptographically suitable. It shows that the bit streams generated in this manner are highly non-linear and pass all the statistical tests for randomness. These maximal length CAs can be used as a non-linear primitive in cryptographic applications.

\section{Discussion}
\label{section:discussion}

Maximal length CA has gone through a long journey from the early days to the modern trends of research. In this survey, we have presented known results over a variety of research using different characterization tools. Various milestone of development regarding maximal length CAs are briefly depicted. Although we have omitted some results regarding maximal length CAs, we have targeted to provide an overview of results that we consider important in the field of maximal length CA theory. We have provided several open problems which are important for maximal length CAs and primitive polynomials. Till date the complexity for deciding a maximal length CA is exponential. Further, there is no proof claimimg that it is an NP-complete or NP-hard problem. So, it is also an important open problem related to maximal length CAs.

\begin{oproblem}
\label{oproblem:P-NP}
Is the problem of deciding maximal length CA NP-complete or NP-hard?
\end{oproblem}

%We conclude this chapter with some open problems related to maximal length CAs. Some of them are quite old and difficult. All the problems are also solved in GF(2), but still unsolved over GF($q$). Through experiments, Problem~\ref{oproblem:existCAgfq} has been answered affirmatively. But, the proof is still not known. 

%\begin{oproblem}
%\label{oproblem:synthesisCAgfq}
%Synthesize a CA from a given polynomial over GF($q$).
%\end{oproblem}
%
%\begin{oproblem}
%\label{oproblem:noofCAgfq}
%How many CAs are there for an irreducible polynomial over GF($q$)? 
%\end{oproblem}
%
%\begin{oproblem}
%\label{oproblem:nlinearCAgfq}
%Existence of non-linear maximal length CA over GF($q$).
%\end{oproblem}
%
%
%\begin{oproblem}
%\label{oproblem:existCAgfq}
%Does there exist a CA for every irreducible polynomial in GF($q$)?
%\end{oproblem}

%\begin{oproblem}
%\label{oproblem:synthesisCAgfq}
%The relation between the characteristic polynomials of conjugate CAs.
%\end{oproblem}

% \appendix
% \section{Example Appendix Section}

\bibliographystyle{elsarticle-harv} 
\bibliography{References}



\end{document}